\documentclass[journal,10pt,draftclsnofoot,onecolumn]{IEEEtran} 
\pdfoutput=1 

\newtheorem{theorem}{Theorem}[section]
\newtheorem{lemma}[theorem]{Lemma}
\newtheorem{definition}[theorem]{Definition}
\newtheorem{corollary}[theorem]{Corollary}
\newtheorem{remark}{Remark}
\newtheorem{example}{Example}
\newtheorem{assumption}{Assumption}
\newtheorem{proposition}{Proposition}

\newcommand{\uX}{\underline{X}}
\newcommand{\ux}{\underline{x}}

\newcommand{\wuthetan}{\widehat{\underline{\theta}}_n}

\newcommand{\utheta}{\underline{\theta}}
\newcommand{\wcalI}{\widehat{\calI}}

\newcommand{\bitm}{\begin{itemize}}
\newcommand{\eitm}{\end{itemize}}
\newcommand{\beqa}{\begin{eqnarray}}
\newcommand{\eeqa}{\end{eqnarray}}
\newcommand{\beqas}{\begin{eqnarray*}}
\newcommand{\eeqas}{\end{eqnarray*}}
\newcommand{\baln}{\begin{align}}
\newcommand{\ealn}{\end{align}}
\newcommand{\balns}{\begin{align*}}
\newcommand{\ealns}{\end{align*}}

\newcommand{\probSimplex}[1] {\mathcal{P}\parenth{#1}}

\newcommand{\kldist}[2] {\mathrm{D}\! \parenth{#1\|#2}}
\newcommand{\E}{\mathbb{E}}
\renewcommand*{\P} {\mathbb{P}}

\newcommand{\prob}[1] {\P\parenth{#1}}

\newcommand{\alphabet}[1] { {\mathsf #1}}

\newcommand{\R}{\mathbb{R}}

\newcommand{\reals}{\R}

\newcommand{\parenth}[1] {\left(#1\right)}

\newcommand{\brackets}[1] {\left[#1\right]}

\newcommand{\abs}[1] {\left|#1\right|}

\newcommand{\pmf}[2] {P_{#1}\!\parenth{#2}}

\newcommand{\calI}{\mathcal{I}}

\newcommand{\I}{\mathrm{I}}

\renewcommand*{\H}{\text{H}}

\newcommand{\X}{\bX}
\newcommand{\x}{\mathbf{x}}

\newcommand{\Y}{\mathbf{Y}}
\newcommand{\Z}{\mathbf{Z}}
\newcommand{\N}{\mathbf{N}}

\def\argmin{\mathop{\arg\,\!\min}\limits}%
\def\argmax{\mathop{\arg\,\!\max}\limits}%

\newcommand{\calX}{\alphabet{X}}

\newcommand{\bW}{\mathbf{W}}

\newcommand{\bX}{\mathbf{X}}
\newcommand{\bY}{\mathbf{Y}}
\newcommand{\bZ}{\mathbf{Z}}

\newcommand{\bx}{\mathbf{x}}

\newcommand{\bz}{\mathbf{z}}
\newcommand{\allX}{\underline{\bX}}

\newcommand{\allx}{\underline{\bx}}
\newcommand{\buz}{\underline{\bz}}

\newcommand{\buZ}{\underline{\mathbf{Z}}}

\newcommand{\uV}{\underline{\mathrm{V}}}

\newcommand{\A}{\mathrm{A}}

\newcommand{\calS}{\mathcal{S}}
\newcommand{\calB}{\mathcal{B}}

\newcommand{\calP}{\mathcal{P}}

\newcommand{\setmi}[2]{  [#1]\backslash\{#2\}}

\newcommand{\Phat}{\widehat{P}}
\newcommand{\Ahat}{\widehat{A}}
\newcommand{\calPhat}{\widehat{\mathcal{P}}}
\newcommand{\uAhat}{\widehat{A}}
\newcommand{\Phatrob}{\Phat_{\mathrm{rob}}}
\newcommand{\calIhat}{\widehat{\mathcal{I}}}

\newcommand{\std}{\mathrm{std}} 

\newcommand{\ph}{\phantom{0}}
\usepackage{cite} 
\usepackage{mathtools}
\usepackage{verbatim}
\usepackage{amsmath,amssymb,amsfonts}
\usepackage{algorithmic}
\usepackage{array}
\usepackage[tight,footnotesize]{subfigure}

\usepackage{graphicx}
\usepackage{overpic}
\usepackage{subfigure}
\usepackage{bm}
\usepackage{bbm}
\usepackage{url}

\makeatletter
\newlength \figwidth
\if@twocolumn
\else
\fi
\makeatother

\begin{document}
\title{Directed Information Graphs}

\author{Christopher~J.~Quinn, 
        Negar~Kiyavash,~\IEEEmembership{Senior~Member,~IEEE,}\\
        and~Todd~P.~Coleman,~\IEEEmembership{Senior~Member,~IEEE}
\thanks{This material was presented in part at Int. Symp. Inf. Theory 2011 \cite{quinn2011equivalence}, NetSciCom 2011 \cite{quinn2011generalized}, Conf. Dec. and Control 2011, and Int. Symp. Inf. Theory 2013. 
C.~J.~Quinn was supported by the Department of Energy Computational Science Graduate Fellowship, which is provided under Grant DE-FG02-97ER25308.  This work was supported by AFOSR Grant FA9550-11-1-0016, MURI under AFOSR Grant FA9550-10-1-0573, NSF Grant CCF-1065352, and NSF Grant CCF-0939370.
}
\thanks{C.~J.~Quinn performed this work at the Department of Electrical and Computer Engineering, Coordinated Science Laboratory, University of Illinois, Urbana, Illinois 61801.  He is now with the  School of Industrial Engineering at Purdue University, West Lafayette, Indiana 47907 
        {\tt\small cjquinn@purdue.edu}}%
\thanks{N. Kiyavash is with the Department of Industrial and Enterprise Systems Engineering, Coordinated Science Laboratory, University of Illinois, Urbana, Illinois 61801
        {\tt\small kiyavash@illinois.edu}}%
\thanks{T.~P.~Coleman is with the Department of Bioengineering, University of California, San Diego, La Jolla, CA 92093
        {\tt\small tpcoleman@ucsd.edu}}%
}

\maketitle

\begin{abstract}

We propose a graphical model for representing networks of stochastic processes, the minimal generative model graph.  It is based on reduced factorizations of the joint distribution over time.  We show that under appropriate conditions, it is unique and consistent with another type of graphical model, the directed information graph, which is based on a generalization of Granger causality.  We demonstrate how directed information quantifies Granger causality in a particular sequential prediction setting.  We also develop efficient methods to estimate the topological structure from data that obviate estimating the joint statistics.  One algorithm assumes upper-bounds on the degrees and uses the minimal dimension statistics necessary. In the event that the upper-bounds are not valid, the resulting graph is nonetheless an optimal approximation.  Another algorithm uses near-minimal dimension statistics when no bounds are known but the distribution satisfies a certain criterion. Analogous to how structure learning algorithms for undirected graphical models use mutual information estimates, these algorithms use directed information estimates.  We characterize the sample-complexity of two plug-in directed information estimators and obtain confidence intervals.  For the setting when point estimates are unreliable, we propose an algorithm that uses confidence intervals to identify the best approximation that is robust to estimation error.   Lastly, we demonstrate the effectiveness of the proposed algorithms through analysis of both synthetic data and real data from the Twitter network.  In the latter case, we identify which news sources influence users in the network by merely analyzing tweet times.
\end{abstract}



%
\IEEEpeerreviewmaketitle


\section{Introduction}
%
%
%
%
%
%
\IEEEPARstart{R}{esearch} in many disciplines, including biology, economics, social sciences, computer science, and physics, involves large networks of interacting agents.  For instance, neuroscientists seek to determine which neurons communicate with which other neurons.  Investors want to learn which stocks' fluctuations effect their portfolios.  Computer security experts seek to uncover which computers in a network infected others with malicious software.  Often, researchers can observe time-series of agents' activity, such as neural spikes, stock prices, and network traffic. This work develops tools that analyze network time-series  to identify the underlying causal influences between agents.

A natural question is whether influences between agents can be learned by analyzing their activity.  Suppose an advertiser wants to target a specific population of users in a micro-blogging network, such as Twitter.  The advertiser sees that many users follow several major news companies and celebrities   (see Figure~\ref{fig:intro:social_network}) and wants to identify which sources have strong influence on the users to decide whom to pay to advertise. The advertiser can observe time-series of activities, such as message times (see Figure~\ref{fig:intro:tweet_times}), and wants to calculate a measure of influence using the data.  Clearly, efficent algorithms and reliable methods to compute statistics from data are desired. Moreover,  direct and indirect influences must be distinguished.  That is, if  both celebrities $A$ and $B$ influence user $Y$, the advertiser must be sure both are direct influences, and not that $A$ influences $B$ who in turn influences $Y$.

\begin{figure}[t]
\centering
\includegraphics[width=\figwidth]{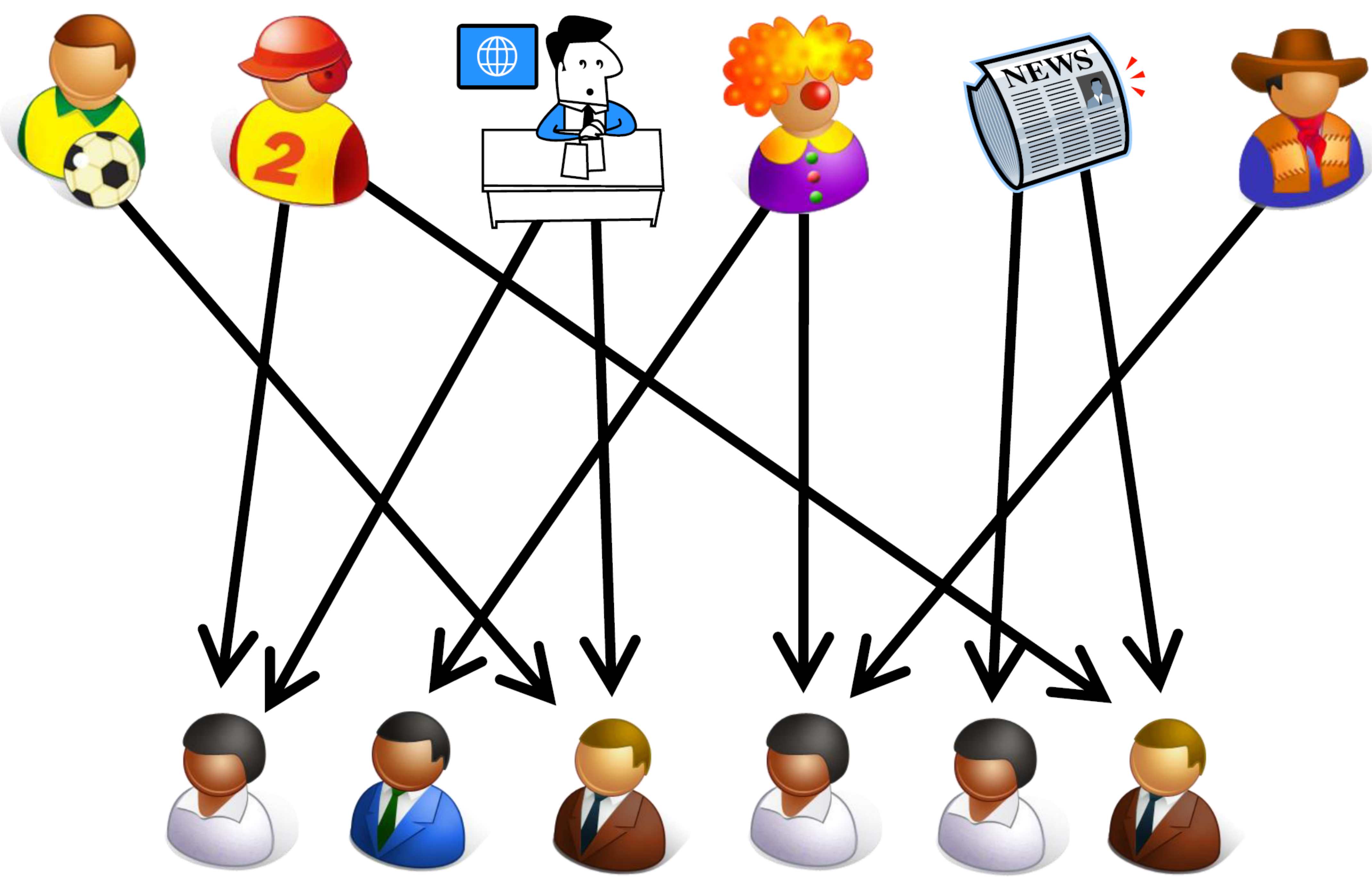}
\caption{This graph depicts the influences between celebrities and news corporations (top) to a population of users (bottom) in an example of an online social network.  For applications such as word-of-mouth advertisement, it is more useful to know the graph of influences between agents than the ``friend'' or ``follower'' graphs.  However, influences are harder to identify and must be inferred from agent activity.} \label{fig:intro:social_network}
\end{figure}

\begin{figure}[t]
\centering
\includegraphics[width=\figwidth]{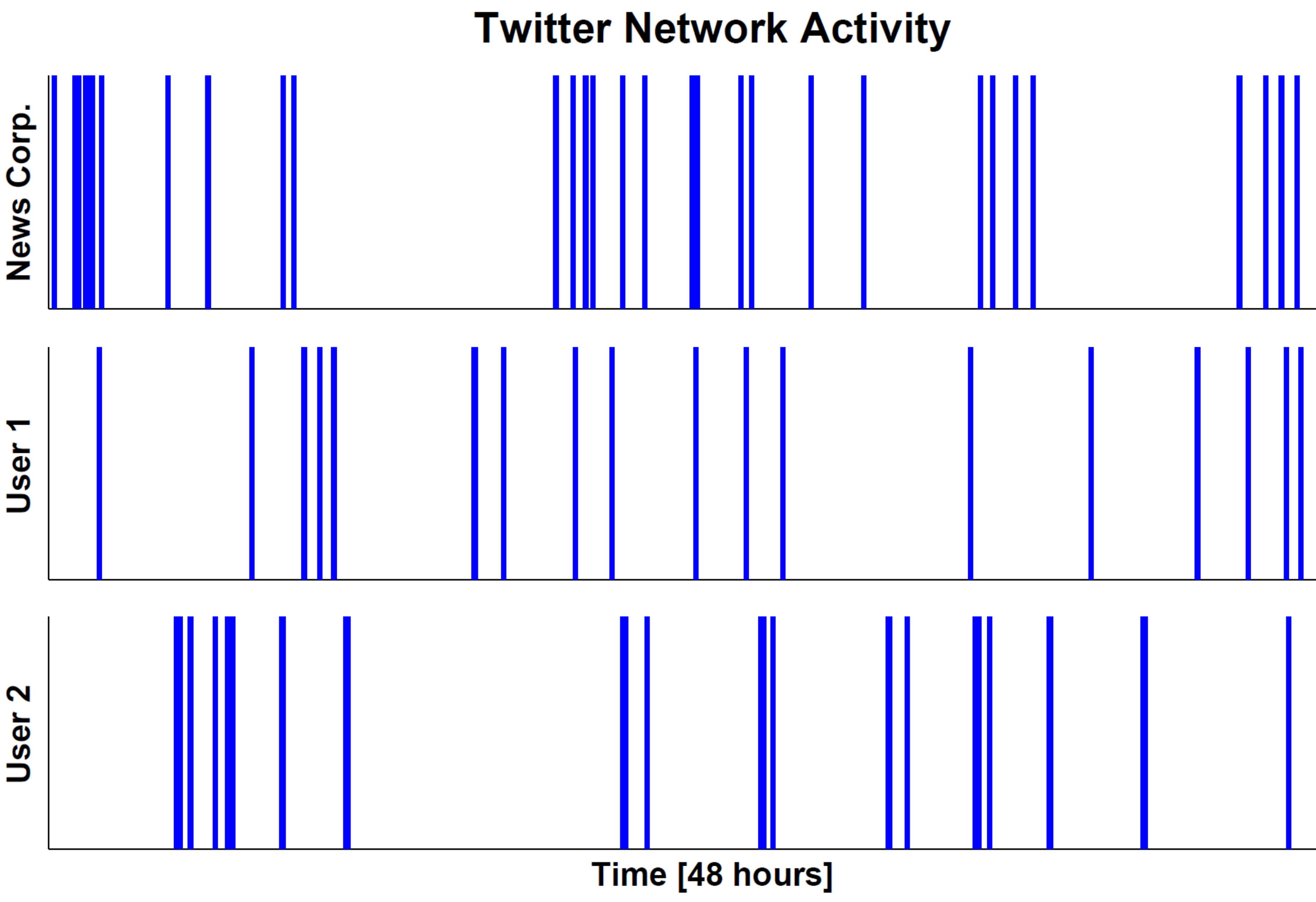}
\caption{This plot shows example microblogging activity of a news corporation and two users over two days in an online social network. Vertical lines depict each time a message was posted by that agent.  A major research goal is to infer whether, and how strongly, the news corporation influences the users by analyzing these time-series.} \label{fig:intro:tweet_times}
\end{figure}

\subsection{Our Contribution}

This work develops methods to address these issues. It proposes a graphical model, the minimal generative model graph, based on reduced factorizations of the joint distribution over time.  Agents are depicted by nodes and directed edges represent inferred influences.  We show how under certain conditions, it is both unique and consistent with another graphical model, the directed information graph \cite{ quinn2011estimating, amblard2011directed}, that is motivated by Granger causality \cite{granger1963economic, granger1969investigating}.  The latter uses the information-theoretic quantity directed information \cite{marko1973bidirectional}, which generalizes the concept of Granger causality.  Clive Granger, a Nobel laureate, proposed a methodology for deciding when, in a statistical sense, one process $\X$ causally influences another process $\Y$ in a network \cite{granger1963economic, granger1969investigating}.  Directed information has been used before to measure Granger causality.   Here we clarify how directed information connects with Granger's original principle beyond agreeing with ``strong'' Granger causality \cite{chamberlain1982general, florens1982note} which uses conditional independence tests.

For networks with large numbers of nodes, such as millions of people in a social network, having efficient algorithms that recover the graphical models is critical.  We propose and prove the correctness of algorithms to infer the graph.  Two algorithms are described without assumptions on the network.  Another, more efficient, algorithm is discussed that recovers the graph when  upper bounds on the in-degrees are known.  We show this algorithm returns an optimal approximation when the bounds are not valid.  We also present a modified version of the algorithm that returns a bounded-degree approximation that is robust to estimation errors.  We prove the correctness of an adaptive algorithm that we proposed in an earlier work \cite{quinn2011estimating}.  We also show that even if the assumption of no instantaneous influence is invalid, our algorithms will still recover the strictly causal influences.

Statistics must often be estimated from data. We identify sample complexity---how much data is needed for reliable estimation---and confidence intervals for plug-in empirical and parametric estimators of directed information.  

Although the proposed framework is theoretically-grounded, we demonstrate its utility by identifying influences in the Twitter network.  Specifically, we record and analyze activity of 16 news corporations and 48 user accounts to infer which  corporations influenced which users in tweeting about events in the Middle East.  Using only knowledge of the message times, not the follower graph or message content, the algorithms accurately infer many influences.

\subsection{Organization}
The paper is organized as follows. We next discuss related work.   
In Section~\ref{sec:background}, we establish definitions and notations.  In Section~\ref{sec:MGM_DI}, we discuss the graphical models.  In Section~\ref{sec:ident_struct}, we propose algorithms that identify the graph when no assumptions about the topology are made.  In Section~\ref{sec:ident_str_side_info}, we describe one algorithm that uses knowledge of in-degree bounds to more efficiently identify the exact graph and another algorithm that finds approximations which are robust to estimation error.  In Section~\ref{sec:estimation}, we evaluate the sample complexity of two plug-in estimators for directed information.   In Section~\ref{sec:sim}, we demonstrate the algorithms using simulations.  In Section~\ref{sec:analysis_twitter} we analyze social network data from Twitter.  The proofs are contained in the appendices.

\subsection{Related Work}

We next discuss related work on directed information and graphical models for networks of processes.

\subsubsection{Directed information}
Directed information was introduced by Marko \cite{marko1973bidirectional} and independently rediscovered by Kamitake et al. \cite{kamitake1984time} and Rissanen and Wax \cite{rissanen1987measures}.  Rissanen and Wax proposed their work as an extension of Granger's framework \cite{granger1969investigating}.  Massey modified Marko's work for the setting of communication channels \cite{massey1990causality}. Directed information has been used in a variety of settings, e.g., to characterize the capacity of channels with feedback \cite{kramer1998directed, kramer2003capacity, tatikonda2009capacity, permuter2009finite, naiss2013extension}, to quantify achievable rates for source encoding with noiseless feed-forward \cite{venkataramanan2007source, naiss2013computable}, and for feedback control \cite{elia2004when,  martins2008feedback, liu2009convergence, gorantla2012interplay}.   Permuter et al. explored its relevance to gambling, hypothesis testing, and portfolio theory \cite{permuter2009interpretations}.      

Applications of directed information include neuroscience studies \cite{quinn2011estimating,kim2011Granger, liu2010information, so2012assessing}, analysis of gene regulatory data \cite{rao2008using}, and social networks \cite{versteeg2012information, versteeg2013information}.  We note that \cite{versteeg2012information} estimated pairwise directed information only for known follower--followee pairs, and \cite{versteeg2013information} used message content.   In our application, we use causally conditioned directed information and do not use prior knowledge of follower relationships or message content.

There have been several works on estimating directed information and many more on entropy and mutual information.  A parametric estimator for directed information was independently proposed by Quinn et al. \cite{quinn2011estimating} and Kim et al. \cite{kim2011Granger}.   Consistency was shown in \cite{quinn2011estimating}. Jiao et al. developed a consistent, universal estimation scheme for directed information in the finite-alphabet setting using context-tree weighting and showed sample complexity results \cite{jiao2012universal}.  Frenzel and Pompe \cite{frenzel2007partial} adapted a k-nearest neighbors mutual information estimator from \cite{kraskov2004estimating}. Data partitioning based methods were investigated in \cite{rao2007motif, rao2008using, liu2009high}.  Wu et al. \cite{wu2012learning} identified sample complexity of mutual information for i.i.d. data.  Our proof for the plug-in empirical estimator has an analogous structure.

\subsubsection{Graphical models for networks of processes}

There is a rich body of literature on graphical models for i.i.d. random variables, such as Markov and Bayesian networks \cite{koller2009probabilistic}.  Dynamic Bayesian networks extend Bayesian networks to the setting of processes \cite{murphy2002dynamic}, representing each variable in each process as a separate node.  

This work follows another approach, developing graphical models where each node represents a whole process.  There have been several works proposing such graphical models based on Granger causality.  Dalhaus \cite{dahlhaus2000graphical} and Eichler \cite{eichler2007granger} developed graphs for autoregressive time-series.  Using conditional independence tests proposed as ``strong'' Granger causality in \cite{florens1982note} and \cite{chamberlain1982general}, Eichler \cite{eichler2012graphical} identified what conditional independencies must hold for a joint distribution $P_{\allX}$ to ``satisfy'' a given graph $G$.  For instance, if there is no edge $\X \to \Y$ in $G$, then $\Y$ should be causally conditionally independent of $\X$ given the rest of the network.  That is equivalent to a directed information being zero.  Eichler \cite{eichler2012graphical} also identified sufficient conditions when pairwise independencies imply a (non-minimal) generative model.  Our preliminary work \cite{quinn2011equivalence} independently proved that result; here we cite \cite{eichler2012graphical} to simplify the proofs.  As this work considers producing a particular graph $G$ that best represents $P_{\allX}$, we  address additional issues such as minimality, uniqueness, algorithms, estimators, and graph approximation.

Although \cite{eichler2012graphical} did not explicitly define a graphical model such as directed information graphs, the pairwise independence conditions it studied can be naturally extended to do so.  Our earlier work \cite{quinn2011estimating} and Amblard and Michel's \cite{amblard2011directed} both proposed directed information graphs, independent of Eichler \cite{eichler2012graphical} and each other.  Recently, \cite{amblard2014causal} explored how instantaneous influences can effect directed information graphs \cite{eichler2012graphical}.

\subsubsection{Structure learning for graphical models for networks of processes}

We first note that there is a large body of literature on exact and heuristic structure learning for graphical models of i.i.d. random variables.  Some comprehensive references are \cite{koller2009probabilistic, koski2012review,koller2009probabilistic}.  Some of the algorithms we develop are analogous, such as testing each edge using a directed information where for Markov networks mutual information would be used.  A recent work by Wu et al. \cite{wu2012learning} proposed a consistent algorithm to identify Markov networks with bounded degree.  It used more numerous and complicated tests than our algorithm for the analogous problem, essentially searching for subsets of parents of both $X$ and $Y$ to conclude whether the edge $X\!\!-\!\! Y$ exists.  

Several works have investigated structure learning specifically for multivariate autoregressive models.  Materassi and Innocenti \cite{materassi2010topological} proposed an algorithm analogous to Chow and Liu's \cite{chow1968approximating}, for the case when the network is a tree.  Materassi and Salapaka \cite{materassi2012problem} extended \cite{materassi2010topological} to larger classes of topologies.  Tan and Willsky also investigated learning tree structured networks \cite{tan2011sample}. An alternative approach identified sparse networks used group Lasso \cite{bolstad2011causal}.  

In \cite{quinn2013efficient}, we proposed an algorithm to approximate the network with a tree topology, analogous to Chow and Liu's work in \cite{chow1968approximating}, but it was not restricted to autoregressive processes as \cite{materassi2010topological, materassi2012problem, tan2011sample, bolstad2011causal} were.  Lastly, we proposed an adaptive structure learning algorithm for directed information graphs in \cite{quinn2011estimating}; here we prove correctness.

\section{Background} \label{sec:background}

\subsection{Notation and Information-Theoretic Definitions}
\begin{itemize}
\item For a sequence $a_1,a_2,\ldots$, denote $(a_i,\ldots,a_j)$ as $a_i^j$ and $a^k := a_1^k$.

\item Denote $[m] := \{1,\ldots,m\}$ and the power set $2^{[m]}$ on $[m]$ to be the set of all subsets of $[m]$.  

\item For any finite alphabet $\calX$, denote the space of probability measures on $\calX$ as $\probSimplex{\calX}$.

\item Throughout this paper, we will consider $m$ finite-alphabet, discrete-time random processes over a horizon $n$.  Let $\calX$ denote the alphabet. Denote the $i$th random variable at time $t$ by $X_{i,t}$, 
the $i$th random process as $\bX_i = (X_{i,1},\ldots,X_{i,n})^{\top}$, the whole collection of all $m$ random processes as $\allX= (\bX_1,\ldots,\bX_m)^{\top}$, and a subset of $K$ processes indexed by $A \subseteq [m]$ as $\allX_{A}= (\bX_{A(1)},\ldots,\bX_{A(K)})^{\top}$.

\begin{remark} We consider the finite-alphabet setting to simplify the presentation.  The results easily extend to more general cases.
\end{remark}

\item Denote the conditional distribution and {\it causally conditioned} distribution \cite{kramer1998directed} of $\X_i$ given $\X_j$ as
\beqa
\!\!\!\!P_{\X_i|\X_j=\x_j}(\x_i) \!\!\!\!&:=&\!\!\!\! \pmf{\X_i|\X_j}{\x_i|\x_j} \nonumber \\
 \!&=&\!\!\!\! \prod_{t=1}^n \pmf{X_{i,t}|X_{i}^{t-1},X_j^n}{x_{i,t}|x_{i}^{t-1},x_j^n}  \label{eq:def:cond_distr}\\
\!\!\!\!P_{\X_i\|\X_j=\x_j}(\x_i) \!\!\!\!&:=&\!\!\!\! \pmf{\X_i\|\X_j}{\x_i\|\x_j} \nonumber \\
\!&:=&\!\!\!\! \prod_{t=1}^n \pmf{X_{i,t}|X_{i}^{t-1},X_j^{t-1}}{x_{i,t}|x_{i}^{t-1},x_j^{t-1}}\!. \label{eq:def:causal_cond}
\eeqa
Note the similarity between \eqref{eq:def:cond_distr} and \eqref{eq:def:causal_cond}, though in \eqref{eq:def:causal_cond} the present and future, $x^n_{j,t}$, is not conditioned on.  In \cite{kramer1998directed}, the present $x_{j,t}$ was conditioned on.  The reason we remove it will be made clear in Remark~\ref{rmrk:cond_fut}.

\item  Consider the set of processes $\allX_A$ for some $A \subseteq \setmi{m}{i}$.  Next consider two sets of causally conditioned distributions
$\{P_{\X_i\| \allX_A=\allx_A} \in \probSimplex{\calX}: \allx_A \in \calX^{|A|n}\}$ and 
$\{Q_{\X_i\| \allX_A=\allx_A} \in \probSimplex{\calX}: \allx_A \in \calX^{|A|n}\}$ along with
a marginal distribution $P_{\allX_A} \in \probSimplex{\calX^{|A|n}}$.  Then the conditional Kullback-Leibler (KL) divergence between causally conditioned distributions is given by
\begin{align}  
\kldist{P_{\X_i\| \allX_A}}{Q_{\X_i\| \allX_A} | P_{\allX_A}} := \sum_{t=1}^n \sum_{\ux_A^{t-1} 
} \!\kldist{P_{X_{i,t}| \uX_A^{t-1}=\ux_A^{t-1}}}{Q_{X_{i,t}| \uX_A^{t-1}=\ux_A^{t-1}}}  P_{\uX_A^{t-1}}(\ux_A^{t-1}). \label{eq:def:condKL}
\end{align}
The following lemma will be useful throughout:
\begin{lemma} [\!\cite{cover2006elements} pg.~29] \label{lemma:iffKLdivergenceZero}
$\kldist{P_{\X_i\| \allX_A}}{Q_{\X_i\| \allX_A} | P_{\allX_A}}= 0$ if and only if $P_{\X_i\| \allX_A =\allx_A}(\x_i) = Q_{\X_i\| \allX_A =\allx_A}(\x_i)$
for all $\allx_A$ such that $P_{\allX_A}(\allx_A)>0$.
\end{lemma}

\item
Let $i,j \in [m]$ and $A \subseteq \setmi{m}{i,j}$.  The mutual information,  {\it directed information} \cite{marko1973bidirectional}, and causally conditioned directed information \cite{kramer1998directed} are given by
\begin{align}
\I(\X_j;\X_i) :=& \ \kldist{P_{\X_i,\X_j}}{P_{\X_i} P_{\X_j}} = \kldist{P_{\X_i|\X_j}}{P_{\X_i} | P_{\X_j}} \label{eqn:defn:MutualInformation} \\
=& \sum_{t=1}^n \I(X_j^n; X_{i,t} | X_i^{t-1}) \nonumber \\
\I(\X_j \to \X_i) :=& \  \kldist{P_{\X_i\|\X_j}}{P_{\X_i} | P_{\X_j}}\nonumber 
\\
=& \sum_{t=1}^n \I(X_j^{t-1}; X_{i,t} | X_i^{t-1}) \nonumber\\
\I(\X_j \to \X_i \| \allX_A) :=&  \  \kldist{P_{\X_i\| \allX_{A \cup \{j\}}}}{P_{\X_i\|\allX_A} | P_{\allX_{A \cup \{j\}} }} \label{eqn:defn:ccDirectedInformation} 
\\
=& \sum_{t=1}^n \I(X_j^{t-1}; X_{i,t} | X_i^{t-1},\uX_A^{t-1}). \nonumber
\end{align}
Mutual information and directed information are related.  However, while mutual information quantifies statistical correlation (in the colloquial sense of statistical interdependence), directed information quantifies statistical \emph{causation}.  We later justify this statement showing that directed information is a general formulation of Granger causality.  Note that
$\I(\X_j;\X_i)=\I(\X_i;\X_j)$, but $\I(\X_j \to \X_i) \neq \I(\X_i \to \X_j)$ in general.

\begin{remark} \label{rmrk:cond_fut}
In \eqref{eq:def:causal_cond} and \eqref{eqn:defn:ccDirectedInformation}, there is no conditioning on the present $X_{j,t}$.  This follows Marko's definition \cite{marko1973bidirectional} and is consistent with Granger causality \cite{granger1969investigating}.  Massey \cite{massey1990causality} and Kramer \cite{kramer1998directed} later included conditioning on $X_{j,t}$ for the specific setting of communication channels.  Although there is a small delay between the channel input and output, since the causation is \emph{already known} for that setting, it is notationally convenient to use synchronized time.  That necessitates conditioning on the present $X_{j,t}$.
\end{remark}

We have the following corollary.  \begin{corollary} For $i,j \in [m]$ and $A \subseteq [m] \backslash \{i,j\}$, 
$\I(\X_j \to \X_i \| \allX_A) = 0$ if and only if \[ P_{\X_i\|\X_j=\x_j,\allX_A=\allx_A}(\x_i) = P_{\X_i\|\allX_A=\allx_A}(\x_i), \quad \forall \ \allx_{\{j\} \cup A} \in \calX^{(|A|+1)n} :  P_{\allX_{\{j\} \cup A}}(\allx_{\{j\} \cup A})>0\]  
\end{corollary} %
\begin{proof}
The proof follows immediately from Lemma~\ref{lemma:iffKLdivergenceZero} and \eqref{eqn:defn:ccDirectedInformation}.
\end{proof}

In this case, we say $\X_i$ is causally independent of $\X_j$  causally conditioned on $\allX_A$.  We also denote $\I(\X_j \to \X_i \| \allX_A) = 0$ as $\X_j \to \allX_A \to \X_i$, a causal (in the sense of Kramer's causal conditioning \eqref{eq:def:causal_cond}) Markov chain.  This entails that $X_{i,t}$ is independent of $X_j^{t-1}$ given $\uX_A^{t-1}$ and $X_i^{t-1}$.  

\item Let $G = (V,E)$ denote a directed graph.  For each edge $(u,v)\in E$, $u$ is called the \emph{parent} and $v$ is the \emph{child}.  

\end{itemize}

\section{ Minimal Generative Models and Directed Information Graphs} \label{sec:MGM_DI}

We now consider the problem of graphically representing causal relationships between stochastic processes in a network. We will examine two definitions of ``causal'' in this section.  Both are based on observed time-series.

\subsection{Minimal Generative Models} \label{subsec:MGM}

A deterministic dynamical system is characterized by a set of differential or difference equations.  Those equations describe how the past state of the system influences how the state will evolve. For stochastic dynamical systems, the induced joint distribution factorizes in an analogous manner. 

Let $\allX$ be a network of $m$ random processes with a joint distribution $P_{\allX}$.  The system dynamics are fully described by $P_{\allX}$.  First factorize $P_{\allX}$ over time,
$P_{\allX}(\allx)  =  \prod_{t=1}^n P_{ \allX_{t}   |  \allX^{t-1}  } (  \allx_{t} |  \allx^{t-1}  ). $ If the processes at time $t$ are independent given the full past $\allX^{t-1}$,
\begin{eqnarray}
P_{\allX}(\allx) \!\!\!&=&\!\!\!  \prod_{t=1}^n  \prod_{i=1}^m P_{X_{i,t}| \allX^{t-1}} (x_{i,t}| \allx^{t-1} ),  \label{eq:MGM:2}
\end{eqnarray} we say $P_{\allX}$ is \emph{strictly causal}.  Using causal conditioning notation \eqref{eq:def:causal_cond}, \eqref{eq:MGM:2} can be written as 
\begin{eqnarray}
P_{\allX}(\allx) \!\!&=&\!\!  \prod_{i=1}^m  P_{\bX_i \parallel \allX_{\setmi{m}{i}} } (\bx_i \parallel \allx_{\setmi{m}{i}}). \label{eq:MGM:strc_pos}
\end{eqnarray} By factorizing over time first, each $\bX_i$ is still conditioned on the full past of every other process. 

\begin{assumption}\label{def:spat_cond_ind}
For the remainder of this paper, we only consider joint distributions that are positive, i.e. $P_{\allX}(\allx) > 0$ for all $\allx \in \calX^{mn}$, and strictly causal, satisfying \eqref{eq:MGM:strc_pos}.
\end{assumption}

\begin{remark}
The positivity assumption avoids degenerate cases that arise with purely deterministic relationships between variables.  Granger argued that strict causality is a valid assumption if the sampling rate is high enough and relevant processes are observed \cite{granger1969investigating, granger1980testing}. Strict causality will be essential to obtain a {\em unique}, minimal factorization.  Furthermore, we show in Appendix~\ref{sec:indent_struct:info_decom} that even if $P_{\allX}$ is not strictly causal, our results will apply to the strictly causal part of $P_{\allX}$, i.e. the right hand side of \eqref{eq:MGM:strc_pos}. 
\end{remark}

The following lemma shows that a large class of generative models are strictly causal.
\begin{lemma}
If a set of stochastic processes $\allX$ has a generative model of the form $X_{i,t} = g_{i,t}(\allX^{t-1}, N_{i,t}), $where $\{ g_{i,t} \}_{i\in [m], 1\leq t \leq n}$ are deterministic functions and the random variables $ \{X_{i,0}\}_{i\in[m]}  \cup \{N_{i,t}\}_{i\in [m], 1\leq t \leq n} $ are mutually independent, then $P_{\allX}$ is strictly causal.
\end{lemma}
\begin{IEEEproof}
Let $i \in [m]$ and $B\subseteq [m]\backslash\{i\}$ be arbitrary.  Then 
\beqa
\I(X_{i,t} ; \allX_{B, t} | \allX^{t-1}) &=& \I( g_{i,t}(\allX^{t-1}, N_{i,t}) ; \{  g_{j,t}(\allX^{t-1}, N_{j,t}) \}_{j \in B}    | \allX^{t-1}) \nonumber \\ 
&\leq& \I(  N_{i,t} ; \{  N_{j,t} \}_{j \in B}    | \allX^{t-1}) , \label{eq:strcaus:case:2} \\
&=& \I(  N_{i,t} ; \{  N_{j,t} \}_{j \in B}    ) , \label{eq:strcaus:case:3} \\
&=& 0 , \label{eq:strcaus:case:4} 
\eeqa %
where \eqref{eq:strcaus:case:2} follows from the data processing inequality and conditioning on $\allX^{t-1}$, \eqref{eq:strcaus:case:3} follows because $\allX^{t-1}$ is a function of past noises, $ \{X_{l,0}\}_{l\in[m]}  \cup \{N_{l,t'}\}_{l\in [m], 1\leq t' \leq t-1}$, which are independent of noises at time $t$, $\{N_{l,t}\}_{l\in [m]}$, and \eqref{eq:strcaus:case:4} follows because $i \not \in B$ so  $N_{i,t}$ is independent of $\{  N_{j,t} \}_{j \in B}$. This result implies \eqref{eq:MGM:2} holds.
\end{IEEEproof}

In \eqref{eq:MGM:strc_pos}, there could be unnecessary dependencies.  We next remove those.  For each process $\bX_i$, let $A(i) \subseteq \setmi{m}{i}$ denote a subset of other processes.  Define the corresponding induced probability distribution $P_A$, \begin{eqnarray}
P_A(\allx) = \prod_{i=1}^m  P_{\bX_i \parallel \allX_{A(i)}}(\bx_i \parallel \allx_{A(i)}). \label{eq:MGM:PA}
\end{eqnarray}  We want to pick the parent sets $\{A(i)\}_{i=1}^m$ with minimal cardinalities that preserve the full dynamics of $P_{\allX}$, \beqa \kldist{P_{\allX}}{P_A } =0. \label{eq:MGM:6}
\eeqa

\begin{definition}
\label{def:min_gen_mod}
For a joint distribution $P_{\allX}$, a \emph{minimal generative model} is a function $A:[m]\to 2^{[m]}$ where the cardinalities of the parent sets $\{|A(i)|\}_{i=1}^m$ are minimal such that \eqref{eq:MGM:6} holds.
\end{definition}

By non-negativity of the KL divergence and the factorizations \eqref{eq:MGM:strc_pos} and \eqref{eq:MGM:PA}, \eqref{eq:MGM:6} corresponds to \begin{eqnarray}
\kldist{P_{\bX_i \parallel \allX_{\setmi{m}{i}} }}{P_{\bX_i \parallel \allX_{A(i)} } \big| P_{\allX_{\setmi{m}{i}} }} \!\!\!&=&\!\!\!0 \label{eq:MGM:a1}
\end{eqnarray} for all $i \in [m]$.  Thus, the parent sets can be chosen separately to satisfy \eqref{eq:MGM:6}. 

\begin{remark}
Minimal generative models are well defined for any $P_{\allX}$. An explicit (non-)parametric generative model could be unknown or not exist.
\end{remark}

We now define a corresponding graphical model.
\begin{definition}\label{def:MGMgrph} A \emph{minimal generative model graph} is a directed graph for a minimal generative model $A$, where each process is represented by a node, and there is a directed edge from $\bX_k$ to $\bX_i$ for $i,k \in [m]$ iff $k \in A(i)$.
\end{definition}

The goal is to produce a single graph for the network structure, but in general,  models might not exist or be unique.  Appendix~\ref{sec:indent_struct:info_decom} shows that if $P_{\allX}$ is not strictly causal, no strictly causal $P_A$ satisfies \eqref{eq:MGM:6}.  The following example shows that if positivity is violated, the  model graph need not be unique.
\begin{example}\label{exmpl:nonpos} Let $\N$ and $\N'$ be mutually independent processes with i.i.d. standard normal variables.  Let $\X$, $\Y$, and $\Z$ be three processes with $X_t = N_t$, $Z_t = X_{t-1}$, and $Y_t =Z_{t-1}+N_t'.$  It is natural to posit that the minimal generative model graph {\em should} have edges $\{\X \to \Z, \Z \to \Y\}$.  However, there are two equally valid minimal generative model graphs, one with edges $\{\X \to \Z, \X \to \Y\}$ and one with edges $\{\X \to \Z, \Z \to \Y\}$. It is ambiguous which to use. 
\end{example}
We will show in Theorem~\ref{thm:MGMandDIsame} that under Assumption~\ref{def:spat_cond_ind},  models graphs exist and are unique.

We next discuss an alternative graphical model from \cite{quinn2011estimating, amblard2011directed}, which is based on the framework of Granger causality and directly identifies relationships between pairs of processes. It always exists and is unique.

\subsection{Granger Causality and Directed Information Graphs} \label{sec:GrangerDI}

An alternative approach to defining causal influences is the widely adopted Granger causality.  In the 1960s, motivated by earlier work by Wiener \cite{wiener1956theory}, Nobel laureate Clive Granger proposed \cite{granger1963economic, granger1969investigating}: ``We say that $\bX$ is causing $\bY$ if we are better able to predict [the future of] $\bY$ using all available information than if the information apart from [the past of] $\bX$ had been used.'' 
Granger's original formulation involved statistical hypothesis testing with linear models.  For the setting  of two processes, later works used directed information (in value, not in name) \cite{rissanen1987measures, bouissou1986tests, gourieroux1987kullback}. 
Directed information is equal to linear Granger causality for jointly Gaussian processes\cite{barnett2009granger, amblard2009relating}.  

Causally conditioned directed information has recently been used in the setting of networks of processes \cite{quinn2011estimating, amblard2011directed}.  Recall that Granger's statement was in terms of the value of causal side information in sequential prediction.  The next proposition shows directed information is that value for a specific yet flexible sequential prediction problem.

\begin{proposition} \label{prop:DIisGranger}
The directed information $\I(\X_j \to \X_i \| \allX_{[m]\backslash \{j\}})$ is precisely the value of the side information $X_j^{t-1}$ in terms of expected cumulative reduction in loss when sequentially predicting $X_{i,t}$ with knowledge $\allX_{[m]\backslash \{j\}}^{t-1}$, the predictors are distributions with minimal expected loss, and log loss is used.
\end{proposition}
For the formulation and proof see Appendix~\ref{app:prop:DIisGranger}.

We thus use directed information to determine the causal influences in a network in the sense of Granger.  

\begin{definition}[\!\!\cite{quinn2011estimating, amblard2011directed}]\label{def:dir_info_grph} For a set of random processes $\allX$, the \emph{directed information graph} is a directed graph where each node represents a process and there is a directed edge from  $\X_j$ to  $\X_i$ (for $i,j \in [m]$) iff
\begin{eqnarray}
 \I(\X_j \to \X_i \parallel \allX_{ \setmi{m}{i,j}}) > 0. \label{eq:def:dir_inf_grph1}
\end{eqnarray}
\end{definition}

\begin{corollary} \label{cor:DIunique}
For any $P_{\allX}$, the directed information graph exists and is unique.
\end{corollary}
\begin{IEEEproof}
The proof follows immediately as \eqref{eq:def:dir_inf_grph1} is true or false for each edge.
\end{IEEEproof}

Corollary~\ref{cor:DIunique} is good in that directed information graphs always depict a single topology. However, in some cases the graph can be misleading.  Consider the network from Example~\ref{exmpl:nonpos}.  Although the directed information graph exists and is unique, it will only have the edge $\X \to \Z$. Thus, it will appear as $\Y$ was independent which is arguably worse than the ambiguity of the two minimal generative model graphs from Example~\ref{exmpl:nonpos} with edge sets $\{\X \to \Z, \X \to \Y\}$ and $\{\X \to \Z, \Z \to \Y\}$ respectively.

In general the two graphical models can disagree.  We next show that under appropriate conditions, not only is there a unique minimal generative model graph, but it is consistent with the directed information graph.

\begin{theorem} \label{thm:MGMandDIsame}
If $P_{\allX}$ satisfies Assumption~\ref{def:spat_cond_ind}, there is a unique minimal generative model graph and it is equivalent to the directed information graph.  
\end{theorem}

The proof is in Appendix~\ref{app:prf:MGMandDIsame}.  In the remainder it will become clear that it is convenient having two characterizations,  the edge test \eqref{eq:def:dir_inf_grph1} and the factorization \eqref{eq:MGM:PA}. We next consider how to efficiently find the graphs.

\section{ Graphical Model Identification -- General}\label{sec:ident_struct}
In this section, we discuss algorithms to identify the graph when no assumptions about the topology are made.  Efficiency will correspond to the dimension of the statistics that will be necessary, such as only needing joint statistics of pairs of processes as compared to the full joint distribution.  By Theorem~\ref{thm:MGMandDIsame}, the algorithms can learn the network by either testing edges \eqref{eq:def:dir_inf_grph1} or searching for parent sets \eqref{eq:MGM:PA}. 

First consider the following lemma which will motivate later algorithms and help prove their correctness.  
\begin{lemma}
\label{lem:caus_data_proc}  Let $P_{\allX}$ satisfy Assumption~\ref{def:spat_cond_ind}.  For any process $\X_i$, consider any set $B(i)$ containing the parent set $A(i)$, $A(i) \subseteq B(i) \subseteq \setmi{m}{i}$, and any subset $W(i)\subseteq\setmi{m}{i}$.  Then $\allX_{W(i)} \to \allX_{B(i)} \to \bX_i$ holds, meaning $\I(\allX_{W(i)} \to \bX_i \| \allX_{B(i)}) = 0.$ Furthermore, \beqa \I( \allX_{W(i)}  \to \bX_i) \leq \I( \allX_{B(i)} \to \bX_i), \label{eq:lem:data_proc:gen} \eeqa with equality iff $ A(i) \subseteq W(i)$.  
\end{lemma}

The proof is in Appendix~\ref{app:prf:caus_Mark_chain}.  Equation~\eqref{eq:lem:data_proc:gen} implies that any set $W(i)$ containing the full parent set captures the full dynamics, and even if $W(i)$ is only missing a single parent, it will not accurately describe $\X_i$'s dynamics. We now consider algorithms to learn the topology.

\begin{table}[t] \begin{normalsize} \begin{center}
\begin{tabular*}{\linewidth}{@{}llrr@{}}
{\bfseries Algorithm 1. MGMconstruct}\\
\hline {\bf Input:} $\mathcal{DI}_{\mathrm{MGM}}, \ m$ \\
\hline
~1. {\bf For} $i \in [m]$  \vspace{-0.15cm}\\ \vspace{-0.15cm}
~2. $\quad$ $A(i) \gets \setmi{m}{i}$ \\ \vspace{-0.15cm}
~3. $\quad $ {\bf For} $j \in A(i) $ \\ \vspace{-0.15cm}
~4. $\quad \quad$ $ B(i) \gets A(i) \backslash \{j\}$ \\ \vspace{-0.15cm}
~5. $\quad \quad$ {\bf If} $\I( \bX_{j} \to \bX_i  \parallel \allX_{ B(i)}) = 0$ \\ \vspace{-0.15cm}
~6. $\quad \quad \quad$ \( A(i) \gets B(i)\) \\
~7. {\bf Return} $\{A(i)\}_{i=1}^m$\\
\hline
\end{tabular*}\end{center}\label{alg:MGMconstruct2} \end{normalsize}
\vspace{-1cm}
\end{table}

\subsection{Algorithm~1 --- Parent Set Search}

Identifying the parent sets of each process $\bX_i$ requires determining the minimal cardinality  set $A(i)$ that satisfies \eqref{eq:MGM:a1}.  No search order is prespecified.  One approach to finding the smallest $A(i)$ is to test increasing sizes of subsets of potential parents.  This would require calculating an exponential number of causally conditioned directed informations \eqref{eqn:defn:ccDirectedInformation}. An alternative method, motivated by Markov chains and Lemma~\ref{lem:caus_data_proc}, is to start with all other processes as a trivial Markov blanket and sequentially test and remove insignificant ones as described in Algorithm~1.

Let $\mathcal{DI}_{\mathrm{MGM}}$ denote input to Algorithm~1---the set of all causally conditioned directed information values from one process to another, causally conditioned on a subset of the rest,
\begin{eqnarray}  \mathcal{DI}_{\mathrm{MGM}} \!\!\!\!\!\!&=&\!\!\!\!\!\! \left\{\! \I\!\left( \! \bX_{j} \!\! \to\! \bX_i \! \parallel\! \allX_{ B(i)}\!\right) \!\!:\! j,i\!\in\!\! [m], \! B(i) \!\subseteq \!\setmi{m}{i,j}\!   \right \} \!. \nonumber
\end{eqnarray}

\begin{theorem}
\label{thm:alg1}  If $P_{\allX}$ satisfies Assumption~\ref{def:spat_cond_ind}, Algorithm~1 recovers the minimal generative model.
\end{theorem}

The proof is in Appendix~\ref{prf:thm:alg1}.  Algorithm~1 requires the full joint distribution.  However, it only uses $\mathcal{O}(m^2)$ tests.  Note that the tests used in line~5 are adaptive, using the current $B(i)$.  Next consider an alternative algorithm following the definition of directed information graphs (Definition~\ref{def:dir_info_grph}).

\subsection{Algorithm~2 --- Edge Tests}

Directed information graphs are defined by separate edge tests.  Testing an edge entails computing a directed information from one process to another, causally conditioned on all other processes. This is described in Algorithm~2.  Let $\mathcal{DI}_{\mathrm{DI}}$ denote the input to Algorithm~2---the set of causally conditioned directed informations \begin{eqnarray} \!\!\!\!\! \mathcal{DI}_{\mathrm{DI}} \!\!\!
&=&\!\!\! \{ \I( \bX_j \to \bX_i  \parallel \allX_{ \setmi{m}{i,j}}) : i,j\in [m] \} . \nonumber
\end{eqnarray}

\begin{table}[t] \begin{normalsize} \begin{center}
\begin{tabular*}{\linewidth}{@{}llrr@{}}
{\bfseries Algorithm 2. DIconstruct}\\
\hline {\bf Input:} $\mathcal{DI}_{\mathrm{DI}}, \ m$ \\
\hline
~1. {\bf For} $i \in [m]$  \vspace{-0.15cm}\\\vspace{-0.15cm}
~2. $\quad$ \( A(i) \gets \emptyset\) \\\vspace{-0.15cm}
~3. {\bf For} $i,j \in [m]$  \\\vspace{-0.15cm}
~4. $\quad$ {\bf If} \(\I( \bX_j \to \bX_i  \parallel \allX_{ \setmi{m}{i,j}})>0\) \\\vspace{-0.15cm}
~5. $\quad \quad$ \( A(i) \gets A(i) \cup \{j\} \) \\ 
~6. {\bf Return} $\{A(i)\}_{i=1}^m$\\
\hline
\end{tabular*}\end{center}\label{alg:DIconstruct} \end{normalsize}
\vspace{-1cm}
\end{table}

\begin{theorem}
\label{thm:alg2}  If $P_{\allX}$ satisfies Assumption~\ref{def:spat_cond_ind}, Algorithm~2 recovers the directed information graph.
\end{theorem}
\begin{proof} The proof follows immediately from Definition~\ref{def:dir_info_grph}.\end{proof}

Unlike Algorithm~1, Algorithm~2 uses each of the $\mathcal{O}(m^2)$ elements in $\mathcal{DI}_{\mathrm{DI}}$.  Line~4 could be executed in parallel.  The number of directed information tests is the same as Algorithm~1, though the tests themselves are different.

\subsection{Algorithm~3 --- Adaptive}
Even when the graph is sparse, Algorithms~1~and~2  use high-dimensional statistics.  A natural question is whether for sparse graphs low-dimensional statistics, such as only between pairs of processes, is sufficient to recover the graph.  Example~\ref{ex:noisyXOR} in Appendix~\ref{ex:alg3:xor} shows that is not true.  It demonstrates that if  $K$ is the size of the largest parent set, then any algorithm that uses directed informations involving $K$ or fewer processes cannot guarantee recovery in general.  We next investigate an algorithm that can recover the graph using directed informations of $K+2$ processes.

In \cite{quinn2011estimating}, Quinn et. al. proposed Algorithm 3.  We include it here for completeness.  It is an adaptive algorithm for networks with unknown in-degrees.  It identifies parents by first using pairwise tests, then conditioning on one process, then two, etc.  Thus, if the processes have small in-degrees, then low-dimensional statistics suffice to learn the structure.  However, the following assumption is required.

\begin{assumption}\label{assump:faithful1}
For a distribution $P_{\allX}$, for all $i,j \in [m]$ and $S  \subseteq [m] \backslash \{i,j \}$, 
\beqas
\I(\X_j \to \X_i \| \allX_{[m] \backslash \{i,j \}}) > 0 \implies \I(\X_j \to \X_i \| \allX_{S})>0. \label{eq:faithfulness}
\eeqas
\end{assumption}

\begin{remark} 
Assumption 2 excludes certain non-linear relationships such as the exclusive-or (see Appendix~\ref{ex:alg3:xor}) and relationships with ``perfect'' cancellation.  For a linear counter-example, let $\mathbf{N}_1$, $\mathbf{N}_2$, and $\mathbf{N}_3$ be independent processes with i.i.d. $\mathcal{N}(0,1)$ variables.  Let $X_t=N_{1,t}$, $Y_t=-X_{t-1} + N_{2,t}$, and $Z_t=X_{t-2}+Y_{t-1}+N_{3,t} = N_{2,t-1} + N_{3,t}$.  Then $\I(\X \to \bZ\|\bY)>0$ but $\I(\X \to \Z)=0$.  
\end{remark}

\begin{table}[t] \begin{normalsize} \begin{center}
\begin{tabular*}{\linewidth}{@{}llrr@{}}
{\bfseries Algorithm 3. GenStructAdapt} \cite{quinn2011estimating}\\
\hline {\bf Input:} $\mathcal{DI}_{\mathrm{MGM}}, \ m$ \\
\hline

~\phantom{1}1. {\bf For} $i \in [m]$ \vspace{-0.15cm} \\  \vspace{-0.15cm}
~\phantom{1}2. $\quad  A(i) \gets [m]\backslash \{ i\}$ \\\vspace{-0.15cm}
~\phantom{1}3. $\quad K \gets 0$ \\\vspace{-0.15cm}
~\phantom{1}4. $\quad \!\!$ {\bf While} $K +1 \leq |A(i)|$ \\\vspace{-0.15cm}
~\phantom{1}5. $\quad \quad$ {\bf For} \(j \in A(i)\) \\\vspace{-0.15cm}
~\phantom{1}6. $\quad \quad \quad$ \( \mathcal{B} \gets \{ B : B \subseteq A(i) \backslash \{j\}, \ |B| = K  \}\) \\ \vspace{-0.15cm}
~\phantom{1}7. $\quad \quad \quad$ {\bf For} \(B \in  \mathcal{B}\) \\ \vspace{-0.15cm}
~\phantom{1}8. $\quad \quad \quad \quad$ {\bf If} \( \I(\X_j \to \X_i \|   \allX_{B} ) =0 \) \\ \vspace{-0.15cm}
~\phantom{1}9. $\quad \quad \quad \quad \quad$  \(A(i) \gets A(i)\backslash \{j \}\) \\ \vspace{-0.15cm}
~10. $\quad \quad \quad \quad \quad$ {\bf Go to} line~5\\ \vspace{-0.15cm} 
~11. $\quad \quad $  \(K \gets K+1\) \\
~12. {\bf Return} $\{A(i)\}_{i=1}^m$\\ 

\hline
\end{tabular*}\end{center}\label{alg:StructRecov} \end{normalsize}
\vspace{-1cm}
\end{table}

\begin{theorem}
\label{thm:alg3}  If $P_{\allX}$ satisfies Assumptions~\ref{def:spat_cond_ind}~and~\ref{assump:faithful1}, Algorithm~3 recovers the directed information graph.
\end{theorem}

The proof is in Appendix~\ref{prf:thm:alg3}.  Note no proof of correctness was given in \cite{quinn2011estimating}.   For each process $\X_i$, $K$ will increment until it is at most the size of the parent set.  Thus, even though in-degrees are not known, the graph can be recovered using near minimal-dimensional statistics.

\section{Graphical Model Identification -- Constrained Topology}\label{sec:ident_str_side_info}

We now discuss an algorithm that identifies the graph with minimal dimensional statistics when   bounds on the in-degrees are known. We show even if the bounds are invalid, the resulting graph is an optimal approximation.

\subsection{Algorithm~4 --- Bounded In-Degree}

We next show that when there are known upper bounds $\{ K(i)\}_{i=1}^m$ on the in-degrees, directed informations involving $(K(i)\!+\!1)$ processes are sufficient to identify $\X_i$'s parents.  

Lemma~\ref{lem:caus_data_proc} showed that among all sets of $K(i)$ processes, those that contain $\X_i$'s parents have maximal influence on $\X_i$.  The intersection of those sets is $\X_i$'s parent set.  Algorithm~4 formally describes this.  Let $\mathcal{DI}_{\mathrm{BndInd}}$ denote the input to Algorithm~4---the set of directed information values from each $K(i)$-sized subset of processes to $\bX_i$,  
\begin{eqnarray} \!\! \mathcal{DI}_{\mathrm{BndInd}} \!\!\!
&=&\!\!\! \left\{ \I( \allX_{ B(i)} \to \bX_i  ) : i\in [m],  B(i) \subseteq \setmi{m}{i}, |B(i)|=K(i)   \right\} . \nonumber
\end{eqnarray}

\begin{table}[t] \begin{normalsize} \begin{center}
\begin{tabular*}{\linewidth}{@{}llrr@{}}
{\bfseries Algorithm 4. BoundedIn-Degree}\\
\hline {\bf Input:} $\mathcal{DI}_{\mathrm{BndInd}}, K, \ m$ \\
\hline

~1. {\bf For} $i \in [m]$  \vspace{-0.2cm} \\  \vspace{-0.15cm}
~2. $\quad\  A(i) \gets \emptyset$ \\  \vspace{-0.15cm}
~3. \(\quad \ \calB \gets \{ B':B'\subseteq \setmi{m}{i}, \ |B'| = K(i)  \}\) \\  \vspace{-0.1cm}
~4. $\quad $  \( \calB_{\mathrm{max}} \gets  \argmax_{B \in \calB} \ \I( \allX_{B} \to \bX_i)\) \\  \vspace{-0.1cm}
~5. $\quad $  \( A(i)  \gets \hspace{-0.25cm} \bigcap \limits_{B \in \calB_{\mathrm{max}}} \hspace{-0.35cm} B \) \\
~6. {\bf Return} $\{A(i)\}_{i=1}^m$\\
\hline
\end{tabular*}\end{center}\label{alg:BndInDgr} \end{normalsize}\vspace{-1cm}
\end{table}

Let $A^*(i)$ and $A(i)$ denote the true and returned parent sets for $\X_i$ respectively.  Note that for the trivial in-degree bound $K(i)=m-1$, $|\mathcal{B}|=1$ and the algorithm cannot resolve the parent sets.

\begin{theorem}
\label{thm:genstructmain} Let $P_{\allX}$ satisfy Assumption~\ref{def:spat_cond_ind}.
Algorithm~4 recovers the directed information graph for a given \(P_{\allX}\) if for each $i \in [m]$, $|A^*(i)| \leq K(i) \leq m-2$.
\end{theorem}

The proof is in Appendix~\ref{prf:thm:alg4}.

\begin{remark} In practice, one should test the output $\{A(i)\}_{i=1}^m$ of Algorithm~4, checking that \[\I( \allX_{A(i)} \to \bX_i) = \argmax_{B \in \calB} \ \I( \allX_{B} \to \bX_i).\]   
\end{remark}

\subsection{Bounded In-Degree Approximations} \label{sec:BndIndegrAppx}

Algorithm~4 requires bounds on the in-degrees.  A natural question is whether the output is useful if the bounds are invalid.  We next show that by modifying line~5 in Algorithm~4 to return any set $B \in \calB_{\mathrm{max}}$, the result is an optimal approximation, regardless of the validity of the bounds.

Consider approximating $P_{\allX}$ with 
\beqa
\Phat_{\allX}( \allx ) := \prod_{i=1}^m P_{\bX_{i}\parallel\bX_{\uAhat(i)}} (\bx_{i}\parallel \bx_{\uAhat(i)}). \label{eq:def:appxP}
\eeqa  In \eqref{eq:def:appxP}, the conditional marginals are exact, but the parent sets $\{\uAhat(i)\}_{i=1}^m$ are approximate.  The divergence $\kldist{P_{\allX}}{\Phat_{\allX}} $ measures how close $\Phat_{\allX}$ is to $P_{\allX}$.  Let $\calPhat_K$ denote the set of all approximations of the form \eqref{eq:def:appxP} with parent set cardinalities $|\uAhat(i)| = K(i)$.  Denote any optimal approximation as 
$\Phat_{\allX}^* := \argmin_{\Phat_{\allX} \in \calPhat_K} \kldist{P_{\allX}}{\Phat_{\allX}}.$

For the setting where $\Phat_{\allX}$ is constrained to be a directed tree, Quinn et al. \cite{quinn2013efficient} show that $\Phat_{\allX}^*$ is the directed tree with the maximal sum of directed informations along its edges.  We show an analogous result here, where $\Phat_{\allX}$ only has specified in-degrees.  Let $\{\Ahat^*(i)\}_{i=1}^m$ denote the parent sets corresponding to the optimal approximation $\Phat_{\allX}^*$.  The following theorem states that the $\Ahat^*(i)$ can be selected independently.

\begin{theorem}
\label{thm:apx:gen_K_apx} Let $P_{\allX}$ satisfy Assumption~\ref{def:spat_cond_ind}. For all $i \in [m]$,
\beqa
\Ahat^*(i) \in \argmax_{\uAhat(i) : |\uAhat(i)|=K(i) } \!\!\!\!\! \I ( \X_{\uAhat(i)} \to  \X_{i} ). \label{eq:apx:gen_K:thm}
\eeqa
\end{theorem}
\begin{IEEEproof}
The proof follows from \cite{quinn2013efficient}, which proved Theorem~\ref{thm:apx:gen_K_apx} for the special case $K=1$.  That proof naturally extends to the general case $K>1$.  
\end{IEEEproof}

\begin{corollary}\label{cor:sumDIweight} Under Assumption~\ref{def:spat_cond_ind}, $\sum_{i=1}^m\I ( \X_{\uAhat^*(i)} \!\!\to\!\!  \X_{i} ) = \max_{\{\uAhat(i) : |\uAhat(i)|=K(i)\}_{i=1}^m }\sum_{i=1}^m\I ( \X_{\A(i)} \!\!\to\!\!  \X_{i} )$
\end{corollary}
\begin{IEEEproof}
The proof follows immediately from Theorem~\ref{thm:apx:gen_K_apx}.
\end{IEEEproof}

By Theorem~\ref{thm:apx:gen_K_apx}, for any user-specified parent set cardinalities $\{K(i)\}_{i=1}^m$, Algorithm~4 can return an optimal approximation.  However, so far the algorithms have used exact directed information values or point estimates.  We next consider using confidence intervals. 

\subsection{Robust Graph Identification} \label{sec:rob_str_ident}

Algorithms~1-4 require calculations or point estimates of directed information to recover the graph.  When point estimates $\{\widehat{\I}(\allX_{A(i)} \to \X_i)\}_{i \in [m], A \subseteq [m] \backslash\{i\}}$ are not available or reliable, confidence intervals $\{\wcalI(\allX_{A(i)} \to \X_i)\}_{i \in [m], A \subseteq [m]\backslash\{i\}}$ can be used instead.  We next develop an algorithm that will find the ``best'' approximation which is robust to estimation error.  We will discuss estimation and confidence intervals more fully in Section~\ref{sec:estimation}.  

\begin{remark}
In this work, we consider the practical setting of using a set of simultaneous confidence intervals resulting in a rectangular, joint confidence region In general, multidimensional confidence regions need not be rectangular.  We also use a constant in-degree $K$ for notational simplicity; the results generalize.
\end{remark}

Denote the Cartesian product of confidence intervals as 
$\calS := \bigtimes_{ (i, A(i))} \wcalI(\allX_{A(i)} \to \X_i). \nonumber \label{eq:def:calS}$
Note that $\calS$ is a subset of $\mathbb{R}^{m {m-1 \choose K}}$.  Each element $s  \in \calS$ is a length $m {m-1 \choose K}$ vector. We refer to each $s \in \calS$ as a \emph{scenario}.  Each scenario selects an estimate value $\widehat{\I}_s(\allX_{A(i)} \to \X_i) \in \wcalI(\allX_{A(i)} \to \X_i)$ for all pairs $(i,A(i))$.

We now modify the optimal approximation version of Algorithm~4  in Section~\ref{sec:BndIndegrAppx}, to select approximate parent sets so that the approximation will be robust to estimation errors. For a given scenario $s\in \calS$ and approximation $\Phat_{\allX}$ \eqref{eq:def:appxP}, by Corollary~\ref{cor:sumDIweight} the quality of the approximation is characterized by its \emph{weight}, 
$W(\Phat_{\allX},s) := \sum_{i =1}^m \widehat{\I}_s(\allX_{\widehat{A}(i)} \to \X_{i}).\nonumber $
~The best approximation for a particular scenario $s$ is
\beqa
\Phat^*_{\allX}(s) := \argmax_{\Phat_{\allX} \in \calPhat_K} W(\Phat_{\allX},s). \label{eq:def:given_s_opt_Phat}
\eeqa
 
While Algorithm~4 can solve \eqref{eq:def:given_s_opt_Phat} to give the best parent sets for a given $s$, those parents might perform poorly in a different scenario $s' \in \calS$ compared to $\Phat^*_{\allX}(s')$.  A natural question is whether there is a  $\Phat_{\allX}$ that performs well under all scenarios.  In particular, we want to select the ``robust'' approximation $\Phatrob$ that attains the minimax regret,
\beqas
\Phatrob :=  \hspace{0.1cm} \argmin_{\Phat_{\allX} \in \calPhat_K} \hspace{0.1cm} \max_{s \in \calS}  \hspace{0.1cm}  \{ W(\Phat^*_{\allX}(s),s) - W(\Phat_{\allX},s) \}. \label{eq:def:Phatrob}
\eeqas

\begin{table}[t] \begin{normalsize} \begin{center}
\begin{tabular*}{\linewidth}{@{}llrr@{}}
{\bfseries Algorithm 5. RobustBoundedIn-Degree}\\
\hline {\bf Input:} $\calS, K, \ m$ \\
\hline

~1. \ph {\bf For} $i \in [m]$  \vspace{-0.15cm} \\ \vspace{-0.15cm}
~2. \ph $\quad\ \Ahat(i) \gets \emptyset$ \\ \vspace{-0.15cm}
~3. \ph \(\quad \ \calB \gets \{ B':B'\subseteq \setmi{m}{i}, \ |B'| = K  \}\) \\ \vspace{-0.15cm}
~4. \ph $\quad$ {\bf For} $j \in \{1, \cdots, |\calB|\}$ \\ \vspace{-0.15cm}
~5. \ph \(\quad \ \quad \ B_j \gets \calB(j) \) \\ \vspace{-0.15cm}
~6. \ph \(\quad \ \quad \ M(B_j) \gets \mathrm{midpoint}(\wcalI(\allX_{B_j} \to \X_i)) \) \\  \vspace{-0.15cm}
~7. \ph \(\quad \ \quad \ H(B_j) \gets \mathrm{max}(\wcalI(\allX_{B_j} \to \X_i)) \) \\  \vspace{-0.15cm}
~8. \ph \(\quad \ \quad \ L(B_j) \gets \mathrm{min}(\wcalI(\allX_{B_j} \to \X_i)) \) \\ \vspace{-0.15cm}
~9. \ph  \(  \quad \ j_1 \gets  \argmax\limits_j H(B_j)\) \\ \vspace{-0.15cm}
~10. \(\quad \ j_2 \gets  \argmax\limits_j L(B_j)\) \\ \vspace{-0.15cm}
~11. \(\quad \ j_3 \gets  \argmax\limits_{j\neq j_1} H(B_j)\) \\ \vspace{-0.1cm}
~12. $\quad$ {\bf If} $M(B_{j_1}) \geq \frac{1}{2} \left( H(B_{j_3}) + L(B_{j_2}) \right)$ \\ \vspace{-0.15cm}
~13. \(\quad \ \quad \ \Ahat(i) \gets B_{j_1} \) \\ \vspace{-0.15cm}
~14. $\quad$ {\bf Else } \\ \vspace{-0.1cm}
~15. \(\quad \ \quad \ \Ahat(i) \gets B_{j_2} \) \\
~16. {\bf Return} $\{\Ahat(i)\}_{i=1}^m$\\
\hline
\end{tabular*}\end{center}\label{alg:RobustBndIndDeg} \end{normalsize}
\vspace{-1cm}
\end{table}

Algorithm~5 extends the optimal approximation version of Algorithm~4 to the setting of using confidence intervals.
\begin{theorem}\label{thm:robustbnddegree}
Under Assumption~\ref{def:spat_cond_ind}, Algorithm~5 identifies $\Phatrob$.  
\end{theorem}

The proof appears in Appendix~\ref{app:prf:robustbnddegree}.  A natural question is how robust Algorithm~4's approximations are. 
\begin{corollary} \label{cor:rob_bnd_degree_alg4}
Under Assumption~\ref{def:spat_cond_ind}, if the confidence intervals $\{\wcalI(\allX_{A(i)} \to \X_i)\}_{i \in [m], A \subseteq [m] \backslash \{i\}}$ all have the same width, Algorithm~4 identifies $\Phatrob$.
\end{corollary}
\begin{IEEEproof}
Let $\Delta$ denote the common width. Then for all candidate sets $ B_j \in \calB(j)$, $H(B_j) = L(B_j) + \Delta$.  Thus
\beqa
M(B_{j_1}) &=& \frac{1}{2} ( H(B_{j_1}) + L(B_{j_1})) \nonumber \\
&=& \frac{1}{2} ( H(B_{j_1}) + H(B_{j_1}) - \Delta) \label{eq:prf:rob:apx:r1} \\
&\geq& \frac{1}{2} ( H(B_{j_3}) + H(B_{j_2}) - \Delta) \label{eq:prf:rob:apx:r2} \\
&=& \frac{1}{2} ( H(B_{j_3}) + L(B_{j_2}) ) \label{eq:prf:rob:apx:r3} 
\eeqa where \eqref{eq:prf:rob:apx:r1} uses the common width $\Delta$, \eqref{eq:prf:rob:apx:r2}  uses that $\wcalI(\allX_{B_{j_1}} \to \X_i)  $ has the largest maximal value by construction, and  \eqref{eq:prf:rob:apx:r3} uses the common width $\Delta$.  Since \eqref{eq:prf:rob:apx:r3} is the condition in line~12, this finishes the proof.
\end{IEEEproof} 

\section{Estimation of Directed Information} \label{sec:estimation}

In this section, we derive sample complexity and confidence bounds for two finite-alphabet plug-in estimators, the first based on the empirical distribution and the second on a parametric distribution.  For simplicity of presentation, we focus on jointly estimating directed information between all pairs of processes, $\{\I(\X_j \to \X_i)\}_{i,j\in [m]}$.  The results generalize to jointly estimating directed information involving sets of processes, $\{\I(\X_B \to \X_i)\}_{i \in [m], B \subseteq [m]\backslash\{i\}}$.

\begin{assumption} \label{assump:basic} We assume the network $\allX$ is jointly stationary, ergodic, and Markov of finite order $l$. We further assume that each pair of processes $\{\X_i, \X_j\}$ are Markov order $l$.
\end{assumption}

\begin{remark} 
The pairwise Markovicity is used to simplify notation for  \eqref{eq:didef:2}.  The network Markovicity is needed to ensure joint convergence for all pairwise directed information estimates.   To extend these results for jointly estimating all directed informations with $K+1$ processes, $\{\I(\allX_B \to \X_i) : i\in [m], B \subseteq [m]\backslash\{i\}, |B|=K\}$, each $(K+1)$-tuple of processes $\{\allX_{B\cup\{i\}} \}_{i\in[m], B\subseteq [m]\backslash\{i\}}$ must be Markov order $l$.  
\end{remark}

The network Markovicity, coupled with the fact that the alphabet $\calX$ is finite and $P_{\allX}$ is positive (Assumption~\ref{def:spat_cond_ind}) implies that network is irreducible and aperiodic.  Thus it has a unique stationary distribution.

For notational simplicity, shift the time indexing to start at $t=-l+1$ so for $t=1$ there is a length $l$ history.  Under Assumption~\ref{assump:basic}, the directed information for all ordered pairs $(i,j)$ has the form 
\beqa
\hspace{-0.5cm}  \frac{1}{n} \I(\X_{j} \to \X_i) &=&  \frac{1}{n} \sum_{t=1}^n \I(X_{i, t} ; X_{j,t-l}^{t-1} | X_{i, t-l}^{t-1} ) \label{eq:didef:2}\\
 \hspace{-0.0cm} &=&   \I(X_{i, l+1} ; X_{j,1}^{l} | X_{i, 1}^{l} ) \label{eq:didef:3}\\
 \hspace{-0.0cm} &=& \hspace{-0.2cm} \sum_{x_{j}^l, x_i^{l+1} 
} \hspace{-0.3cm} P_{X_{j}^l, X_i^{l+1}}(x_{j}^l, x_i^{l+1}) 
\log \frac{P_{ X_{i,l+1} |   X_{j}^l, X_{i}^{l}}(  x_{i,l+1} |   x_{j}^l, x_{i}^{l} )}{P_{ X_{i,l+1} |    X_{i}^{l}}(  x_{i,l+1} |   x_{i}^{l} )}.   \label{eq:didef:4}
\eeqa 
Eq.~\eqref{eq:didef:2} follows from Markovicity, \eqref{eq:didef:3} from stationarity, and \eqref{eq:didef:4} from the definition of mutual information \eqref{eqn:defn:MutualInformation}.

We first jointly estimate all pairwise distributions, $\{\widehat{P}_{X_{j}^l, X_i^{l+1}}\}_{i,j\in [m]}$, and then plug those into \eqref{eq:didef:4} to obtain directed information estimates $\{\widehat{\I}(\X_{j} \to \X_i)\}_{i,j \in [m]}$.  The confidence interval for each $(i,j)$ is set as \beqa \wcalI(\X_{j} \to \X_i) := \left[ \widehat{\I}(\X_{j} \to \X_i) - \delta , \, \widehat{\I}(\X_{j} \to \X_i) + \delta \right] \label{eq:est:CI}\eeqa for a given constant $\delta>0$.  Let $B_{\delta}$ denote the event that each confidence interval contains the true value 
\beqa
\hspace{-0.5cm} B_{\delta} \!\!\!\!&:=&\!\!\!\! \mathbbm{1}_{\{ \I(\X_j \to \X_i) \in \wcalI(\X_{j} \to \X_i)  :  i,j \in [m]\}}. \nonumber
\eeqa  We next examine the sample complexity of the two plug-in estimators to characterize $\P( B_{\delta} )$ as a function of $n$.

\subsection{Empirical Distribution} \label{sec:est:emp}
  First consider the ``empirical'' distribution.  For each ordered pair $(i,j)$, compute a distribution $\widehat{P}_{X_{j}^l, X_i^{l+1}}$, where for each possible realization $\{x_{j}^l, x_i^{l+1}\} \in \calX^{2l+1}$ of $\{X_{j}^{l}, X^{l+1}_{i}\}$, the estimates are \beqa
\widehat{P}_{X_{j}^l, X_i^{l+1}}(x_{j}^l, x_i^{l+1}) := \! \frac{1}{n} \!\sum_{t=1}^n \!\mathbbm{1}_{\left\{ \{X_{j,t-l}^{t-1}, X^t_{i,t-l}\} = \{x_{j}^l, x_i^{l+1}\} \right\}}. \label{eq:def:full_emp_dist}
\eeqa

To determine how quickly the empirical distributions $\{\widehat{P}_{X_{j}^l, X_i^{l+1}}\}_{i,j\in[m]}$ jointly converge, and thus the directed information estimates, we need to measure how fast the network converges to its stationary distribution, its ``mixing time.''  To simplify notation, denote the state of the network from time $t-l$ to time $t$ by $\uV_t := \allX^t_{[m], t-l}.$ Then $\{\uV_t\}_{t = 1}^n$ forms a first-order Markov chain.  Let $\pi$ denote its stationary distribution.  Let $\lambda$ be a constant $0 < \lambda \leq 1$ and $d \geq 2$ an integer such that for all $v_1 \in \calX^{m(l+1)}$, $\P(V_{d} = v | V_{1} = v_1) \geq \lambda \pi(v). $

\begin{theorem}\label{thm:emp_est_smpl} Under Assumptions~\ref{def:spat_cond_ind}~and~\ref{assump:basic}, for a given $\delta > 0$, $\P( B_{\delta} ) \geq 1 - \rho,$ where $\epsilon$ is chosen so that $\delta = - 4 |\calX|^{2l+1} \epsilon  \log \epsilon$ and 
\beqa \hspace{-0.5cm} \rho \!\!\!\!&=&\!\!\!\! 8m(m-1)|\calX|^{2l+1}\exp \!\!\left(\!\! - \frac{(n \epsilon - 2 d/\lambda)^2}{2 n d^2 / \lambda^2} \right). \label{eq:def:rho}
\eeqa For any $\epsilon'>0$, the sample complexity of jointly estimating all pairwise directed informations $\{ \I(\X_j \to \X_i   \}_{i,j\in [m]} \}$ is $\delta = \mathcal{O}(n^{-1/2+\epsilon'})$ for fixed $m$ and $n = \mathcal{O}(\log m)$ for fixed $\delta$.
\end{theorem}

The proof appears in Appendix~\ref{app:prf:emp_est_smpl}.

\subsection{Parametric Distribution}

Parametric models are widely used for modeling time-series in economics, biology, and other fields.  We next identify the sample complexity for parametric plug-in estimators.  We consider a network of stochastic processes whose conditional distribution $P_{\allX_t | \allX^{t-1}_{t-l}; \utheta^*}$ is characterized by a $Q$-dimensional parameter vector $\utheta^*$.  We next discuss conditions for the maximum likelihood estimate (MLE) $\wuthetan$ to exist.  These are analogous to the i.i.d. case.

Suppose $\utheta^*$  is in the interior of $\Theta$, a compact subset of $\mathbb{R}^Q$. Let $q$ index the parameter vector $\utheta = \{ \theta_q\}_{q=1}^Q \in \Theta$.
Denote the conditional log-likelihood of $\allX_t$ parameterized by $\utheta$ as
\beqas
L_t(\utheta) := \log P_{\allX_t | \allX^{t-1}_{t-l};\utheta}(\allX_t | \allX^{t-1}_{t-l}).
\eeqas     Define the negative Hessian matrix $A_t(\utheta )$ evaluated at $\utheta '$ as
\beqas A_t(\utheta ')  &=& \left[  - \frac{\partial^2 L_t(\utheta) }{\partial \theta_{q_1} \partial \theta_{q_2} } \bigg|_{\utheta=\utheta'} \right]_{1 \leq q_1,q_2 \leq Q} .
\eeqas

Analogous to the i.i.d. case (eg. pg~384 in \cite{bickel2007mathematical}), the following conditions are sufficient to guarantee asymptotic normality of the MLE error $(\wuthetan - \utheta^*)$ \cite{ling2010general}.

\begin{assumption}\label{asmp:param} (i) $L_n(\utheta)$ is continuously twice differentiable in terms of $\utheta$.  (ii) $\E [L_t(\utheta) ]$ has a unique maximizer at $\utheta^*$.  (iii) $\E[ A_t(\utheta^*)]$ is finite and positive definite.  
\end{assumption}

Define the covariance matrix
\beqa
\Sigma := \left[ \E[ A_t(\utheta^*)] \right]^{-1}. \label{eq:def:param_cov_mat}
\eeqa

\begin{lemma} \label{lem:asym_norm_theta}   Under Assumptions~\ref{def:spat_cond_ind},~\ref{assump:basic},~and~\ref{asmp:param},
\beqa
\sqrt{n}\, ( \wuthetan - \utheta^*) \to \mathcal{N}(0, \Sigma) \hspace{0.5cm} \text{in distribution}. \label{eq:param_est:paramestnorm}
\eeqa
\end{lemma}  The proof is in Appendix~\ref{app:prf:lem:asym_norm_theta}.  It uses \cite{ling2010general} for the main conclusion.  Lemma~\ref{lem:asym_norm_theta} extends to functions of the parameters. Let $\{g_r(\utheta)\}_{r=1}^R $ be a set of $R$ functions of the parameter vector $\utheta$, indexed by $r$.  Using the $Q \times Q$ parameter covariance matrix $\Sigma = (\sigma_{q,q'})$ \eqref{eq:def:param_cov_mat}, define the $R \times R$ covariance matrix $\Sigma' = (\sigma_{r,r'}')$ as
$\sigma_{r,r'}' = \sum_{q = 1}^{Q} \sum_{q' = 1}^{Q} \sigma_{q,q'} \frac{\partial g_r}{\partial \theta_q}\frac{\partial g_{r'}}{\partial \theta_{q'}} \bigg|_{\utheta = \utheta^* }.$

\begin{theorem}[Theorem~5.4.6 of \cite{lehmann1999elements}] \label{lem:asym_norm_DI}   If \eqref{eq:param_est:paramestnorm} holds, each function in the set $\{g_r(\utheta)\}_{r=1}^R $ is continuous and differentiable in the neighborhood of $\utheta^*$, and if the Jacobian matrix with $(r,q)$-th entry $\frac{\partial g_r}{\partial \theta_q}\big|_{\utheta = \utheta^* }$ is non-singular, 
\beqa
&& \hspace{-1.4cm} \sqrt{n} \left[ (g_1(\wuthetan) - g_1(\utheta^*)), \dots,(g_{R}(\wuthetan) - g_{R}(\utheta^*))  \right] 
\to \mathcal{N}(0, \Sigma') \hspace{0.3cm} \text{in distribution}. \label{eq:param_est:DIestnorm}
\eeqa
\end{theorem}

Theorem~\ref{lem:asym_norm_DI} is known as the multivariate delta method. Let $g_r(\utheta)$ specifically be the directed information of the $r$th pair $(i_r, j_r)$ computed with $\utheta$,
$g_r(\utheta) := \I( \X_{j_r} \to \X_{i_r}).$
Assumption~\ref{asmp:param} (i) implies that $g_r(\utheta)$ is continuously differentiable.  The Jacobian matrix with $(r,q)$-th entry $\frac{\partial g_r}{\partial \theta_q}\big|_{\utheta = \utheta^* }$ will be singular if there are linear dependencies between different directed informations.  Even if that occurs, we can nonetheless upper bound the joint convergence rate of the estimates using a ``worse'' covariance matrix $\Sigma'$, as will be done for the proof of the following theorem.  

\begin{theorem} \label{thm:par_est_smpl} Under Assumptions~\ref{assump:basic}~and~\ref{asmp:param}, the sample complexity is $\delta = \mathcal{O}(n^{-1/2})$ for fixed $m$ and $n = \mathcal{O}(\log m)$ for fixed $\delta$.
\end{theorem}

The proof appears in Appendix~\ref{app:prf:par_est_smpl}.

\begin{remark}  Under Assumptions~\ref{assump:basic}~and~\ref{asmp:param}, the unknown covariance matrices $\Sigma$ and $\Sigma'$ in \eqref{eq:param_est:paramestnorm} and \eqref{eq:param_est:DIestnorm} respectively can be consistently estimated by using $\wuthetan$ in place of the unknown $\utheta^*$ \cite{ling2010general}. Calculating $\Sigma'$ might be difficult in some cases.  For practical implementation, confidence intervals for directed information can be approximated as follows.  Separately fit the conditional marginals $P_{Y_t|Y^{t-1}_{t-l};\wuthetan'}$ and  $P_{Y_t|Y^{t-1}_{t-l}, X^{t-1}_{t-l};\wuthetan''}$.  Confidence intervals for  $\wuthetan'$ and $\wuthetan''$ can be calculated.  Sample from those confidence intervals and compute the directed information for each sample to estimate the confidence interval for $\I(\X\to\Y)$. 
\end{remark}

\section{Simulations} \label{sec:sim}

\subsection{Exact Recovery Smulations -- Algorithms 2, 3, and 4} \label{sec:sim:exact}

\subsubsection{Setup} \label{sec:sim:exact:setup}
We first tested Algorithms 2, 3, and 4 using Markov order-1 autoregressive (AR) processes, 
\beqas
\allX_t = C \allX_{t-1} + N_t \label{eq:sim:ARdef}
\eeqas for a given $m$ by $m$ coefficient matrix $C$ and  noise vector $N_t$.

Two network sizes $m \in \{ 6, 15\}$ were used.  For each size $m$, there were $200$ trials each of $n = 750$ time-steps.  In each trial, the parent sets and coefficients were generated.  For each node, the number of parents was chosen at uniform between $0$ and $3$ for $m=6$ and between $0$ and $6$ for $m=15$.  Non-zero AR coefficients were i.i.d.  standard normal.  The matrix $C$ was scaled so that the largest magnitude of its eigenvalues was $0.9$ to have a limiting stationary distribution (see pg. 88 in \cite{lutkepohl2004applied}).  The noise process $\{N_t\}_{t=1}^n$ had i.i.d. $\mathcal{N}(0, \frac{1}{4})$ entries.

Performance was measured by both the proportion of edges correctly identified and by the ratio of the sum of directed information from estimated parents to children as compared to the true parents
\beqa
\frac{\sum_{i = 1}^m \I(\allX_{\widehat{A}(i)} \to \X_i) }{\sum_{i = 1}^m \I(\allX_{A(i)} \to \X_i) }, \label{eq:sim:DIperfrat}
\eeqa where $A(i)$ and $\widehat{A}(i)$ denote the true and inferred parent sets respectively.  This second measure characterizes how much of the dynamics are captured by the estimated parent sets. Performance is averaged over the trials.

Directed information estimates were calculated using least square model fits of the form 
\beqa
Y_t \!\!\!\!&=&\!\!\!\! b_1 Y_{t-1} + b_2 Z_{t-1} + b_3 X_{t-1} + N_t \label{eq:sim:ARfit1}\\
Y_t \!\!\!\!&=&\!\!\!\! b_1' Y_{t-1} + b_2' Z_{t-1}  + N_t' \label{eq:sim:ARfit2}.
\eeqa Let $\sigma$ and $\sigma'$ denote $\std(N_t)$ and $\std(N_t')$ respectively. From Theorem~8.4.1 of \cite{cover2006elements}, the entropy $\H(\Y \| \Z, \X)$ is $1/2 \log( 2 \pi e \sigma^2)$.  The directed information is then $\log \sigma'/ \sigma.$  To avoid over-fitting, we used the minimum description length (MDL) penalty \cite{grunwald2007minimum}, $J * \log_2(n)/(2n)$, where $J$ is the number of parameters.  The first model \eqref{eq:sim:ARfit1}'s total complexity is $\widehat{\H}(\Y\|\Z,\X) + 3 \log_2(n)/(2n)$ and the second model \eqref{eq:sim:ARfit2}'s total complexity is $\widehat{\H}(\Y\|\Z) + 2 \log_2(n)/(2n)$.  To select edges, Algorithms~2~and~3 tested if $\widehat{\I}(\X \to \Y \| \Z) > \frac{(3-2)}{2n} \log_2(n).$

For Algorithm 4, $\calB_{\mathrm{max}}$ initially consisted of a single parent set, denote as $B_{\mathrm{max}}$.  To resolve which other parent sets had the same maximal influence except for numerical discrepancies, we again used the MDL penalty.  We set 
\beqa
\calB_{\mathrm{max}} \!\!\gets\! \left\{\! B\!:\! \I(\allX_B \!\to\! \X_i) > \I(\allX_{B_{\mathrm{max}}} \!\!\!\!\to\! \X_i) - \frac{\log_2(n)}{2n} \! \right \}\! . \label{eq:sim:Alg4mod}
\eeqa  Eq.~\eqref{eq:sim:Alg4mod} uses the property that if the true parent set is $A$, then from over-fitting, $\I(\allX_{B_{\mathrm{max}}} \to \X_i)$ would have value up to $\I(\allX_A \to \X_i) + (|B| - |A|)\log_2(n)/(2n)$.  Instead of searching over values of $(|B| - |A|)$, for simplicity we only considered $|B|-|A|=1.$  Letting $\widehat{A} = \cap_{B \in \calB} B$ denote the inferred parent set, we tested whether
\[
\I(\allX_{\widehat{A}} \to \X_i) > \I(\allX_{B_{\mathrm{max}}} \!\!\to \X_i) - (|B_{\mathrm{max}}| - |\widehat{A}|)   \frac{\log_2(n)}{2n}. 
\]  Otherwise, the difference cannot be explained through over-fitting, so 
we defaulted to accepting $B_{\mathrm{max}}$ as the parent set.  The in-degree bound for Algorithm~4 was set at $K=4$ and $K=8$ for $m=6$ and $m=15$ respectively.

\subsubsection{Results}
The results are shown in Figure~\ref{fig:sim:AR234}.  Standard error bars are drawn.  The algorithms all performed well.  Algorithm~2~and~4 were the best, and their performances were almost identical.  Increases in $m$ did not result in significant degradation. Algorithms~2~and~4 captured almost all the dynamics though misclassified some edges.  That suggests the missed edges were weak influences.  

\begin{figure}[t]
\centering

\subfigure[The proportion of dynamics ($m=6$).]{\label{fig:sim:AR234_m6_DI}\includegraphics[width=.7\figwidth]{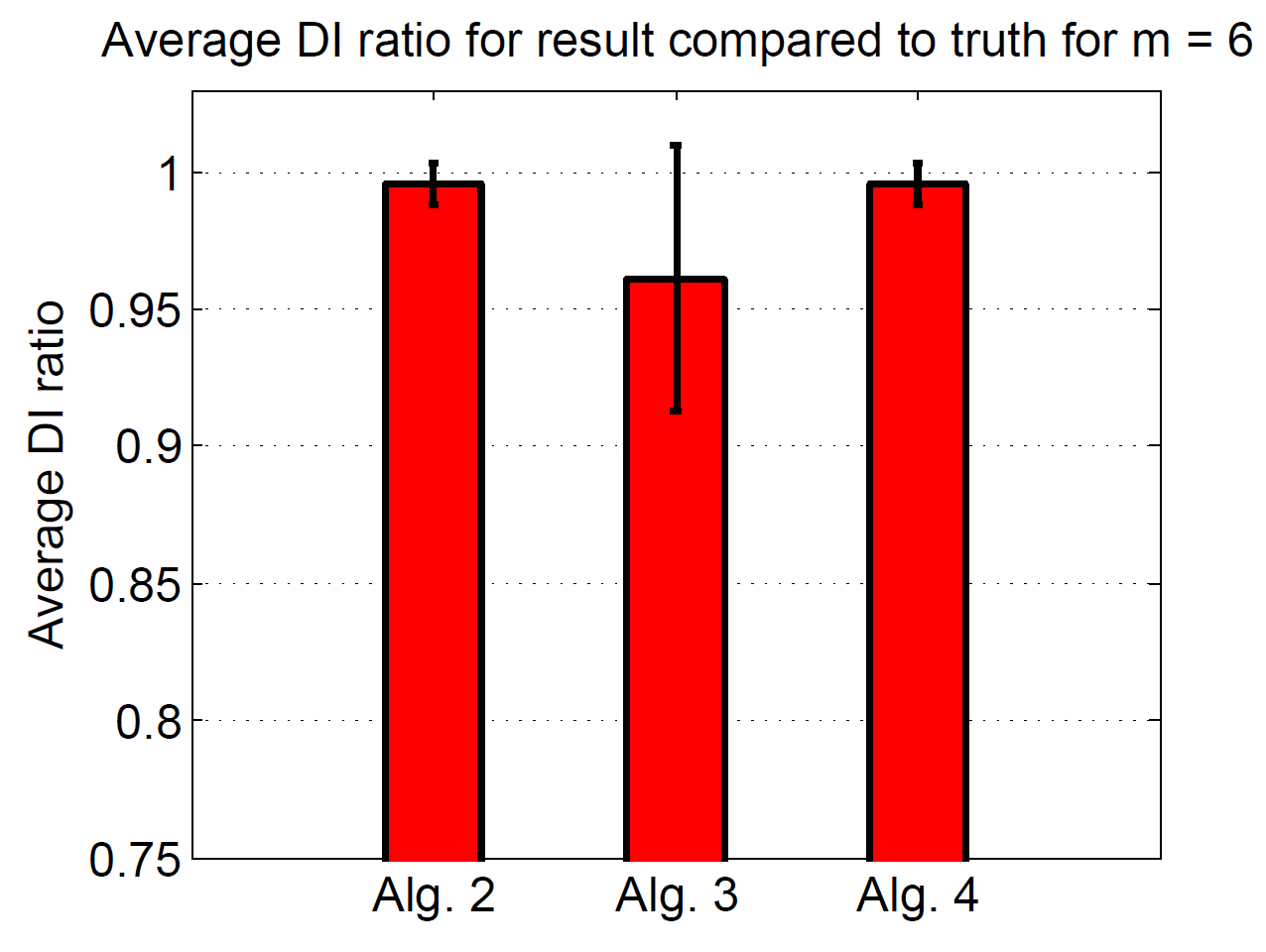}} \hspace{0.1cm}
\subfigure[The proportion of edges ($m=6$).]{\label{fig:sim:AR234_m6_edges}\includegraphics[width=.69\figwidth]{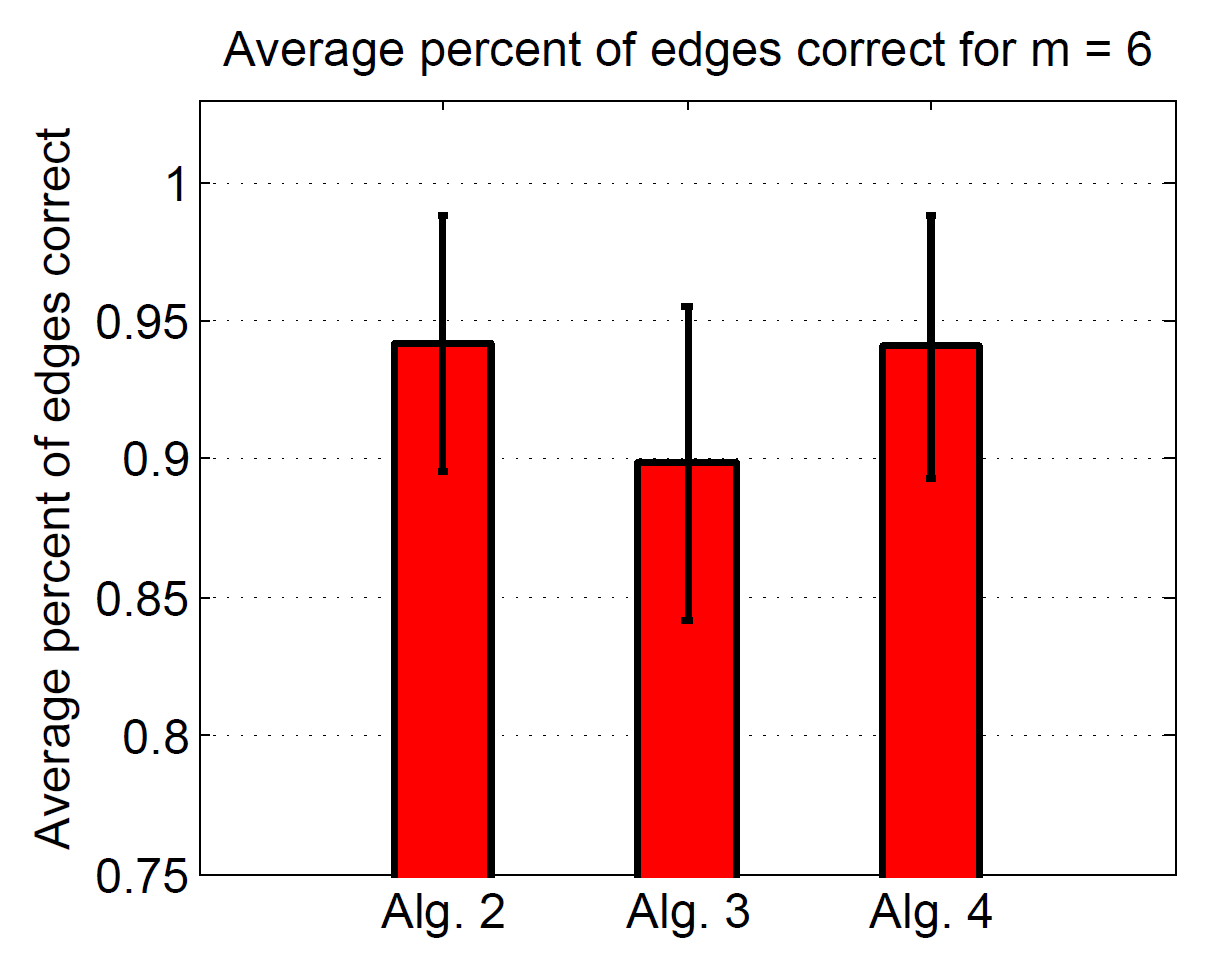}}\\
\subfigure[The proportion of dynamics ($m=15$).]{\label{fig:sim:AR234_m15_DI} \includegraphics[width=.7\figwidth]{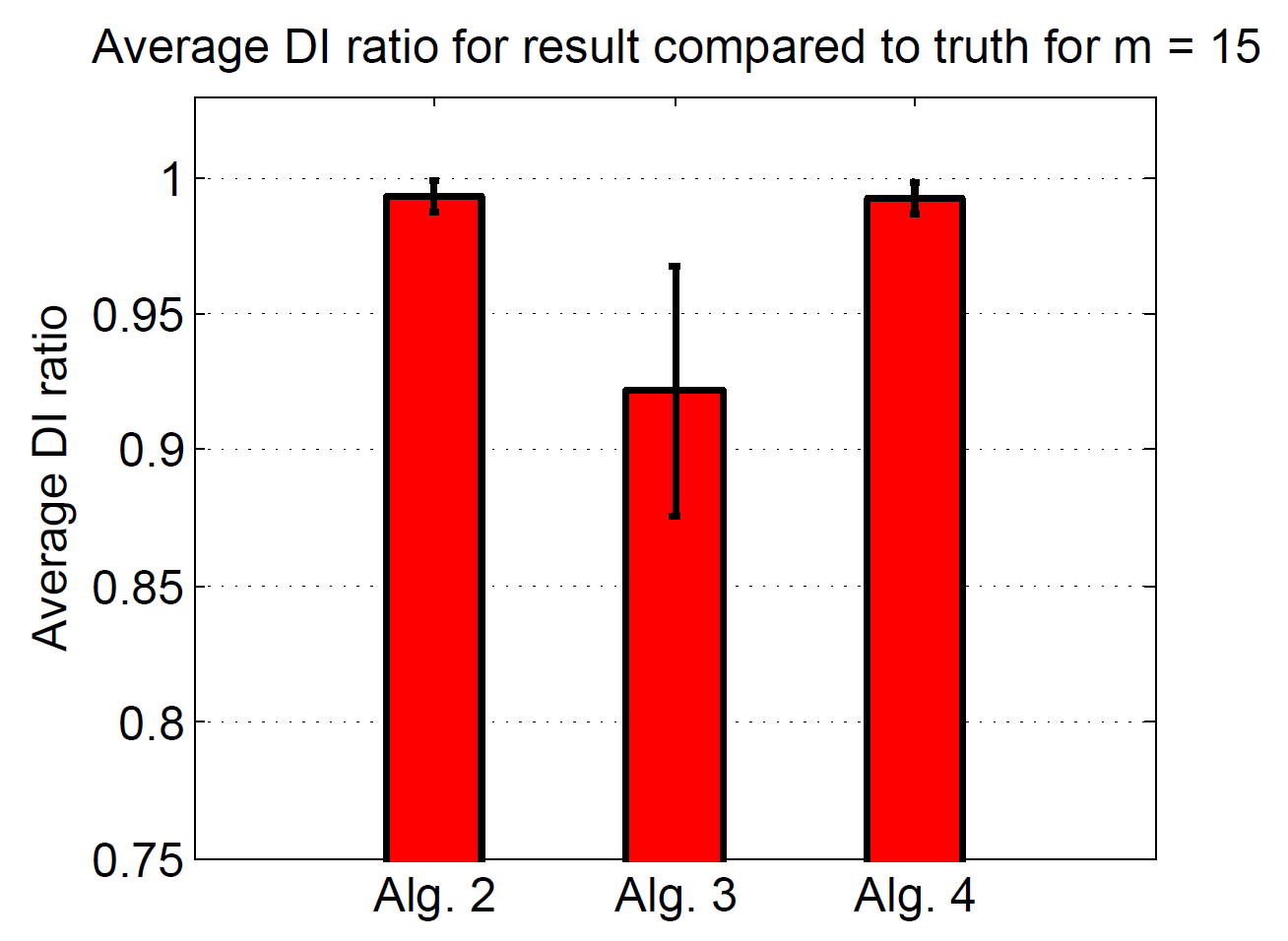}} \hspace{0.1cm}
\subfigure[The proportion of edges ($m=15$).]{\label{fig:sim:AR234_m15_edges}\includegraphics[width=.65\figwidth]{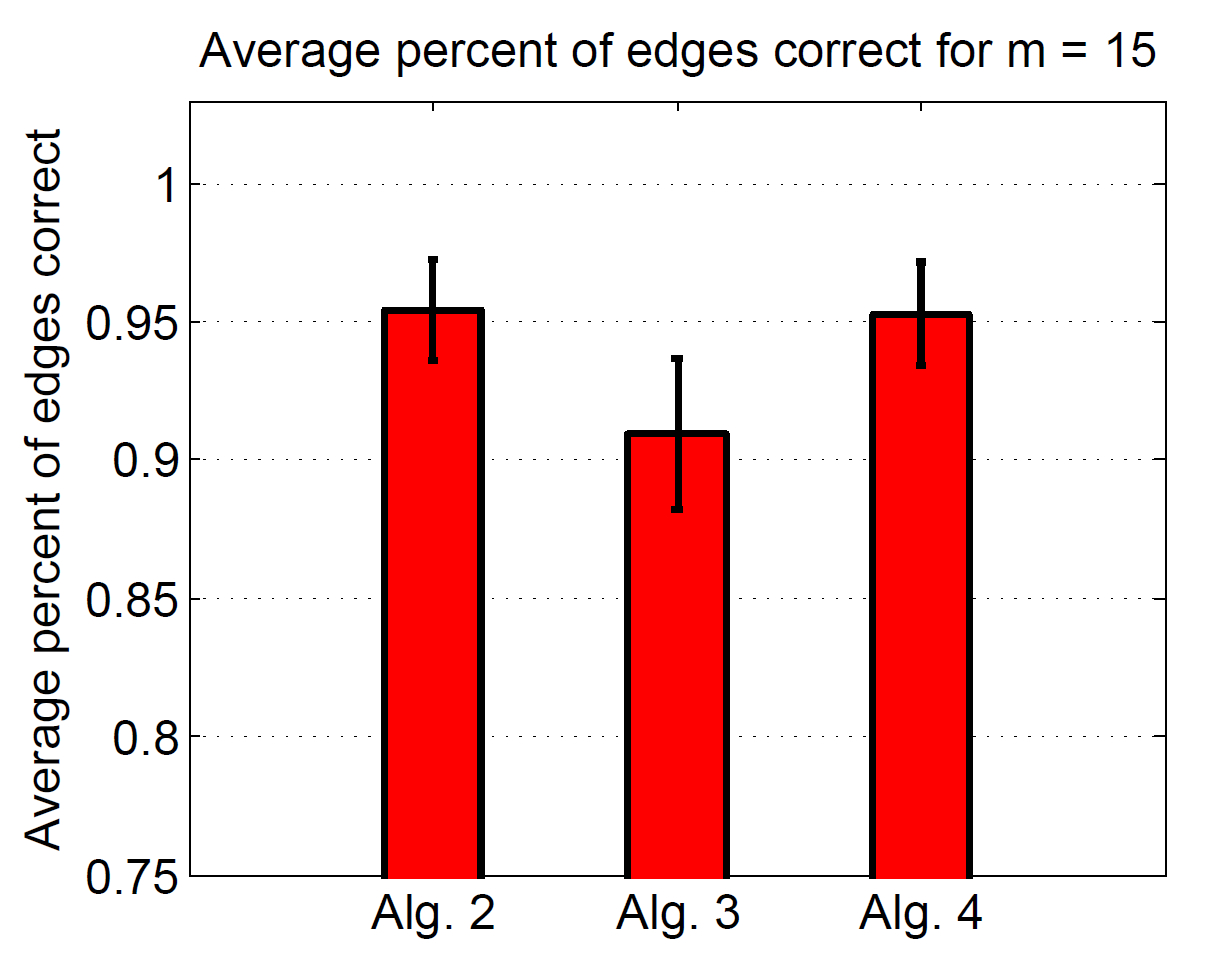}}  
\caption{These figures depict the performance for Algorithms 2, 3, and 4 using randomly generated networks of AR processes.  There were $200$ networks for $m=6$ and $m=15$.  Performance was measured by the ratio between estimated and true parent sets for the sum of directed information from the parent sets to children \eqref{eq:sim:DIperfrat}.  Also, the proportion of edges correctly identified as present or absent was computed.} \label{fig:sim:AR234}
\end{figure}

\subsection{Optimal Approximation -- Algorithm 4}

We also characterized the optimal approximation version of Algorithm~4, discussed in Section~\ref{sec:BndIndegrAppx}.

\subsubsection{Setup}
The setup is similar to Section~\ref{sec:sim:exact:setup}.  Networks of sizes $m = 6$ and $m=15$ were simulated for $n = 750$ time-steps.  In-degrees were not constrained.  Edges were picked i.i.d. with probability $1/2$.  The non-zero AR coefficients were drawn i.i.d. from a standard normal distribution and scaled for stationarity, as in Section~\ref{sec:sim:exact:setup}.  There were $200$ trials for each $m$.  For each trial for $m=6$, in-degree bounds $K \in \{1,2,3,4 \}$ were used.  For $m=15$, bounds $K \in \{2,4,6,8,10 \}$ were used.  The same performance measures in Section~\ref{sec:sim:exact} were used here.

\subsubsection{Results}
The results are shown in Figure~\ref{fig:sim:AR4appx}.  The proportion in dynamics kept by the approximation increases monotonically with the in-degree bound.  Note, however, the percentage of edges correctly identified is concave.  The peak is near the expected number of parents per node, $2.5$ parents for $m=6$ and $7$ parents for $m=15$.  Algorithm~4 does not remove weak edges, though such variations could be done.  Note that for an optimal parent set $B$, if $l$ parents are removed, the resulting parent set is not necessarily the optimal set with $|B| - l$ parents.

\begin{figure}[t]
\centering
\subfigure[ The proportion of dynamics ($m=6$).]{\label{fig:sim:AR4appx_m6_DI}\includegraphics[width=.7\figwidth]{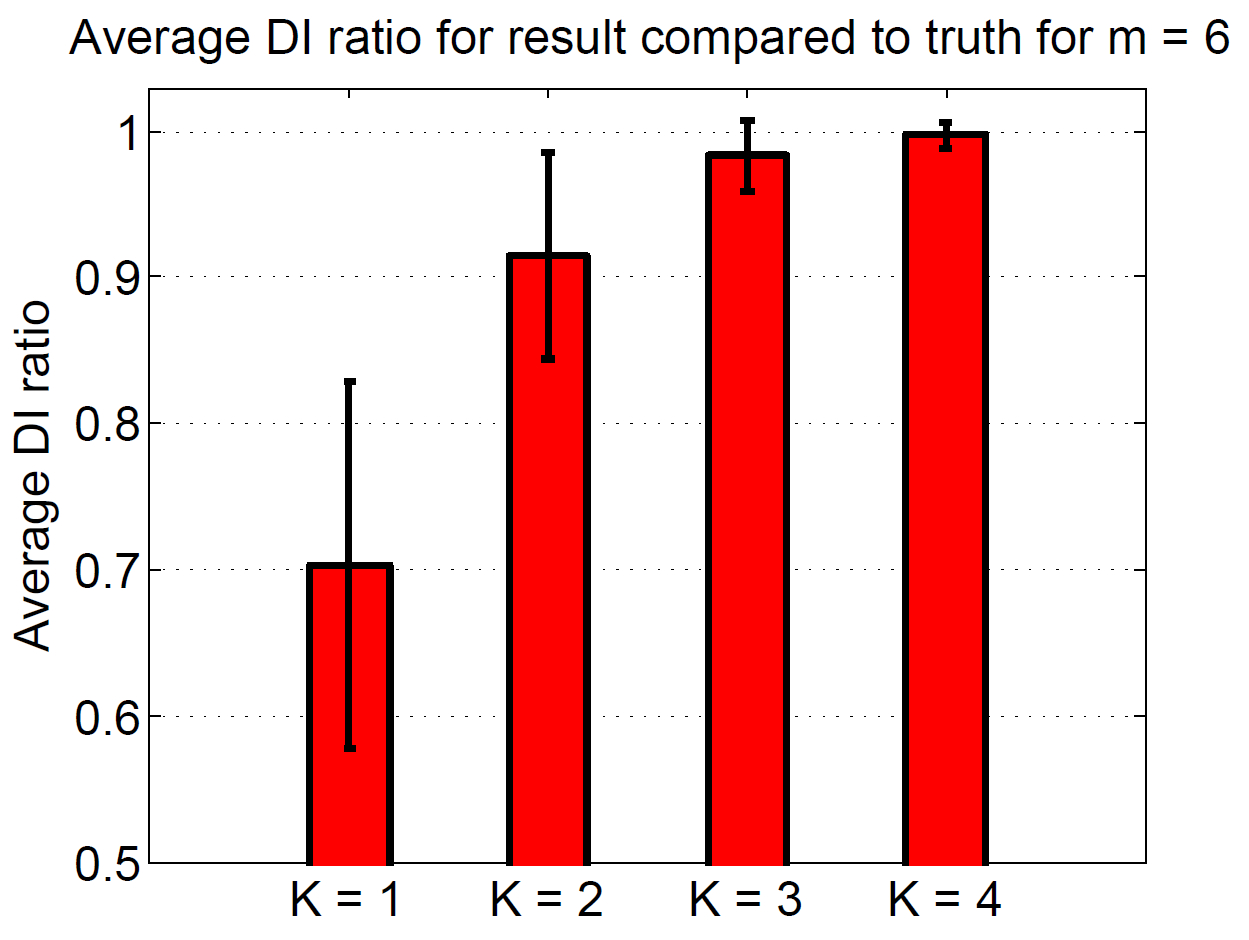}} \hspace{0.1cm}
\subfigure[The proportion of edges ($m=6$).]{\label{fig:sim:AR4appx_m6_edges}\includegraphics[width=.65\figwidth]{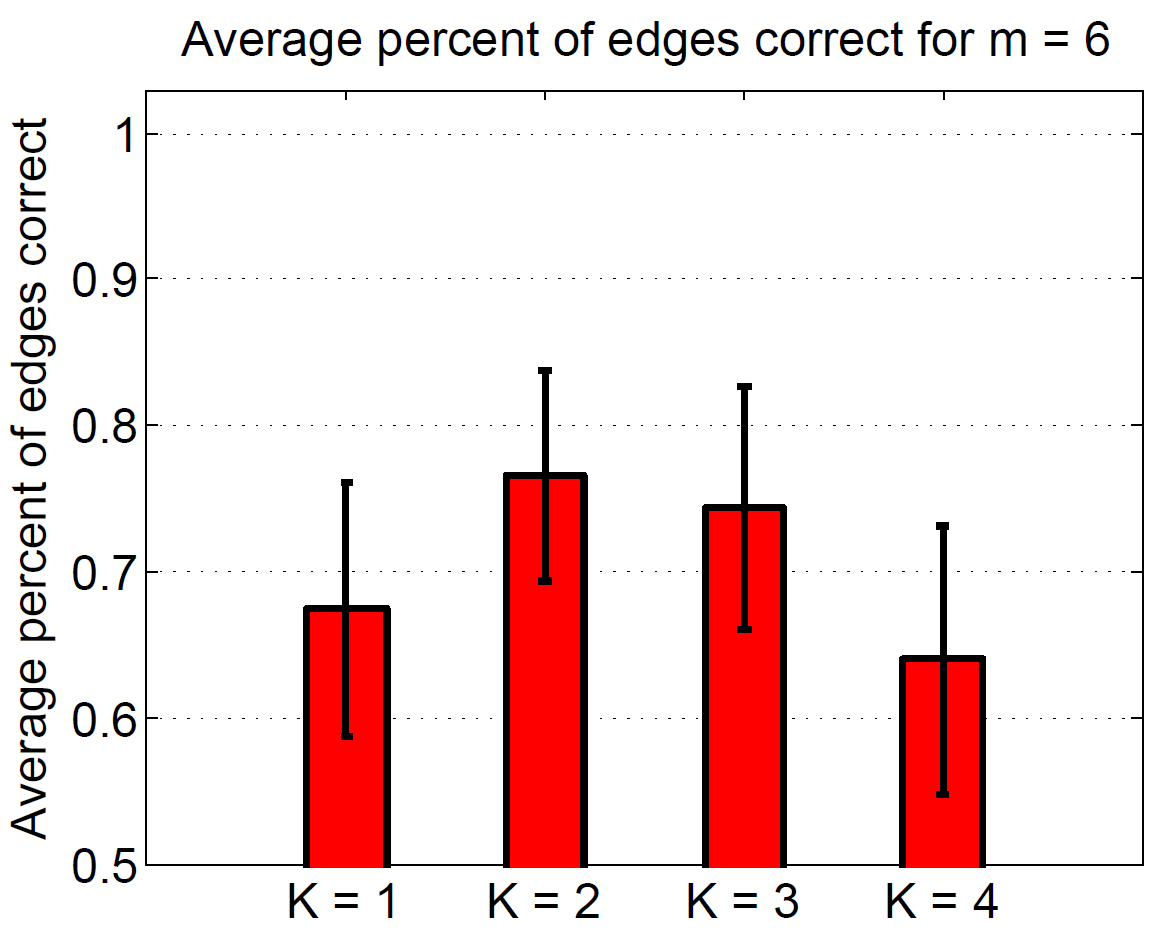}}\\
\subfigure[The proportion of dynamics ($m=15$).]{\label{fig:sim:AR4appx_m15_DI} \includegraphics[width=.7\figwidth]{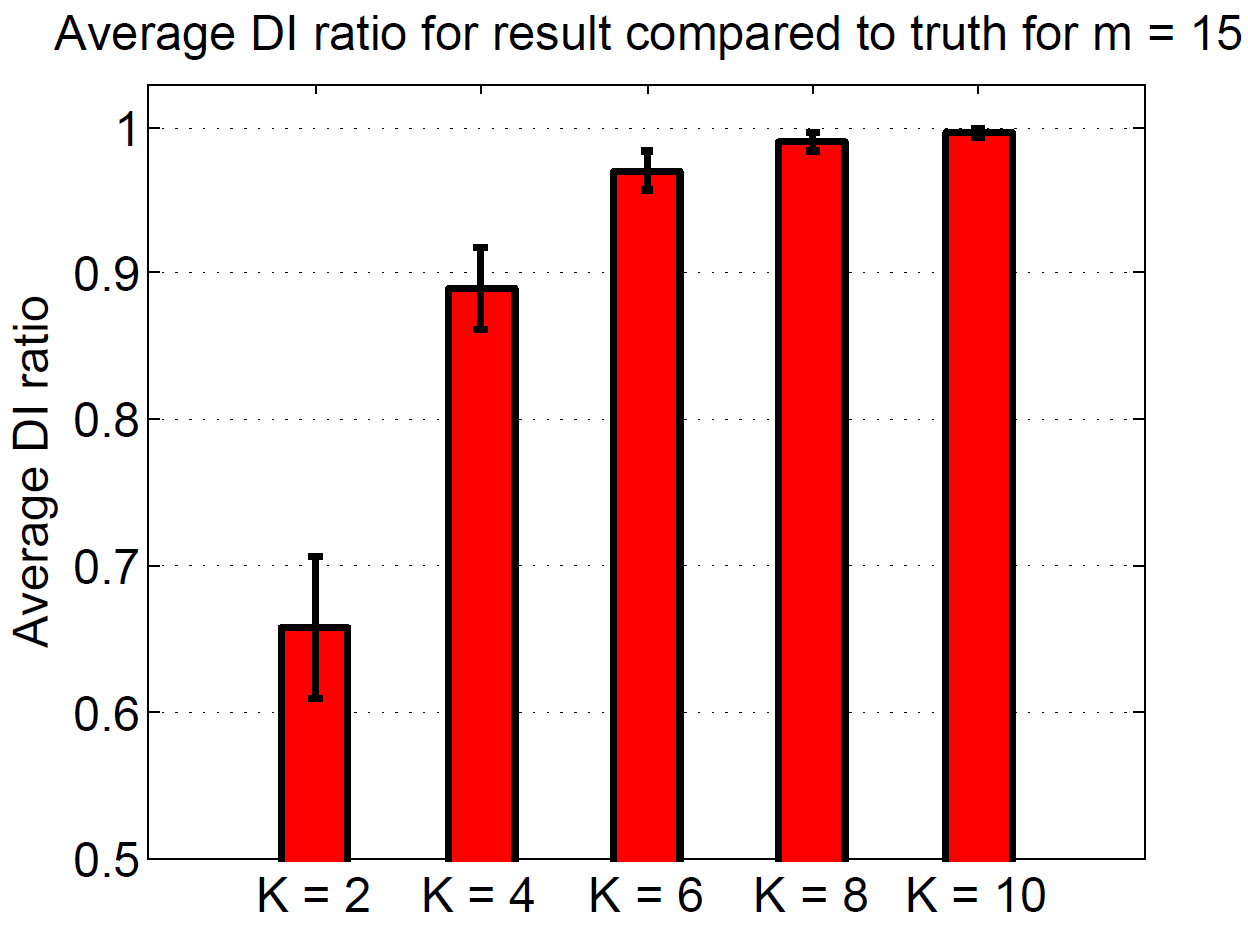}} \hspace{0.1cm}
\subfigure[The proportion of edges ($m=15$).]{\label{fig:sim:AR4appx_m15_edges}\includegraphics[width=.66\figwidth]{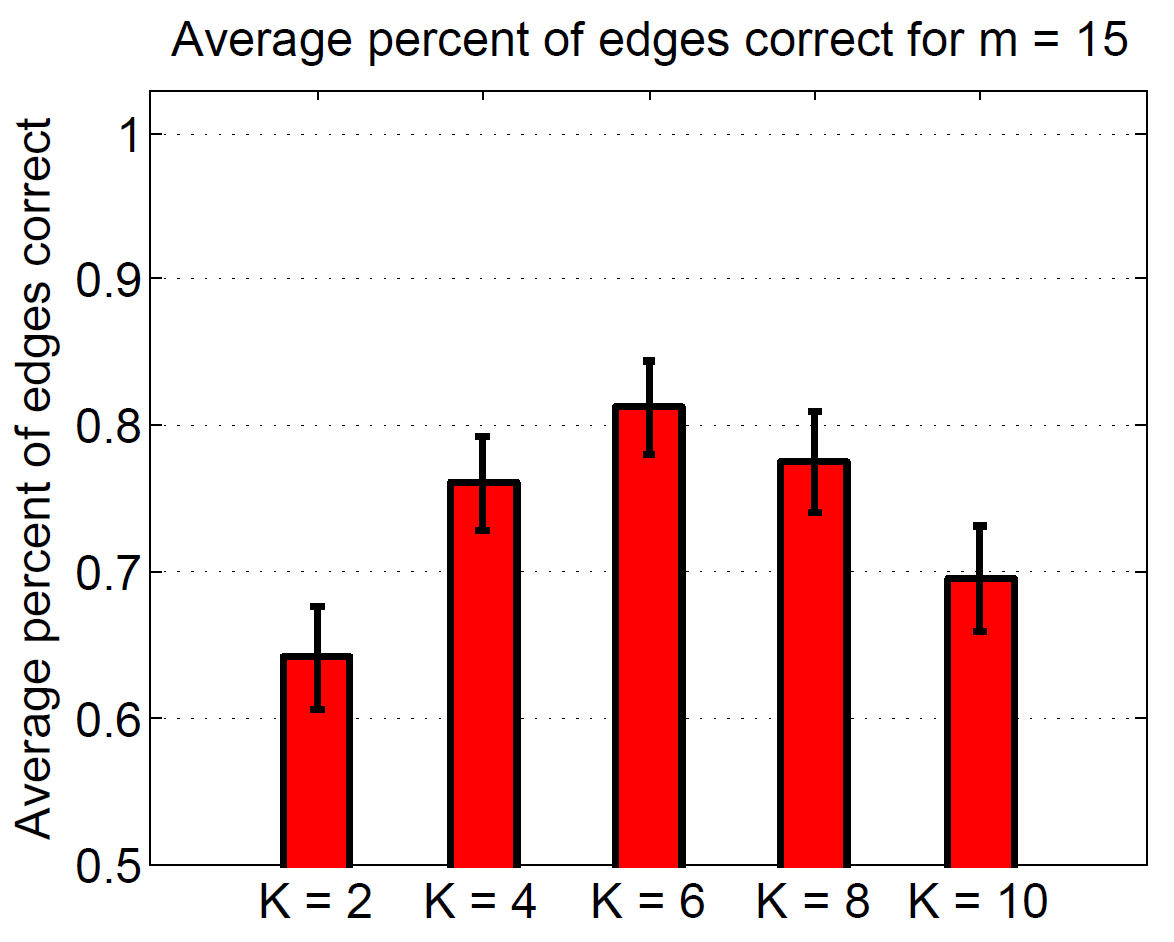}}  
\caption{ These depict the performance for the optimal approximation version of Algorithm~4.  There were $200$ randomly generated AR networks for $m=6$ and $m=15$.  Performance was measured by the ratio between estimated and true parent sets for the sum of directed information from the parent sets to children \eqref{eq:sim:DIperfrat}.  Also, the proportion of edges correctly identified as present or absent was computed.} \label{fig:sim:AR4appx}
\end{figure}

\subsection{Robust Approximations -- Algorithms 4 and 5} \label{sec:sim:rob}

Algorithm~5 identifies the best approximation that is robust to estimation errors.  We next evaluate how Algorithms~4 and 5 compare.  Both use the plug-in empirical estimator discussed in Section~\ref{sec:est:emp}.

\subsubsection{Setup} \label{sec:sim:Alg5:setup}

Network consensus games were simulated.  Each game randomly generated a network of $m=6$ binary valued nodes with in-degree two.  The objective was for the nodes to agree on a value, despite individual bias and limited knowledge.  The biases for $+1$ and $0$ are denoted as $a_{1}, a_0 \geq 0$ respectively with $a_{1} + a_0 = 1$.  At time $1\leq t < 20$, each node $i$ observed its value $X_{i,t}$ and its parents $\allX_{A(i),t}$, and picked its next value $X_{i, t+1}$ using \beqa P_{X_{i,t+1}| \allX_{A(i) \cup \{i\},t}}(1|  \allX_{A(i) \cup \{i\},t})= \frac{ a_1(X_{i,t}+\sum_{j \in A(i)} X_{j,t} ) } {  a_1(X_{i,t}+\sum_{j \in A(i)} X_{j,t})  + a_0(1-X_{i,t}+\sum_{j \in A(i)} 1-X_{j,t})}. \label{eq:alg5sims:1}\eeqa %
Each game consisted of a single network and fixed biases with 50 rounds of play, each round 20 timesteps long. The initial states were i.i.d. Bernoulli($\frac{1}{2}$).  See Figure~\ref{fig:sim:alg5:3} for an example network and biases.

\begin{figure}[t]
\centering
\subfigure[An example network.]{\label{fig:sim:alg5:3} \includegraphics[width=.6\figwidth]{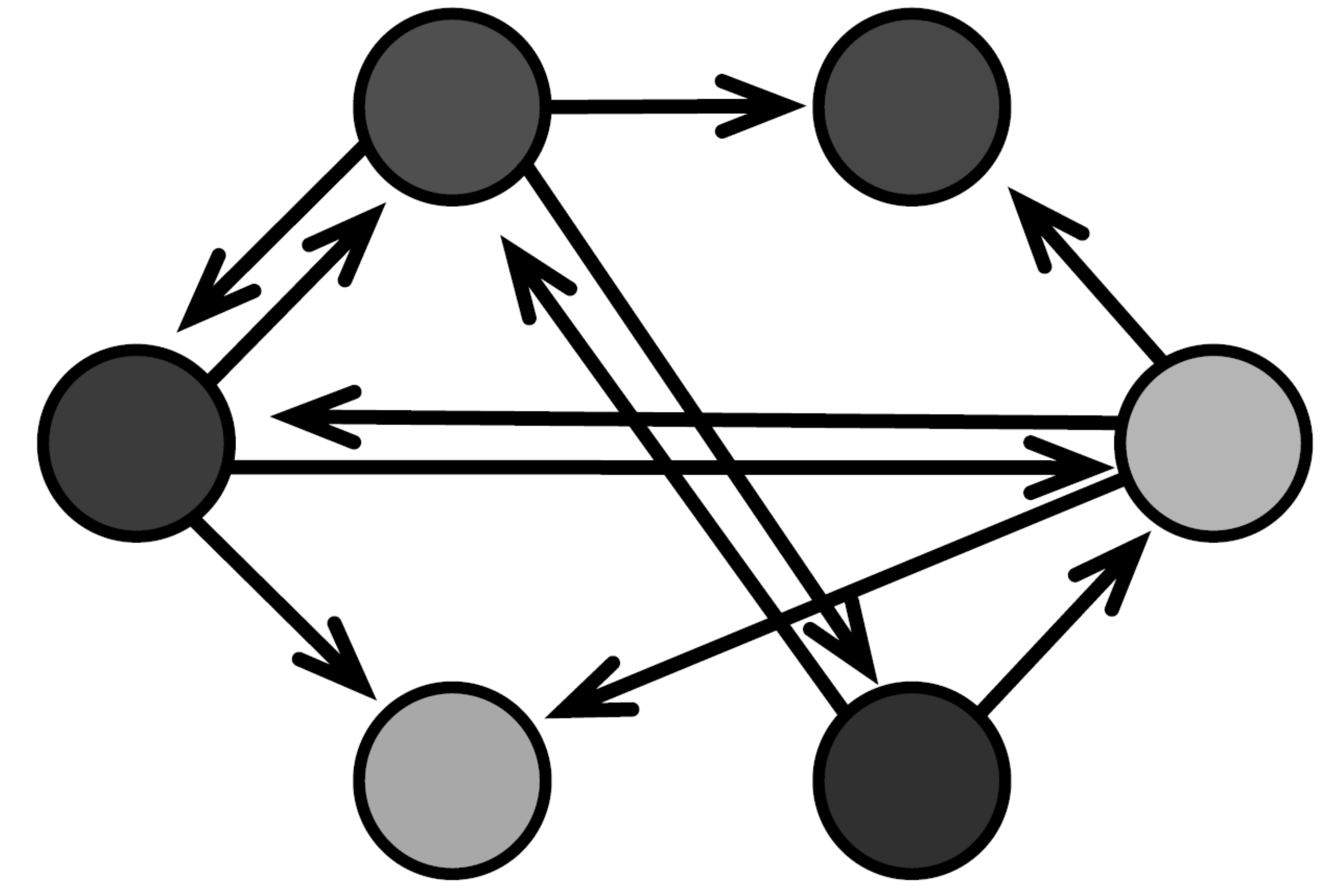}} \hspace{1.1cm}
\subfigure[The percentage of parents identified.]{\label{fig:sim:alg5:4}\includegraphics[width=.75\figwidth]{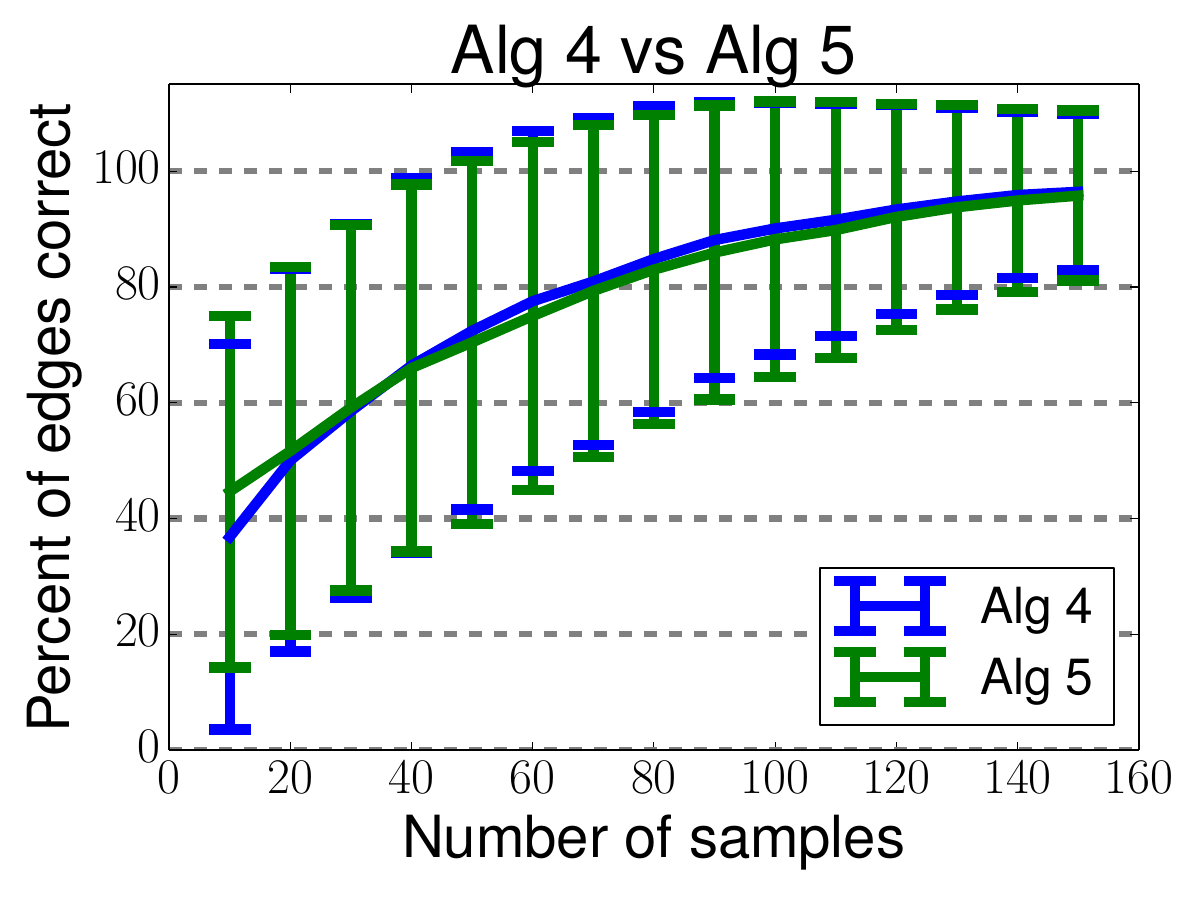}}  
\caption{Figure~\ref{fig:sim:alg5:3} shows the network from one of the games.  Darker colors correspond to larger bias $a_1$.  Figure~\ref{fig:sim:alg5:4} shows the results from the analysis.  Algorithm~4 performs almost the same as Algorithm~5, with a small gap for $n=10$ samples. } \label{fig:sim:alg5}
\end{figure}

\subsubsection{Analysis} There were 150 games total.  For each game, data from the rounds were combined.  Network approximations were obtained using Algorithms~4 and 5 for varying amounts of data $n$ from each game, such as $n=10$ or $n=20$ data points.  The plug-in empirical estimator from Section~\ref{sec:est:emp} was used.  Confidence intervals for Algorithm~5 were derived using bootstrap resampling.   For each potential parent set $\allX_{\widehat{A}}$ of $\X_i$ in each game, $n$ resamples were drawn with replacement from the original $n$ samples, then an estimate $\widehat{\I}(\allX_{\widehat{A}} \to \X_i)$ computed.  This was repeated 500 times.  Then $\widehat{\calI}(\allX_{\widehat{A}} \to \X_i)$ was calculated using the mean of the 500 directed information estimates with a $95\%$ interval width under normality assumptions.  Confidence intervals from Section~\ref{sec:est:emp} were not used because they have uniform width $\delta$, so by Corollary~\ref{cor:rob_bnd_degree_alg4} Algorithms~4 and 5 would output the same graph.

\subsubsection{Results}

Algorithm~4 performed comparably with Algorithm~5.  Even though confidence intervals had different widths, the algorithms often produced the same graph.  Figure~\ref{fig:sim:alg5:4} plots the percentage of correct parents inferred by both algorithms, averaged across 150 games.  The x-axis is the number of samples used. Standard error bars are shown.  Random guessing of parents with $K=2$ would average $40\%$ correct parents.  Thus, with $n=10$ samples, both were close to random guessing, though Algorithm~5 did slightly better.  Algorithm~4 quickly caught up.  These results empirically suggest that Algorithm~4's approximation was typically robust.

\section{Social network analysis} \label{sec:analysis_twitter}

We also demonstrated the utility of Algorithms~2,~3,~and~4 by inferring which news sources influenced which users in the online micro-blogging network Twitter.  All of the news sources covered major events in the Middle East during late 2013.  By analyzing only the times of relevant posts from the news outlets and the users, the algorithms identified which news agency accounts influenced which user accounts with high precision. 

\subsection{Setup}
\subsubsection{Data}
We analyzed activity on the micro-blogging platform Twitter.  Users view messages, ``tweets,'' from accounts they follow.  Users can post novel messages or repost others' messages, ``retweets.''   For data collection, $16$ accounts of major news corporations were selected, such as ABC News, Agence France-Presse, and Reuters Top News. Three corporations had multiple accounts which re-tweeted each other.  We 
retrieved the news accounts' tweets between October 10, 2013 and Dec. 10, 2013.  We focused on tweets containing at least one of the keywords \{`Syria', `Strike', `Assad', `Chemical', `Intervention', `Iraq',`Afghanistan',`Iran',`Terrorist'\}.  A group of $48$ users was picked who had at least five retweets with a keyword from one news source. 

Figure~\ref{fig:twitter:example_timeseries} shows a 48 hour snapshot of activity from two news sources and two user accounts.  Tweets containing a relevant keyword are represented by long black lines.  Other tweets are depicted with short green lines.  User {\bf BoneToBone\_} retweeted content from {\bf BBCBreaking}, and user {\bf hrblock\_21} retweeted content from {\bf FoxNews}.  Note the long periods of inactivity of user {\bf hrblock\_21}.  Also note that the first tweets of {\bf BoneToBone\_} and {\bf hrblock\_21} containing a keyword were after {\bf BBCBreaking} and {\bf FoxNews} tweeted using the keywords.

\begin{figure}[t]
\centering
\subfigure[Tweet activity for four accounts.]{\label{fig:twitter:example_timeseries} \includegraphics[width=.95\figwidth]{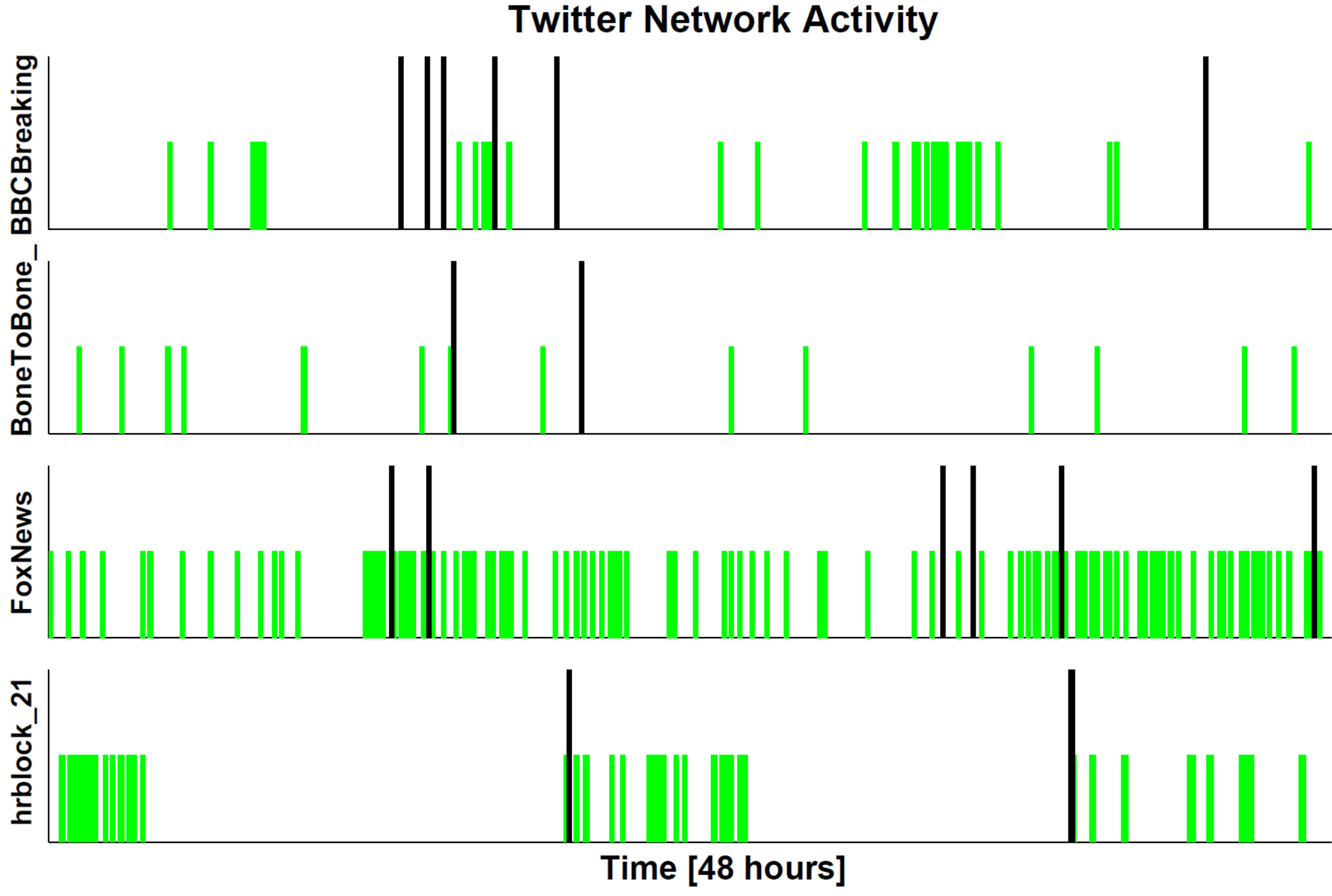}} \hspace{1.1cm} %
\subfigure[User modeling.]{\label{fig:twitter:diagram_modeling} \includegraphics[width=.5\figwidth]{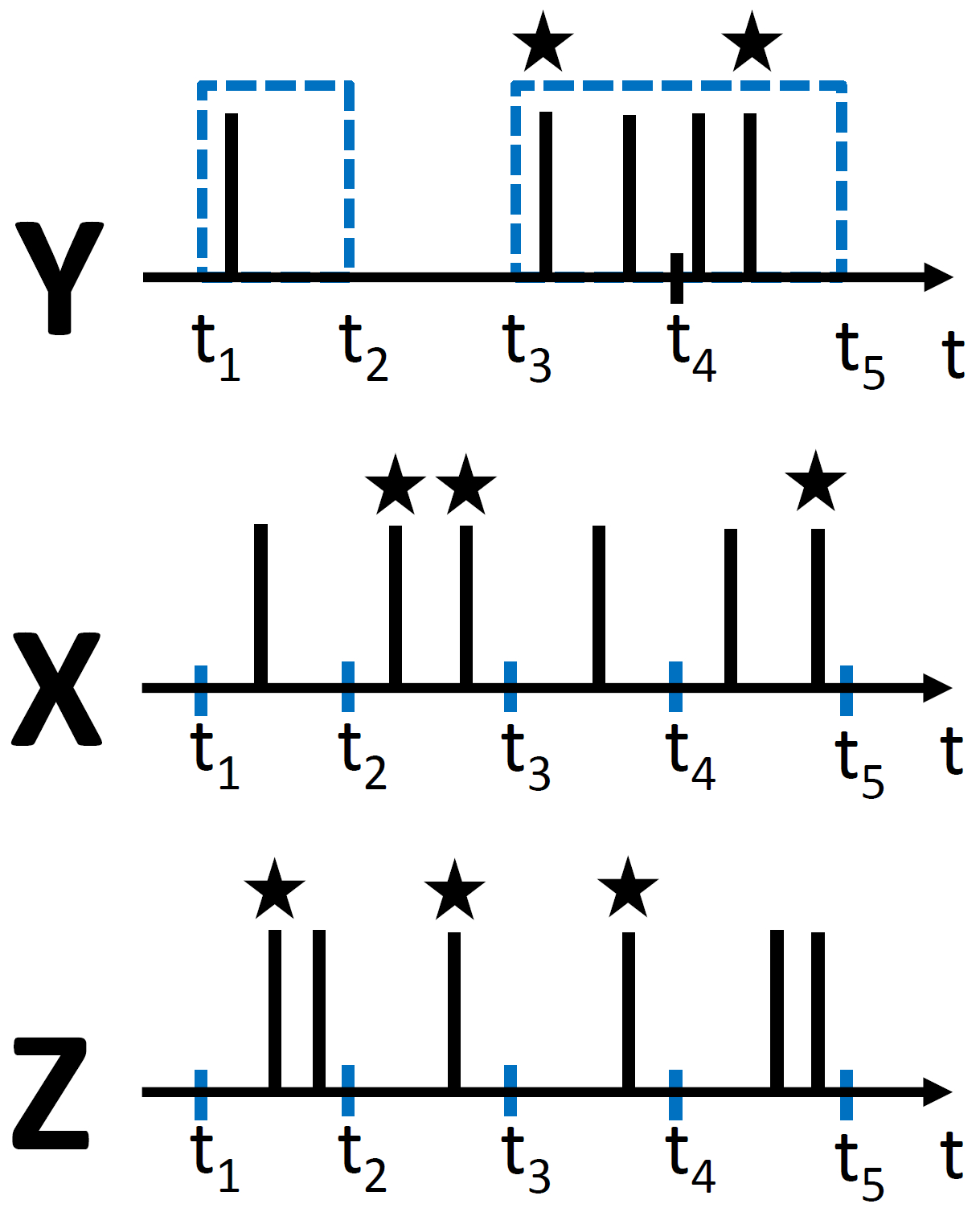}}  
\caption{Figure~\ref{fig:twitter:example_timeseries} shows tweet activity for four accounts over a 48 hour period.  Tweets containing a relevant keyword are represented by large black lines.  Other tweets are depicted with short green lines.  User {\bf BoneToBone\_} retweeted {\bf BBCBreaking}.  User {\bf hrblock\_21} retweeted {\bf FoxNews}.      Figure~\ref{fig:twitter:diagram_modeling}  shows a diagram of how user $\Y$'s tweet activity was modeled as depending on its own past, the past of $\X$, and the past of $\Z$.  The solid vertical lines denote tweet times.  The dashed boxes denote $\Y$'s inferred active periods.  Stars denote tweets with a keyword.  } \label{fig:twitter:usermodeling}
\end{figure}

\subsubsection{Ground-truth} \label{sec:twitter:ground_truth} To establish the ground-truth whether news source $\X$ influenced user $\Y$, the following two conditions were used.  User $\Y$ retweeted at least $5$ tweets of $\X$ that contained a keyword. Of all of $\Y$'s tweets that contained a keyword, at least $15\%$ were retweets from $\X$. 

\begin{remark}  These parameters were manually selected.  No formal sensitivity analysis was conducted.  Performance was observed to degrade for lower thresholds. 
\end{remark}

\subsection{Modeling} We next discuss time-series modeling.  Message arrival times in communication networks form point processes.  We used a logistic model for how  users' tweeting activity depended on past activity.  

Although a user $\Y$ receives messages continuously, $\Y$ might access Twitter intermittently, as is seen in Figure~\ref{fig:twitter:example_timeseries}.  If $\Y$ tweeted at a time $t$, with or without a keyword, then we modeled $\Y$ as being active at least one minute before and $\Y$'s median delay plus three minutes after.  Overlapping active periods were merged.  See Figure~\ref{fig:twitter:diagram_modeling}.

We next describe how the likelihood of each user $\Y$'s activity was modeled, given news sources $\X$ and $\Z$. Active periods of $\Y$ were divided up into intervals of length equal to $\Y$'s median retweet time.  Each interval was modeled as a binary variable, with value one if $\Y$ had a tweet (possibly a retweet) which contained a keyword.  That variable was conditioned on the tweets of $\X$, $\Y$, and $\Z$ during the previous interval.  Let $\Delta t$ denote a time interval and $N_{\X}(\Delta t)$ the number of tweets containing a keyword that were sent by $\X$ during $\Delta t$.  Define $\widetilde{N}_{\X}(\Delta t) $ as 
\begin{equation*}
\hspace{-0.05cm}\widetilde{N}_{\X}(\Delta t) :=
\begin{cases}
N_{\X}(\Delta t) & \text{if } N_{\X}(\Delta t) \leq 2, \\
1 \!+\! \left\lceil \log( 1 \!+\! N_{\X}(\Delta t))\right\rceil    & \text{o/w}.\label{eq:twitter:logcount}
\end{cases} 
\end{equation*} We used $\widetilde{N}_{\X}(\Delta t)$ to describe the past of $\X$ during interval $\Delta t$.  If $\Y$ was inactive for a long period, then $\X$ might have had many tweets, but $\Y$ might only read through a few of them.

Consider a time interval $\Delta t$ during which $\Y$ was active.  Let $\Delta t '$ denote the preceding time interval.  We modeled $\Y$'s tweeting activity during $\Delta t$ as depending on the past of $\X$, $\Y$, and $\Z$ with the following logistic model
\beqa
&& \hspace{-1.3cm} P( N_{\Y}(\Delta t) > 0 \| \X, \Y, \Z) = 1/(1 + e^{-(\alpha_0 + \alpha_1 \widetilde{N}_{\X}(\Delta t ') + \alpha_2 \widetilde{N}_{\Y}(\Delta t ') +\alpha_3 \widetilde{N}_{\Z}(\Delta t '))}), \label{eq:twitter:log_model}
\eeqa %
where $\{\alpha_0, ..., \alpha_3\}$ are coefficients.

\subsection{Estimation} Directed information estimates were computed using the consistent, parametric estimation technique proposed in \cite{quinn2011estimating}.  To estimate the directed information $\I(\X \to \Y \| \Z)$, we first estimated two causally conditioned entropy terms, $\H(\Y\| \Z)$ and $\H(\Y \| \X, \Z)$.  For each entropy term, a logistic model of the form \eqref{eq:twitter:log_model} was fit using generalized linear regression functions. Denote the observed likelihood function as
$
l_{\Y}(\Delta t) := P( N_{\Y}(\Delta t) > 0 \| \X, \Y, \Z).
$
The entropy was then estimated as %
$
\widehat{\H}(\Y \| \X, \Z) := \frac{1}{n'} \sum_{\Delta t} - \log_2 l_{\Y}(\Delta t), \label{eq:twitter:entr_est}
$ where the summation was over all periods $\Delta t$ when $\Y$ was active, and $n'$ was the number of such periods.  The estimate $\widehat{\H}(\Y \| \Z)$ was computed in the same manner.  The directed information estimate was then %
\beqas
\widehat{\I}(\X \to \Y \| \Z) := \widehat{\H}(\Y \|\Z) - \widehat{\H}(\Y \| \X, \Z). \label{eq:twitter:DI_est}
\eeqas

To avoid overfitting, we used the minimum description length (MDL) penalties \cite{grunwald2007minimum}.  For a parametric entropy estimate $\widehat{\H}(\Y \| \X, \Z)$ with $J$ parameters and $n'$ observations, the MDL complexity is $J \log_2(n') / (2 n')$.  The logistic model \eqref{eq:twitter:log_model} with $J-1$ processes has $J$ parameters.  Thus, the estimate $\widehat{\I}(\X \to \Y \| \Z)$ was considered significant if 
\beqa
\widehat{\H}(\Y \| \Z) + \frac{3 \log_2(n')}{2n'} \!\!\!\!&>&\!\!\!\! \widehat{\H}(\Y \| \X, \Z) + \frac{4 \log_2(n')}{2n'} \nonumber \\
\widehat{\I}(\X \to \Y \| \Z) \!\!\!\!&>&\!\!\!\! \frac{\log_2(n')}{2n'}. \nonumber 
\eeqa

MDL penalties were used for Algorithms~2,~3,~and~4.  For Algorithm~4, we took the maximum over $\I(\allX_B \to \X_i) - (|B|+1)\frac{\log_2(n')}{2n'}$ to  quantify how informative conditioning on the past of $\allX_B$ was above the amount expected from over-fitting.  We then set $\mathcal{B}$ as the set of all $B$'s with values within $90\%$ of that maximum.

Numerically, the coefficients $\alpha_i$ in the logistic model \eqref{eq:twitter:log_model} could have been positive or negative.  A positive coefficient $\alpha_1$ for $\widetilde{N}_{\X}(\Delta t)$ corresponded to $\Y$ having an increased likelihood of posting a tweet or retweet with a keyword, if $\X$ posted one or more tweets with keywords in the previous period.  Such positive influences were known to be present in the data; retweeting is an example.   However, a negative coefficient was more likely due to over-fitting than a news agency's activity suppressing a user's activity.  In Algorithms~2,~3,~and~4, any process $\X$ that had corresponding negative coefficient in the logistic model \eqref{eq:twitter:log_model} was rejected.  In Algorithm~4, if a set $B$ had negative coefficients, those processes were removed from $B$ and another fit on the remaining processes was performed.  This was repeated until a subset of $B$ with only positive coefficients was obtained.

\subsection{Evaluation} We now describe evaluation criteria. Each algorithm inferred a graph.  Based on the ground truth in Section~\ref{sec:twitter:ground_truth}, each (non)edge was true positive (TP), false positive (FP), true negative (TN), or false negative (FN).  The following criteria were used to evaluate the performance \cite{fawcett2004roc}. Accuracy $(TP + TN)/(TP + FP + TN + FN)$ measures the proportion of correct labels.  Precision $TP /(TP + FP)$ measures the proportion of correctly inferred edges.  True positive rate $TP/(TP + FN)$ measures the proportion of influences that were identified. False positive rate $FP/(FP + TN)$ is the proportion of non-influences that inferred as edges. We compared the algorithms to the expected performance of a baseline algorithm that knew in-degrees but randomly guessed influences.

\subsection{Results} 

The algorithms performed comparably and significantly better than baseline. Fig.~\ref{fig:Twitter:ROC} depicts the algorithms' performance on an ROC plot.   Table~\ref{tbl:comp_props} shows the values of the performance metrics.   The ground-truth graph was sparse, and the average in-degree was $1.3 \pm 0.5$.  The algorithms had approximately $95\%$ accuracy.  They correctly identified many non-influences as TN.  The baseline also had high accuracy due to the sparsity.  If a user had a single parent, even if the baseline guessed the wrong parent, $14$ of the potential influences would have been correctly identified as TN.  Each algorithm had a very low FPR, meaning they were highly conservative.  They selected few edges, but selected correctly.  Hence they had high precision.  There was some variation amongst them.  Algorithm~2 was the most conservative and Algorithm~4 with $K\!=\!3$ was the least.  This is reflected in the monotonic decrease in precision and increase in FPR.  Note that the differences in precision between the algorithms were larger than the increases in FPR because the ground-truth graph is sparse.

\begin{figure}[t]
\centering
\subfigure[ROC plot.]{\label{fig:Twitter:ROC} \includegraphics[width=.8\figwidth]{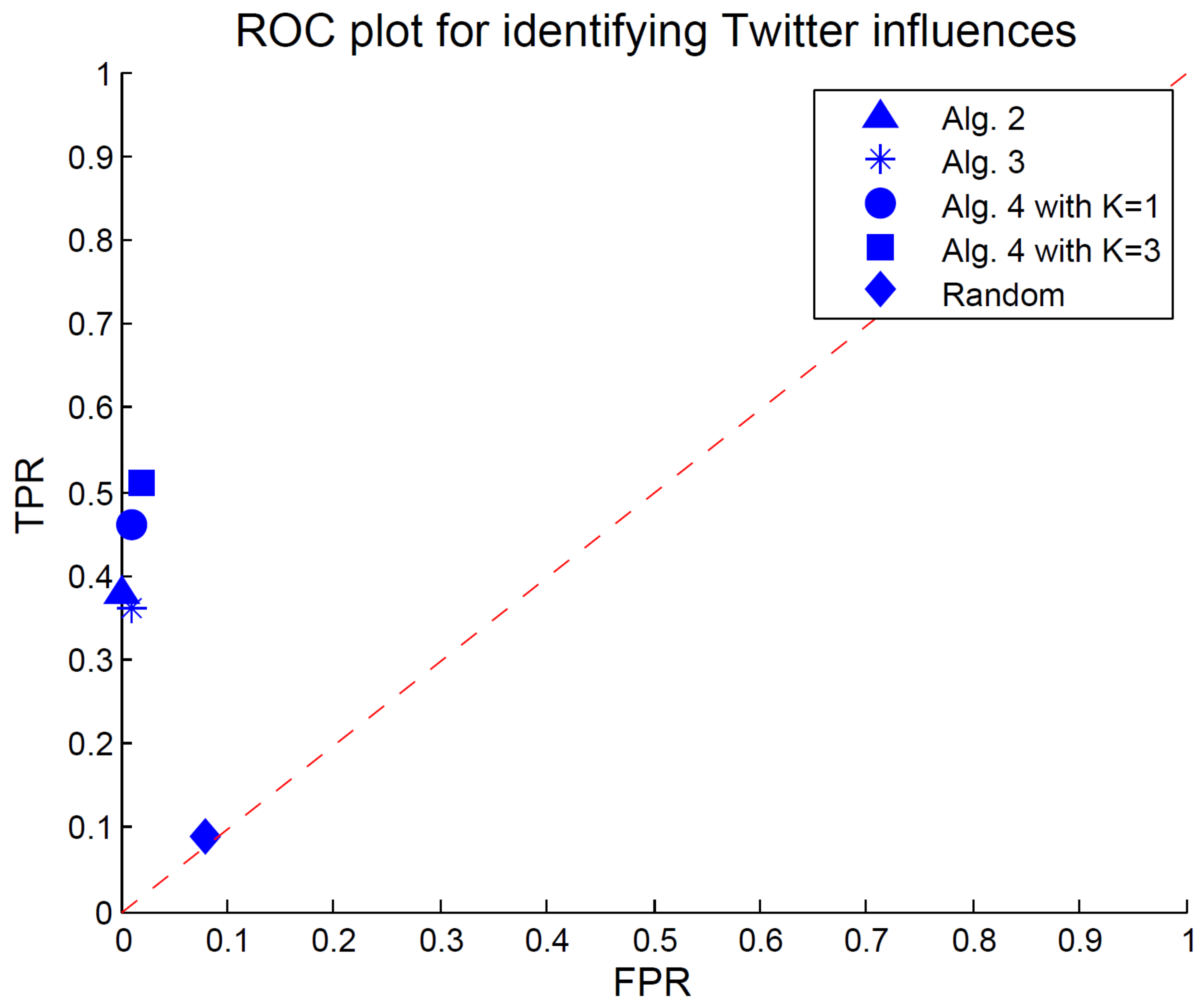}} 
\hspace{1.1cm} %
\subtable[Performance metrics.]{\label{tbl:comp_props} 
				\begin{tabular}[b]{ |l | c | c | c | c |}
				\hline 
					 & Acc. & Prec. & TPR & FPR\\ \hline    
					Alg.~2 & 95& 100& 38& 0  \\ \hline  
					Alg.~3 & 94& 81& 36 & 1  \\ \hline  
					Alg.~4 $K=1$ & 94 & 74& 46 & 1  \\ \hline  
					Alg.~4 $K=3$ & 94& 70& 51 & 2  \\ \hline  
					Random & 86& 9& 9 & 8\\ \hline  
				\end{tabular}
}  
\caption{Figure~\ref{fig:Twitter:ROC}  visually depicts the performance of the algorithms using a ROC plot.  All four algorithms are clustered on the y-axis.  Overall, the algorithms were conservative, selecting few influences, but selecting them correctly.  Table~\ref{tbl:comp_props} shows the accuracy, precision, TPR, and FPR for the algorithms and the baseline algorithm.  } \label{fig:twitter:results}
\end{figure}

\section{Conclusion and Future directions}
Methods that characterize which agents causally influence which other ones in a network could significantly bolster research in a number of diverse disciplines, including social sciences, economics, biology, and physics.  We proposed a widely-applicable framework to address this issue.  It included meaningful graphical representations for networks of interacting agents, multiple algorithms to identify the underlying graph---in some cases using prior knowledge to improve efficiency, and procedures to estimate required statistics from data, as well as robust algorithms when the estimates were not reliable.  We demonstrated the practical utility of the framework by identifying which news agencies influenced which users in the Twitter network with high precision.

There are a number of directions for future research.   One is to improve estimation techniques of directed information.  As discussed in this paper, there are already several estimation techniques \cite{quinn2011estimating, kim2011Granger, so2012assessing, jiao2012universal}.  Computational feasibility of current methods needs to be  further explored, especially for data-rich applications.  Also, small-sample the performance of these estimators is not well characterized.

Another major avenue of future research involves time-varying graphical models and estimation procedures.  The graphical model proposed here assumes the graph itself is time-invariant.  Likewise, the estimation techniques rely on stationarity assumptions for how the future of processes depends on the past of others, in order to establish convergence.  Especially for long-term studies of social networks, biological networks, and economic networks, graphical models and estimation techniques that can handle dynamic topologies would be greatly beneficial.

\appendices

\section{Strictly Causal Approximations and Information Decomposition} \label{sec:indent_struct:info_decom}

This work assumes strict causality (see Assumption~\ref{def:spat_cond_ind}).  Granger discussed that strict causality is a valid assumption if the sampling rate is high enough and all relevant processes are observed \cite{granger1969investigating, granger1980testing}.  In this section we demonstrate that, in an appropriate sense, strictly causal influences and instantaneous influences can be separately accounted for.  We also identify a chain rule for the total correlation of the network.

In this section, $P_{\allX}$ need not be strictly causal.  Denote the strictly causal marginal of $P_{\allX}$ as
\begin{eqnarray}
P_{\allX}^{\mathrm{caus}}(\allx) \!\!&:=&\!\!  \prod_{i=1}^m  P_{\bX_i \parallel \allX_{\setmi{m}{i}} } (\bx_i \parallel \allx_{\setmi{m}{i}}). \nonumber \label{eq:def:app:strc_caus}
\end{eqnarray}

 Let $\widehat{P}_{\allX}^{\mathrm{caus}}$ denote an approximation that is strictly causal. Define \beqa \widetilde{\mathrm{D}}( P_{\allX}^{\mathrm{caus}} \| \Phat_{\allX}^{\mathrm{caus}} )  :=      \E_{P_{\allX}} \left[\log \frac{P_{\allX}^{\mathrm{caus}}}{\widehat{P}_{\allX}^{\mathrm{caus}}} \right]\!. \label{eq:app:strccaus:Dtil}\eeqa  If the expectation were over $P^{\mathrm{caus}}_{\allX}$, then \eqref{eq:app:strccaus:Dtil} would be the KL divergence $\kldist{P_{\allX}^{\mathrm{caus}}}{\Phat_{\allX}^{\mathrm{caus}}}\!.$  The next lemma shows \eqref{eq:app:strccaus:Dtil} nonetheless is non-negative.

\begin{lemma} \label{lem:app:KLtilde_nonneg}
If $P_{\allX}$ is positive, $\widetilde{\mathrm{D}}( P_{\allX}^{\mathrm{caus}} \| \Phat_{\allX}^{\mathrm{caus}} ) \geq 0$.  Equality holds iff $P_{X_i \| \allX=\allx}(\x_i) = \Phat_{X_i \| \allX=\allx}(\x_i)$ for all $\allx \in \calX^{mn}$ and $i \in [m].$
\end{lemma}

\begin{IEEEproof} \vspace{-.3cm}\beqa
\widetilde{\mathrm{D}}( P_{\allX}^{\mathrm{caus}} \| \Phat_{\allX}^{\mathrm{caus}} )  
&=&      \E_{P_{\allX}} \left[\log \frac{P_{\allX}^{\mathrm{caus}}}{\widehat{P}_{\allX}^{\mathrm{caus}}} \right] \nonumber \\
&=&   \sum_{i=1}^m \sum_{t=1}^n   \E_{P_{\allX}} \left[\log \frac{ P_{X_{i,t} | \allX^{t-1}} (X_{i,t} | \allX^{t-1}) }{ \widehat{P}_{X_{i,t} | \allX^{t-1}} (X_{i,t} | \allX^{t-1})} \right] \label{eq:apnd:strctcaus:aa1} \\ 
&=&   \sum_{i=1}^m \sum_{t=1}^n   \E_{P_{\allX^{t-1}}}  \left[ \E_{P_{X_{i,t} | \allX^{t-1}=\allx^{t-1} }} \log \frac{ P_{X_{i,t} | \allX^{t-1}} (X_{i,t} | \allx^{t-1}) }{ \widehat{P}_{X_{i,t} | \allX^{t-1}} (X_{i,t} | \allx^{t-1})} \bigg| \allX^{t-1}=\allx^{t-1} \right] \label{eq:apnd:strctcaus:aa2} \\
&=&   \sum_{i=1}^m \sum_{t=1}^n   \E_{P_{\allX^{t-1}}}  \left[ \kldist{P_{X_{i,t} | \allX^{t-1} = \allx^{t-1}} }{\Phat_{X_{i,t} | \allX^{t-1} = \allx^{t-1}}} \bigg| \allX^{t-1}=\allx^{t-1} \right] \nonumber \\
&=&   \sum_{i=1}^m    \kldist{P_{X_i \| \allX} }{\Phat_{X_i \| \allX} \big| P_{\allX}},  \label{eq:apnd:strctcaus:aa3}
\eeqa %
where \eqref{eq:apnd:strctcaus:aa1} factorizes and uses linearity of expectation, \eqref{eq:apnd:strctcaus:aa2} uses iterated expectation and normalizes, and \eqref{eq:apnd:strctcaus:aa3} uses \eqref{eq:def:condKL}.  The lemma then follows from \eqref{eq:apnd:strctcaus:aa3} using that KL-divergence is non-negative and Lemma~\ref{lemma:iffKLdivergenceZero}.
\end{IEEEproof}

Given that the joint distribution $P_{\allX}$ is not strictly causal, a natural question is how well $\Phat_{\allX}^{\mathrm{caus}} $ approximates $P_{\allX}$.  The following result shows that the divergence $\kldist{P_{\allX}}{\Phat_{\allX}^{\mathrm{caus}}}$ decomposes into the sum of a common penalty for $P_{\allX}$ violating strict causality and a second penalty measuring how well $\Phat_{\allX}^{\mathrm{caus}}$ approximates $P_{\allX}^{\mathrm{caus}}$.

\begin{theorem} \label{thm:inf_decmp:appx} If $P_{\allX}$ is positive,
$
 \kldist{P_{\allX}}{\Phat_{\allX}^{\mathrm{caus}}} = \kldist{P_{\allX} }{P_{\allX}^{\mathrm{caus}}}  %
+ \widetilde{\mathrm{D}}( P_{\allX}^{\mathrm{caus}} \| \Phat_{\allX}^{\mathrm{caus}} ) . $
\end{theorem}  

\begin{IEEEproof} \vspace{-.5cm}
\beqa%
\kldist{P_{\allX}}{\Phat_{\allX}^{\mathrm{caus}}} &=& \E_{P_{\allX}}\! \left[ \log \frac{P_{\allX}(\allX)}{\Phat_{\allX}^{\mathrm{caus}}(\allX)}   \right] \nonumber \\
&=& \E_{P_{\allX}}\! \left[ \log \frac{P_{\allX}(\allX)}{P_{\allX}^{\mathrm{caus}}(\allX)}  \frac{P_{\allX}^{\mathrm{caus}}(\allX)}{\Phat_{\allX}^{\mathrm{caus}}(\allX)} \right]  \label{eq:strctcaus:appx:3}  \\
&=& \kldist{P_{\allX} }{P_{\allX}^{\mathrm{caus}}} + \widetilde{\mathrm{D}}( P_{\allX}^{\mathrm{caus}} \| \Phat_{\allX}^{\mathrm{caus}} ),   \label{eq:strctcaus:appx:6} 
\eeqa %
where \eqref{eq:strctcaus:appx:3} multiplies by one and rearranges,   and \eqref{eq:strctcaus:appx:6} uses \eqref{eq:app:strccaus:Dtil}. 
\end{IEEEproof}

If $P_{\allX}$ is not strictly causal, $\kldist{P_{\allX} }{P_{\allX}^{\mathrm{caus}}}>0$.  Thus, $\kldist{P_{\allX}}{\Phat_{\allX}^{\mathrm{caus}}} >0$.  For this setting, we define minimal generative models according to $\widetilde{\mathrm{D}}( P_{\allX}^{\mathrm{caus}} \| \Phat_{\allX}^{\mathrm{caus}} ) = 0$ instead of \eqref{eq:MGM:6}. 

\begin{theorem}
If $P_{\allX}$ is positive, the parent sets in the directed information graph are the parent sets in $P_{\allX}^{\mathrm{caus}}$.  Algorithms~1--4 correctly recover the directed information graph.  
\end{theorem}
\begin{IEEEproof}
The proof will not repeat the proofs of the algorithms' correctness when $P_{\allX}$ is strictly causal.  When strict causality holds, Theorem~\ref{thm:app:LCeqPC}$\Rightarrow$Corollary~\ref{cor:app:PCisGenMod}$\Rightarrow$\{Theorem~\ref{thm:MGMandDIsame}, Lemma~\ref{lem:caus_data_proc}\}$\Rightarrow$ Algorithms~1, 3, and 4 are correct.    Theorem~\ref{thm:app:LCeqPC} only requires that $P_{\allX}$ be positive, not strictly causal, so it holds here.  Corollary~\ref{cor:app:PCisGenMod} depends on Theorem~\ref{thm:app:LCeqPC} and will hold provided generative models are defined by $\widetilde{\mathrm{D}}( P_{\allX}^{\mathrm{caus}} \| \Phat_{\allX}^{\mathrm{caus}} ) = 0$ instead of \eqref{eq:MGM:6}.  Lemma~\ref{lem:app:KLtilde_nonneg} ensures that equality will only hold for actual generative models (all marginals equivalent).

Theorem~\ref{thm:MGMandDIsame}, stating the directed information graph parent sets are the same as the minimal generative model parent sets, and Lemma~\ref{lem:caus_data_proc} then follow from Corollary~\ref{cor:app:PCisGenMod}.  The proofs for Algorithms~1, 3, and 4 will hold because Lemma~\ref{lem:caus_data_proc} does. Since Theorem~\ref{thm:MGMandDIsame} holds, Algorithm~2 thus also finds the parents sets of $P_{\allX}^{\mathrm{caus}}$. 
\end{IEEEproof}

Now consider problem of general approximations.  Let $\calP^{\mathrm{caus}}_{\allX}$ denote the set of strictly causal distributions.

\begin{corollary} \label{thm:apndx:argminappx} If $P_{\allX}$ is positive,
\beqas
\hspace{-0.6cm} \argmin_{\Phat_{\allX}^{\mathrm{caus}} \in  \calP^{\mathrm{caus}}_{\allX}} \kldist{P_{\allX}}{\Phat_{\allX}^{\mathrm{caus}}} \!\!\!&=&\!\!\! \argmin_{\Phat_{\allX}^{\mathrm{caus}} \in  \calP^{\mathrm{caus}}_{\allX}}  \widetilde{\mathrm{D}}( P_{\allX}^{\mathrm{caus}} \| \Phat_{\allX}^{\mathrm{caus}} ) . \eeqas
\end{corollary}  

\begin{IEEEproof}
The proof is immediate from Theorem~\ref{thm:inf_decmp:appx}  as  $\kldist{P_{\allX} }{P_{\allX}^{\mathrm{caus}}} $ has no dependence on $\Phat_{\allX}^{\mathrm{caus}}$.
\end{IEEEproof}
The importance of Corollary~\ref{thm:apndx:argminappx} is that if a researcher searches for a strictly causal approximation, how good it is depends on how closely it approximates $P_{\allX}^{\mathrm{caus}}$.  From Lemma~\ref{lem:app:KLtilde_nonneg}, only distributions $\Phat_{\allX}^{\mathrm{caus}}$ with correct marginals will minimize $\kldist{P_{\allX}}{\Phat_{\allX}^{\mathrm{caus}}}$.  The best approximation $\Phat_{\allX}^{\mathrm{caus}}$ for $P_{\allX}$ is thus the best for $P_{\allX}^{\mathrm{caus}}$.

We next evaluate $\kldist{P_{\allX} }{P_{\allX}^{\mathrm{caus}}}$ and show a chain rule for all statistical relationships in the network.  We will use 
the total correlation, a generalization of mutual information \cite{watanabe1960information}. For a set of random variables $\{A, B, C\}$, it is 
\beqa
\I(A;B;C) := \kldist{P_{A,B,C}}{P_A P_B P_C}. \label{eq:def:totcor}
\eeqa  
We will use the following notation 
\beqa
\bar{\I}(\X_1; \dots;\! \X_m) \!:=\! \sum_{t=1}^n \I(X_{1,t}; \dots ;\! X_{m,t} |\allX_{[m]}^{t-1}). \label{eq:def:timecondtotcor}
\eeqa  This is the sum over time of the total correlation of all the processes at time $t$, conditioned on the full past.  We first characterize how close $P_{\allX}^{\mathrm{caus}}$ is to $P_{\allX}$.
\begin{lemma} \label{lem:inf_decmp:strct_caus} $\kldist{P_{\allX} }{P_{\allX}^{\mathrm{caus}}} = \bar{\I}(\X_1; \dots; \X_m). $
\end{lemma}

\begin{proof} \vspace{-.7cm}
\beqa
 \kldist{P_{\allX} }{P_{\allX}^{\mathrm{caus}}} 
&=& \sum_{t=1}^n \E_{P_{\allX}} \left[ \log \frac{P_{\allX_{t} | \allX^{t-1}  }(\allX_{t} | \allX^{t-1})}{\prod_{i=1}^m P_{X_{i,t} | \allX^{t-1} } (X_{i,t} | \allX^{t-1}) } \right] \label{eq:info_decomp:1a}\\
&=& \sum_{t=1}^n \I(X_{1,t}; \dots ; X_{m,t} | \allX^{t-1}) \label{eq:info_decomp:2}\\
&=& \bar{\I}(\X_1; \dots; \X_m), \label{eq:info_decomp:3}
\eeqa  where \eqref{eq:info_decomp:1a} plugs in and factorizes over time, \eqref{eq:info_decomp:2} follows from  \eqref{eq:def:totcor},  and \eqref{eq:info_decomp:3} uses the notation \eqref{eq:def:timecondtotcor}. 
\end{proof}

Lastly, we identify a chain rule for the total correlation of the network.  This shows that all of the statistical dependencies between processes decompose into stictly causal and instantaneously correlative components.
\begin{lemma} \label{lem:inf_decmp:main} $\I(\X_1; \dots; \X_m) = \sum_{i=1}^m \I(\allX_{[m]\backslash\{i\}} \to \X_i) + \bar{\I}(\X_1; \dots; \X_m).$
\end{lemma} 

\begin{proof} \label{apdx:prf:lem:inf_decmp:main}\vspace{-.6cm}
\beqa
\I(\X_1; \dots; \X_m) 
%
%
&=&\E_{P_{\allX}} \left[ \log \frac{P_{\allX}(\allX)}{\prod_{i=1}^m P_{\X_i} (\X_i)}   \right] \label{eq:totcor:1} \\
&=&  \E_{P_{\allX}} \left[ \log \frac{P_{\allX}(\allX)}{{\prod_{i=1}^m P_{\X_i} (\X_i) }   }  + \log \frac{P_{\allX}^{\mathrm{caus}}(\allX)}{P_{\allX}^{\mathrm{caus}}(\allX) }  \right]  \label{eq:totcor:2} \\
&=&\E_{P_{\allX}} \left[ \log \frac{P_{\allX}(\allX)}{ P_{\allX}^{\mathrm{caus}}(\allX)   }  \right]  
%
%
+\E_{P_{\allX}} \left[ \log  \frac{\prod_{i=1}^m P_{\X_i \| \allX_{[m]\backslash\{i\}}}(\X_i \| \allX_{[m]\backslash\{i\}})}{\prod_{i=1}^m P_{\X_i}(\X_i)}   \right] \label{eq:totcor:3} \\
&=& \bar{\I}(\X_1; \dots ;\X_m) + \sum_{i=1}^m \I(\allX_{[m]\backslash\{i\}} \to \X_i).  \label{eq:totcor:4}
\eeqa
Eq.~\eqref{eq:totcor:1} follows from the definition \eqref{eq:def:totcor}, \eqref{eq:totcor:2} adds zero inside the expectation, \eqref{eq:totcor:3} rearranges and uses linearity of expectation, and \eqref{eq:totcor:4} uses Lemma~\ref{lem:inf_decmp:strct_caus} and \eqref{eqn:defn:ccDirectedInformation}.
\end{proof}

The chain rule for the two process case, when total correlation is mutual  information $\I(\X;\Y)$, was shown by Marko \cite{marko1973bidirectional} assuming strict causality and by Gourieroux et al. \cite{gourieroux1987kullback} and Solo \cite{solo2008causality} without that assumption.

\section{Proof of Proposition~\ref{prop:DIisGranger}} \label{app:prop:DIisGranger}

Directed information has been proposed as a general measure of Granger causality  \cite{rissanen1987measures, bouissou1986tests, gourieroux1987kullback}.  It is  consistent with ``strong'' Granger causality \cite{chamberlain1982general, florens1982note} which uses conditional independence tests.  Directed information reduces to Geweke's statistic $\log \tilde{\sigma}^2 / \sigma^2$, the common form of linear Granger causality, when the processes are jointly Gaussian \cite{barnett2009granger, amblard2009relating}.  We note that \cite{amblard2012relation}, which came after our preliminary \cite{quinn2011generalized} and a draft of the present paper, also considers a  prediction setting and observes that different setups yield different versions of Granger causality.  In \cite{amblard2012relation}, the result that the conditional mean is the optimal predictor under quadratic loss is used to motivate the conditional dependence based strong Granger causality and thus directed information.

We provide further motivation beyond corresponding to strong Granger causality.  The conditional independence tests are {\em sufficient} in that if $X_{i,t}$ is independent of $X_j^{t-1}$, then $X_j^{t-1}$ should not be help predict $X_{i,t}$.  But they might not be {\em necessary}.  There might be sequential prediction settings, i.e. certain loss functions and prediction spaces, such that even though $X_{i,t}$ is conditionally dependent on $X_j^{t-1}$, $X_j^{t-1}$ does not help predict $X_{i,t}$.  Thus, we do not want strong Granger causality to be the anchor justifying directed information for capturing Granger causality, when strong Granger causality might not capture Granger's statement in certain settings.

Granger's original statement was in terms of how much side information helps in prediction.  We will show that there are many variations of sequential prediction problems with side information and consequently numerous possible formulations of Granger causality.  We then show that directed information is precisely the value of causal side information in a specific, sequential prediction problem.

\newcommand{\pred}{q}
\newcommand{\tpred}{\tilde{q}}
\newcommand{\predSpace}{\mathcal{Q}}
\newcommand{\regret}{R}
\newcommand{\feasibleset}{\mathcal{A}}
\newcommand{\tfeasibleset}{\mathcal{\tilde{A}}}

We will quantify how much the causal side information of $\X_j$ helps in sequentially predicting $\X_i$.  (See \cite{cesa2006prediction} for an overview of sequential prediction.)  Consider two predictors sequentially forming predictions about $\X_i$ in some decision space $\predSpace$, a convex subset of a vector space.    
At time $t$, one predictor knows the full past of all the processes, $\allX^{t-1}$, and specifies a prediction $\pred_t(\allX^{t-1}) \in \predSpace$.   The other predictor knows the past of all the processes except $\X_j$, $\allX^{t-1}_{[m]\backslash \{j\}}$, and specifies a prediction $\tpred_t(\allX^{t-1}_{[m]\backslash \{j\}}) \in \predSpace$.  We will suppress the arguments of $\pred_t$ and $\tpred_t$ for simplicity.  Define the spaces of candidate predictions as $\feasibleset_t \!\!=\!\! \{ \pred_t: \calX^{m(t-1)} \to \predSpace  \} $ and $\tfeasibleset_t \!\!=\!\! \{\tpred_t:  \calX^{(m-1)(t-1)} \to \predSpace \}.$

Subsequently, $X_{i,t}$ is revealed, and a loss function $l:  \predSpace \times \calX \to \reals^+$ assesses the loss $l(p,x)$ for a prediction $p \in \predSpace$ given the outcome $x$.  Thus, one predictor
incurs loss $l(\pred_t,X_{i,t})$ and the other incurs $l(\tpred_t,X_{i,t})$.  The reduction in loss $ r_t(\pred_t,\tpred_t,X_{i,t}  ) := l(\tpred_t,X_{i,t})-l(\pred_t,X_{i,t})$ characterizes how much the side information of $X_j^{t-1}$ helps.  

There are many choices of the decision space $\predSpace$, the loss function $l$, and how to combine the reductions in loss over time.  For instance,  $\predSpace$ could be binary or a probability simplex,  $l$ could be absolute loss or Hellinger loss, and the loss could be combined over the horizon in a discounted or minimax manner \cite{cesa2006prediction}.  For any such sequential prediction setting, the value of the causal side information could be measured as a form of Granger causality.  We next focus on a particular setting for which directed information emerges as the value of Granger causality.

Let the decision space $\predSpace$ be the space of probability measures over $\calX$,
$\predSpace = \{p \in \probSimplex{\calX}\}.$
A natural loss function for probability measures is the logarithmic loss
$ l(\pred,x) =  -\log \pred(x). $ We consider the expected cumulative reduction in loss between the predictions in $\feasibleset_t$ and $\tfeasibleset_t$ respectively whose expected cumulative loss is minimal.  This is analogous to how Granger's test compares the linear models with smallest mean-square error.  By linearity of expectation, we can focus on minimizing instantenous loss \vspace{-0cm}\beqa
\pred^*_t \!\!\!&=&\!\!\! \argmin_{\pred_t \in \feasibleset_t} \E_{P_{X_{i,t}}} \brackets{ l(\pred_t,X_{i,t})}, \label{eqn:optimalPredictor:a} \\
\tpred^*_t \!\!\!&=&\!\!\! \argmin_{\tpred_t \in \tfeasibleset_t} \E_{X_{i,t}} \brackets{ l(\tpred_t,X_{i,t})} \label{eqn:optimalPredictor:b}.
\eeqa

The expected cumulative reduction in loss is
$
\bar{R}(\pred^*,\tpred^*) := \E_{P_{\allX}} \brackets{\sum_{t=1}^n r_t(\pred^*_t,\tpred^*_t,X_{i,t})}. \nonumber
$

We now state our main theorem, showing that the optimal predictors $\pred^*$ and $\tpred^*$ are the true conditional distributions and that the reduction in expected loss is precisely the causally conditioned directed information.
\begin{theorem} \label{thm:log_loss_thm_dir_info_zt}
The optimal solutions to \eqref{eqn:optimalPredictor:a} and \eqref{eqn:optimalPredictor:b} are given by \beqas
\pred^*_t(x_{i,t}) \!\!\!&=&\!\!\! P_{X_{i,t} | \allX_{[m]}^{t-1}}(x_{i,t}| \allX_{[m]}^{t-1}) \\ 
\tpred^*_t(x_{i,t}) \!\!\!&=&\!\!\! P_{X_{i,t} | \allX_{[m] \backslash \{j\}}^{t-1}}(x_{i,t}| \allX_{[m] \backslash \{j\}}^{t-1}), 
\eeqas
where we continue to suppress the argument of the past $\allX_{[m]}^{t-1}$ in $\pred^*_t$ and $\allX_{[m] \backslash \{j\}}^{t-1}$ for $\tpred^*_t$. The expected cumulative reduction in loss is given by the causally conditioned directed information
$
\bar{R}(\pred^*,\tpred^*)  =   \I \left( \X_j \to \X_i \| \allX_{\setmi{m}{i,j}} \right).
$
\end{theorem}
\begin{IEEEproof}
Note that
\beqa
\pred^*_t \!\!\!&=& \!\!\!\!\argmin_{\pred_t \in \feasibleset_t} \ \E_{P_{X_{i,t}}} \brackets{-\log \pred_t(X_{i,t})  } \nonumber \\
          \!\!\!&=& \!\!\!\! \argmin_{\pred_t \in \feasibleset_t} \ \E_{P_{X_{i,t}}} \brackets{-\log  P_{X_{i,t} | \allX^{t-1}}(X_{i,t}| \allX^{t-1})  +\log \frac{ P_{X_{i,t} | \allX^{t-1}}(X_{i,t}| \allX^{t-1}) }{\pred_t(X_{i,t})} }  \label{proof:eqn:optimalPredictor:a:b}  \\
          \!\!\!&=& \!\!\!\!\argmin_{\pred_t \in \feasibleset_t} \  \kldist{P_{X_{i,t} | \allX^{t-1}}}{\pred_t  } \label{proof:eqn:optimalPredictor:a:c}
\eeqa
where \eqref{proof:eqn:optimalPredictor:a:b} multiplies by one inside the $\log$ and \eqref{proof:eqn:optimalPredictor:a:c} follows from the definition of divergence and that the left-hand term in the expectation does not effect the $\argmin$. From the non-negativity of the KL divergence,
$\pred^*_t(x)=P_{X_{i,t} | \allX^{t-1}}(x|\allX^{t-1})$.  Similarly, $\tpred^*_t(x)=P_{X_{i,t} | \allX_{[m] \backslash\{j\}}^{t-1}}(x|\allX_{[m] \backslash\{j\}}^{t-1})$.

We now discuss using Granger's notion of ``better'' to address the two predictors.  The reduction in loss becomes a log-likelihood ratio
\beqas
 r_t(\pred^*_t,\tpred^*_t,X_{i,t}) = \log \frac{\pred^*_t(X_{i,t})}{\tpred^*_t(X_{i,t})} = \log \frac{P_{X_{i,t} | \allX^{t-1}}(X_{i,t}| \allX^{t-1})}{P_{X_{i,t} | \allX_{[m] \backslash\{j\}}^{t-1}} (X_{i,t} | \allX_{[m] \backslash\{j\}}^{t-1}) }.
\eeqas
Thus,\vspace{-.4cm}
\beqas
\bar{R}(\pred^*,\tpred^*) &=& \E_{P_{\allX}} \brackets{\sum_{t=1}^n r_t(\pred^*_t,\tpred^*_t,X_{i,t})} \nonumber \\
                          &=& \E_{P_{\allX}} \brackets{\sum_{t=1}^n \log \frac{P_{X_{i,t} | \allX^{t-1}}(X_{i,t}| \allX^{t-1})}{P_{X_{i,t} | \allX_{[m] \backslash\{j\}}^{t-1}} (X_{i,t} | \allX_{[m] \backslash\{j\}}^{t-1}) }} \\
                          &=&   \I \left( \X_j\to \X_i \| \allX_{\setmi{m}{i,j}} \right). \vspace{-1cm}\eeqas 
\end{IEEEproof}

Theorem~\ref{thm:log_loss_thm_dir_info_zt} states that in sequentially predicting $X_{i,t}$, the expected cumulative reduction in loss due to the causal side information $X_j^{t-1}$ is precisely the directed information when the predictors are probability measures and the loss is the logarithmic loss.  Thus, in this setting, we can interpret the value of directed information as quantifying the ``strength'' of the influence in the reduction in bits.  In the preliminary work \cite{quinn2011generalized}, we explore other sequential prediction settings, such as minimax, where the value of Granger's statement is a different quantity.

A recent work \cite{jiao2014justification} demonstrated that for the above sequential prediction problem, the log loss is not ``special.''  If a data processing axiom holds and the alphabet is non-binary, $|\calX|>2$, with any other loss function $l$ the expected reduction in loss is proportional to the directed information $\bar{R}(\pred^*,\tpred^*)\propto \I \left( \X_j\to \X_i \| \allX_{\setmi{m}{i,j}} \right).$ Thus, for any such loss function $l$, only when $\X_i$ is causally conditionally independent of $\X_j$ given $\X_{[m] \backslash \{i,j\}}$ will $\bar{R}(\pred^*,\tpred^*)=0.$

\section{Proof of Theorem~\ref{thm:MGMandDIsame}} \label{app:prf:MGMandDIsame}

Before proving Theorem~\ref{thm:MGMandDIsame}, first consider Definition~2.2 of \cite{eichler2012graphical}.  For any graph $G$ and distribution $P_{\allX}$   the ``pairwise Granger-causal Markov property'' (PC) and the ``local Granger-causal Markov property'' (LC) are defined   
\beqa
&& \text{(PC): for all $i,j \in [m]$, if the edge $\X_j \to \X_i$ is not in $G$, } \I(\X_j \to \X_i \| \allX_{[m]\backslash \{i\}})  = 0, \nonumber\\
&& \text{(LC): for all $i \in [m]$, let $A'(i)$ denote parent indices for $\X_i$ in $G$, then } \I(\X_{[m]\backslash \{A'(i) \cup \{i\}\}} \to \X_i \| \allX_{A'(i) })  = 0. \nonumber
\eeqa

A distribution $P_{\allX}$ ``satisfies'' the pairwise property with respect to $G$ if (PC) holds, and likewise for (LC). 

\begin{theorem}[\!\!\cite{eichler2012graphical}] \label{thm:app:LCeqPC}
Under Assumption~\ref{def:spat_cond_ind}, for any graph $G$ and distribution $P_{\allX}$, (PC)$\Leftrightarrow$(LC).
\end{theorem}

\begin{corollary} \label{cor:app:PCisGenMod} Under Assumption~\ref{def:spat_cond_ind}, for a graph $G$ with parent set function $A'$ and for a distribution $P_{\allX}$, (PC) $\Leftrightarrow$ (LC) $\Leftrightarrow$ $A'$ is a generative model, satisfying \eqref{eq:MGM:6}.
\end{corollary} 
\begin{IEEEproof}
We show (LC) $\Leftrightarrow$ $A'$ is a generative model, then invoke Theorem~\ref{thm:app:LCeqPC}.  For all $i \in [m]$, (LC) means $\I(\X_{[m]\backslash \{A'(i) \cup \{i\}\}} \to \X_i \| \allX_{A'(i) })  = 0.$  We can extend the definition of causally conditioned directed information \eqref{eqn:defn:ccDirectedInformation} with $[m]\backslash \{A'(i) \cup \{i\}\}$ in place of $j$ to obtain \eqref{eq:MGM:a1}, which in turn means $A'$ is a generative model.  The steps are reversible so by Theorem~\ref{thm:app:LCeqPC},  (PC) $\Longleftrightarrow$ (LC) $\Longleftrightarrow$ $A'$ is a generative model. %
\end{IEEEproof}

\newcommand{\wtA}{\widetilde{A}}

We now can prove Theorem~\ref{thm:MGMandDIsame}.  
\begin{IEEEproof}
Let $A'$ denote the parent set function induced by the directed information graph $G$.  By construction, $G$ satisfies (PC) and so by Corollary~\ref{cor:app:PCisGenMod}, the parent set function $A'$  is a generative model. 

We now show that the induced $A'$ is minimal and unique.  Let $\wtA$ be a minimal generative model for $P_{\allX}$.  We first show by contradiction that for all $i \in [m]$, $A'(i) \subseteq \wtA(i)$.  Suppose not, so for some $i$, $A'(i) \backslash \wtA(i)$ is non-empty. Let $j \in A'(i) \backslash \wtA(i)$.  Since $j \notin \wtA(i)$ and $\wtA$ is a generative model, by Corollary~\ref{cor:app:PCisGenMod} (PC) holds, so $\I(\X_j \to \X_i \| \allX_{[m]\backslash\{j\}}) = 0$.  However, $j \in A'(i)$, equivalently the edge $\X_j \to \X_i$ present in $G$, means by Definition~\ref{def:dir_info_grph} that $\I(\X_j \to \X_i \| \allX_{[m]\backslash\{j\}}) > 0$, a contradiction. Since for all $i \in [m]$, $A'(i) \subseteq \wtA(i)$, and $\wtA$ is a minimal generative model by construction, $|\wtA(i)| \leq |A'(i)|$ for all $i$, then $A' = \wtA$.  As $\wtA$ was an arbitrary minimal generative model, $A'$ must be the unique, minimal generative model for $P_{\allX}$. %
\end{IEEEproof}


\section{Proof of Lemma~\ref{lem:caus_data_proc}} \label{app:prf:caus_Mark_chain}

\begin{IEEEproof}  We first show the inequality \eqref{eq:lem:data_proc:gen} in Lemma~\ref{lem:caus_data_proc}.   Combining \eqref{eq:MGM:a1} in the definition of minimal generative models and directed information \eqref{eqn:defn:ccDirectedInformation} yields 
\beqa
0 \!\!&=&\!\! \I (\allX_{[m]   } \to \X_i  \|\allX_{A(i)}) \nonumber \\
&=&\!\!\!\! \I (\allX_{ B(i) } \!\to\! \X_i \|\allX_{A(i)}) + \I (\allX_{W(i)  } \!\to\! \X_i \|\allX_{B(i)}) 
%
%
+ \I (\allX_{[m] \backslash \{ W(i) \cup B(i) \} } \!\to\! \X_i \| \allX_{W(i) \cup B(i) }) \label{eq:prfMrkChn:a1} 
\eeqa 
where \eqref{eq:prfMrkChn:a1} uses the chain rule with $A(i) \subseteq B(i)$.  Since directed information is non-negative, each term in \eqref{eq:prfMrkChn:a1} must be zero.
Using the chain rule again,
\beqa
\I (\allX_{W(i) \cup B(i) } \to \X_i ) %
&=& \I (\allX_{W(i)} \to \X_i ) + \I (\allX_{B(i) } \to \X_i \|\allX_{W(i)}) \label{eq:prfMrkChn:a3} \\ 
&=&  \I (\allX_{B(i)} \to \X_i ) + \I (\allX_{W(i) } \to \X_i \|\allX_{B(i)}) \label{eq:prfMrkChn:a5} \\ 
&=& \I (\allX_{B(i)} \to \X_i )  \label{eq:prfMrkChn:a6}  
\eeqa where \eqref{eq:prfMrkChn:a3} and \eqref{eq:prfMrkChn:a5} apply the chain rule in different ways and \eqref{eq:prfMrkChn:a6}  uses that $\I (\allX_{W(i)  } \!\to\! \X_i \|\allX_{B(i)})$ must be zero from \eqref{eq:prfMrkChn:a1}.  Consequently, \eqref{eq:prfMrkChn:a3} and \eqref{eq:prfMrkChn:a6} imply 
$\I( \allX_{W(i)} \to \bX_i  ) \leq  \I( \allX_{B(i)} \to \bX_i  ). $  This is \eqref{eq:lem:data_proc:gen} from Lemma~\ref{lem:caus_data_proc}.

We show that equality only occurs when $A(i) \subseteq W(i)$.  The proof is by contradiction.  Suppose equality holds but $A(i) \not \subseteq W(i)$.    From \eqref{eq:prfMrkChn:a3} and \eqref{eq:prfMrkChn:a6}, equality occurs when $ \I( \allX_{B(i)} \to \bX_i \| \allX_{W(i) } ) = 0 $.  Using the chain rule, %
\beqa
0 \!\!\!\!&=&\!\!\!\! \I( \allX_{B(i)} \to \bX_i \| \allX_{W(i)} ) \nonumber \\
&=&\!\!\!\! \I( \allX_{B(i)} \to \bX_i \| \allX_{W(i) } ) + \I (\allX_{[m] \backslash \{ W(i) \cup B(i) \} } \!\to\! \X_i \| \allX_{W(i) \cup B(i) }) \label{eq:prfMrkChn:c1}  \\
&=&\!\!\!\! \I (\allX_{[m] \backslash W(i)   } \to \X_i  \|\allX_{W(i)}) \label{eq:prfMrkChn:c2} 
\eeqa where \eqref{eq:prfMrkChn:c1} adds a zero term from \eqref{eq:prfMrkChn:a1} and \eqref{eq:prfMrkChn:c2} uses the chain rule in reverse.  By Corollary~\ref{cor:app:PCisGenMod}, \eqref{eq:prfMrkChn:c2} implies that $W(i)$ is a valid generative model parent set.  Since $A(i)$ is the minimal generative model parent set,  applying that (LC)$\Rightarrow$(PC) for both $W(i)$ and $A(i)$, then for all $j \in [m] \backslash (W(i) \cap A(i))$, $\I(\X_j \to \X_i \| \allX_{[m]\backslash\{j\}}) = 0$.  Corollary~\ref{cor:app:PCisGenMod} implies $W(i) \cap A(i)$ is also a generative model parent set.  However, $|W(i) \cap A(i)| < |A(i)|$ and $|A(i)|$ was minimal by construction, a contradiction. Thus equality only occurs when $A(i) \subseteq W(i).$ %
\end{IEEEproof}

\section{Proof of Theorem~\ref{thm:alg1}} \label{prf:thm:alg1}

\begin{IEEEproof}
We first show by recursion that for each $i\in [m]$, the parent set $A(i)$ returned by Algorithm~1 contains the true set $A^*(i)$.  Suppose line~6 has been called $r$ times, and let $A_r(i)$ denote the current parent set.  The inductive hypothesis is $A^*(i) \subseteq A_r(i)$.  The base case is $A_0(i) = [m]\backslash \{i\}$ and the hypothesis trivially holds.

Let $r\geq 0$ and suppose the inductive hypothesis is true.  If line~6 is called again for the $(r+1)$th time, then for the current $k$ and $B(i)$, with $\{k\} \cup B(i) = A_r(i)$, $\I(\X_k \to \X_i \| \allX_{B(i)}) = 0$.  By the chain rule, \beqa
0 &=& \I(\X_k \to \X_i \| \allX_{B(i)}) \nonumber \\
&=& \I(\X_{\{k\} \cup B(i)} \to \X_i ) - \I(\X_{ B(i)} \to \X_i ) \nonumber \\
\Longrightarrow \I(\X_{ B(i)} \to \X_i ) &=&  \I(\X_{\{k\} \cup B(i)} \to \X_i ). \label{prf:alg1:1} 
\eeqa  Since $\{k\} \cup B(i) = A_r(i)$ contains the full parent set $A^*(i)$ by the inductive hypothesis, by Lemma~\ref{lem:caus_data_proc} \eqref{prf:alg1:1} only holds if $B(i)$ also contains the full parent set $A^*(i)$.  This concludes the first part of the proof.

We prove that the inferred set $A(i)$ only contains the true $A^*(i)$.  For each $k \not \in A^*(i)$, when line~5 is evaluated for that $k$, the current $B(i)$ will contain $A^*(i)$.  By Lemma~\ref{lem:caus_data_proc} \eqref{prf:alg1:1} holds, so line~5 will evaluate as true and $k$ will be removed.  The {\bf for} loop in line~3 is over all  $\X_k$ with $k \in [m]\backslash \{i\}$, so all non-parents will be removed. \end{IEEEproof}

\section{Example of a Network Requiring High-Dimensional Statistics for Recovery} \label{ex:alg3:xor}

Algorithms~3~and~4 recover the graph using directed informations of up to $K+2$ and $K+1$ processes respectively, where $K$ is the size of the largest parent set.  The following example shows that in general, no algorithm can recover the graph only using directed informations involving $K$ processes or less.

\begin{example} \label{ex:noisyXOR}
Let $\bW$, $\bX$, $\bY$, and $\bZ$ be four processes, with $\bW$, $\bX$, and $\bY$ independent processes, each i.i.d. Bernoulli($\frac{1}{2}$).  Let $\oplus$ denote the exclusive-or, which sums its arguments modulo-2 (i.e. $0 \oplus 1 = 1$, $1 \oplus 1 = 0$).  Let $Z_t = W_{t-1} \oplus X_{t-1} \oplus Y_{t-1} +N_t$ for i.i.d. Gaussian noise $N_t$.  Note that $W_{t-1} \oplus X_{t-1} \oplus Y_{t-1}$ is a Bernoulli($\frac{1}{2}$) variable whether conditioned on zero, any one, or any two of its arguments.  Thus, $\I(\bW \to \Z) = \I(\X \to \Z) = \I(\Y \to \Z) = 0$ and $\I(\bW, \X \to \Z) = \I(\bW, \Y \to \Z) = \I(\X, \Y \to \Z) = 0,$  even though $\I(\bW, \X, \Y \to \Z) > 0$.  Thus, any algorithm only using directed informations of up to $K=3$ processes could not distinguish whether $\Z$ had three or zero parents. 
\end{example}

\section{Proof of Theorem~\ref{thm:alg3}} \label{prf:thm:alg3}

\begin{IEEEproof}
Let $A(i)$ and $A^*(i)$ denote the returned and true parent sets for $\X_i$, respectively.  We first prove by contradiction that for each process $\X_i$, no true parent is removed, so $A^*(i) \subseteq A(i)$.  Consider any $i \in [m]$ and assume that some parent $j \in A^*(i)$ is removed in line~9.  For the corresponding set $B$ such that line~8 evaluated as true, $B \subseteq [m] \backslash \{i\}$ necessarily.  But Assumption~2 then implies that if $ \I(\X_j \to \X_i \|   \allX_{B\backslash \{j\}} ) =0 $ then  $\I(\X_j \to \X_i \|   \allX_{[m]\backslash \{i,j\}} ) =0 $, which contradicts $j \in A^*(i)$.  Thus, $A^*(i) \subseteq A(i)$.

We next show that $A^*(i) = A(i)$.  From the above, $A^*(i) \subseteq A(i)$, so $|A^*(i)| \leq |A(i)|$ always holds.  While $K < |A^*(i)|$, $K+1 \leq |A^*(i)|\leq |A(i)|$, so line~4 will evaluate as true and lines~5--11 will execute.  We consider two cases.  First, if line~4 evaluates as false when $K=|A^*(i)|$, then $|A^*(i)|+1=K+1>|A(i)|\geq|A^*(i)|$, which implies that $|A(i)|=|A^*(i)|$.  With $A^*(i) \subseteq A(i)$, that implies $A^*(i) = A(i)$.

The second case is if line~4 evaluates as true when $K=|A^*(i)|$.  Then consider any non-parent $j \in A(i) \backslash A^*(i)$.  Since $K=|A^*(i)|$, $A^*(i) \in \mathcal{B}$, when line~7 sets $B=A^*(i)$, line~8 will find
\beqa
\I(\X_j \to \X_i \|   \allX_{B} ) &=& \I(\X_j \to \X_i \|   \allX_{A^*(i)} ) \nonumber \\
&=& \I(\allX_{A^*(i) \cup \{j\}} \to \X_i ) - \I(\allX_{A^*(i) } \to \X_i ) = 0, \label{eq:prf:alg3:1}
\eeqa
where \eqref{eq:prf:alg3:1} follows from the chain rule and Lemma~\ref{lem:caus_data_proc}.
\end{IEEEproof}

\section{Proof of Theorem~\ref{thm:genstructmain}} \label{prf:thm:alg4}

\begin{IEEEproof}
First consider the case when the bounds are tight, $K(i) = |A^*(i)|$. By Lemma~\ref{lem:caus_data_proc}, only the $B \in \mathcal{B}$ containing the parent set, $A^*(i)\in B$, will have maximal value, so $\mathcal{B} = \{ A^*(i) \}.$

Next consider the case when the bounds are loose, $K(i) > |A^*(i)|$. Then by Lemma~\ref{lem:caus_data_proc}, only the $B \in \mathcal{B}$ with $ A^*(i) \subseteq B$  will have maximal value, so $A^*(i) \subseteq \hspace{-0.2cm}  \bigcap \limits_{B \in \calB_{\mathrm{max}}} \hspace{-0.4cm} B.$  We next show that only $A^*(i)$ is in the intersection.

Let $B_j \in \mathcal{B}$ denote any set containing $A^*(i)$ and a non-parent $j \in [m] \backslash A^*(i) \backslash \{i\}$.  Since $|B_j| = K(i) < m-1$, there is at least one non-parent $j'$ not in $B$.  By construction (line~3), there is another set $B_{j'} \in \mathcal{B}$ that is the same except $j$ and $j'$ are swapped, $B_{j'} = \{j'\} \cup B_j \backslash \{j\}$.  By Lemma~\ref{lem:caus_data_proc} both sets have maximal influence, so $B_j, B_{j'} \in \calB_{\mathrm{max}}$ and thus neither $j$ nor $j'$ appear in the intersection, $j, j' \not \in \hspace{-0.2cm}  \bigcap \limits_{B \in \calB_{\mathrm{max}}} \hspace{-0.4cm} B$.  Therefore $A^*(i) = \hspace{-0.2cm} \bigcap \limits_{B \in \calB_{\mathrm{max}}} \hspace{-0.4cm} B.$  
\end{IEEEproof}

\section{Proof of Theorem~\ref{thm:emp_est_smpl}} \label{app:prf:emp_est_smpl}

\begin{proof} We  obtain concentrations on the empirical distribution from the Hoeffding and union bounds and then use an $L_1$ bound on entropy to translate concentrations on entropies to ones on the directed information estimates.  

Recall that $\epsilon$ measures the error in the empirical probability estimates, $\delta$ measures the error in the directed information estimates, and $B_{\delta}$ is the event that all of the directed information estimates have error at most $\delta$.  We require the following concentrations on the empirical probability distributions.  For every pair $(i,j)$, for every possible realization $\{x_{j}^{l}, x_{i}^{l+1} \} \in \calX^{2l+1} $, and, for a given $\epsilon>0$ which we will later fix as a function of $\delta$, \begin{equation}\begin{aligned} \label{eq:nonpar:prob_conc}
|\widehat{P}_{X_{j}^l, X_i^{l+1}}(x_{j}^l, x_i^{l+1}) - P_{X_{j}^l, X_i^{l+1}}(x_{j}^l, x_i^{l+1}) | &<&\!\!\!\! \epsilon\phantom{.} \\
|\widehat{P}_{X_{j}^l, X_i^{l}}(x_{j}^l, x_i^{l}) - P_{X_{j}^l, X_i^{l}}(x_{j}^l, x_i^{l}) | &<&\!\!\!\! \epsilon \phantom{.}\\
|\widehat{P}_{ X_i^{l+1}}( x_i^{l+1}) - P_{ X_i^{l+1}}( x_i^{l+1}) | &<&\!\!\!\! \epsilon\phantom{.} \\
|\widehat{P}_{ X_i^{l}}( x_i^{l}) - P_{ X_i^{l}}( x_i^{l}) | &<&\!\!\! \epsilon .
\end{aligned}\end{equation}  

From a Hoeffding inequality for Markov chains \cite{glynn2002hoeffding}, for any $(i,j)$ and realization $\{x_{j}^{l}, x_{i}^{l+1}\} \in \calX^{(2l+1)}$,  
\beqa
\prob{ \abs{\widehat{P}_{X_{j}^l, X_i^{l+1}}(x_{j}^l, x_i^{l+1}) - P_{X_{j}^l, X_i^{l+1}}(x_{j}^l, x_i^{l+1}) } \geq  \epsilon}   \leq 2 \exp\left( - \frac{(n \epsilon - 2 d/\lambda)^2}{2 n d^2 / \lambda^2} \right). \label{eq:hoeff}
\eeqa

Applying the union bound to \eqref{eq:hoeff}, the four inequalities in \eqref{eq:nonpar:prob_conc} hold for each of the $|\calX|^{2l+1} $ realizations for each ordered pair of processes $\{ (\X_i, \X_j) \}_{i,j\in [m]}$ with probability $\rho$, given in \eqref{eq:def:rho}.

We next find the value of $\epsilon$ that corresponds to the event $B_{\delta}$.  For simplicity, denote $\{ X_{j}^{l}, X_{i}^{l+1}\}$ by $\buZ$.  We want a concentration on $|\widehat{H}(\buZ) - H(\buZ)|$. Using \eqref{eq:nonpar:prob_conc}, the $L_1$ norm evaluates as 
\beqa
  \| \widehat{P}_{\buZ} - P_{\buZ} \|_{1}  \hspace{-0.3cm} &:=& \hspace{-0.3cm} \sum_{\buz 
	} \hspace{-0.0cm}  | \widehat{P}_{\buZ}(\buz) - P_{\buZ}(\buz) | \leq |\calX|^{2l+1} \epsilon, \label{eq:L1probdistdef}
\eeqa 

Using an $L_1$ bound on entropy, if $\| \widehat{P}_{\buZ} - P_{\buZ} \|_{1} \leq \frac{1}{2}$, then
\beqa
|\widehat{H}(\buZ) - H(\buZ)|  \leq - \| \widehat{P}_{\buZ} - P_{\buZ} \|_{1}
  \log \frac{\| \widehat{P}_{\buZ} - P_{\buZ} \|_{1} }{|\calX|^{2l+1}}. \label{eq:L1bndentr}
\eeqa

The bound is of the form $- b \log \frac{b}{c}$, which is concave in $b$ and maximized at $b = c/e$.  With $\epsilon \leq 1/e$, the upper bound in \eqref{eq:L1probdistdef}, $|\calX|^{2l+1} \epsilon$, is in the interval $(0, |\calX|^{2l+1} / e]$ where the bound \eqref{eq:L1bndentr} is increasing.  Thus, \eqref{eq:L1bndentr} can be bounded using \eqref{eq:L1probdistdef} \vspace{-.5cm}
\beqa
\hspace{-0.0cm} |\widehat{H}(\buZ) - H(\buZ)| \!\!\!&\leq&\!\!\!- |\calX|^{2l+1} \epsilon  \log \frac{|\calX|^{2l+1} \epsilon }{|\calX|^{2l+1}} \nonumber \\
&=&\!\!\! - |\calX|^{2l+1} \epsilon  \log \epsilon. \label{eq:prf:nonpar:f2}
\eeqa

Note that the directed information \eqref{eq:didef:3} decomposes into a linear combination of entropies,
\beqa && \hspace{-1.5cm}\I(X_{i, l+1} ; X_{j}^l | X_{i}^{l} ) =  H (X_i^{l+1} ) - H (X_i^l) - H (X_{i}^{l+1} , X_{j}^{ l})  + H (X_i^l, X_j^l). \label{eq:DIexpandentr}
\eeqa  Applying the triangle inequality to \eqref{eq:DIexpandentr} with \eqref{eq:prf:nonpar:f2} gives that for all $m(m-1)$ ordered pairs $(i, j)$,
$
 \!\! |\widehat{\I}(\X_{j} \to \X_{i}) - \I(\X_{j} \to \X_{i})|  \leq - 4 |\calX|^{2l+1} \epsilon  \log \epsilon.
$

Setting $\delta = - 4 |\calX|^{2l+1} \epsilon  \log \epsilon$ would conclude the proof.  However, to obtain an analytic expression for how $\epsilon$ depends on $\delta$, we bound $\epsilon  \log \epsilon$ with a polynomial expression.  The function $- \epsilon \log \epsilon$ has a maximum value of $1/e$ on the interval $\epsilon \in (0,1)$  attained at $\epsilon =  1/e$.  For $0 < a < 1$,
$-\epsilon \log \epsilon = \frac{1}{a}  (- \epsilon^a \log \epsilon^a) \epsilon^{1-a} \leq \frac{1}{ae} \epsilon^{1-a}. $  For large $\epsilon$, the bound with larger $a$ is tighter; for small $\epsilon$, the bound with small $a$ is tighter.  For all $0<a<1$ and all $(i,j)$,
$|\widehat{\I}(\X_{j} \to \X_{i}) - \I(\X_{j} \to \X_{i})|  \leq\frac{4 |\calX|^{2l+1}}{ae} \epsilon^{1-a}.
$  Setting $\epsilon=\left(\frac{ae\delta}{4 |\calX|^{2l+1}} \right)^{\frac{1}{1-a}} $ finishes the proof that $\prob{ B_{\delta} } \geq 1 - \rho$.

Note that for a fixed probability of error $\rho$ \eqref{eq:def:rho}, fixed $m$, and sufficiently large $n\epsilon$, that as $n$ increases, $\epsilon$ decays as $n^{-1/2}$ which implies that $\delta = \mathcal{O}(n^{-1/2 + \epsilon'})$ for all $\epsilon'>0$.  Alternatively, if $m$ is increasing, to maintain a fixed probability of error $\rho$ with a fixed $\delta$, $n$ needs to increase as $\log m$. \end{proof}

\section{Proof of Lemma~\ref{lem:asym_norm_theta}} \label{app:prf:lem:asym_norm_theta}

The main result follows from \cite{ling2010general}.  Several conditions are first checked here.  Define the Fisher information matrix \beqas G_t(\utheta ') &=& \left[  \frac{\partial L_t(\utheta) }{\partial \theta_{q_1} } \frac{\partial L_t(\utheta) }{\partial \theta_{q_2} }  \bigg|_{\utheta=\utheta'} \right]_{1 \leq q_1,q_2 \leq Q} .
\eeqas

\begin{lemma} \label{lem:est:par:HisI}
Under Assumption~\ref{asmp:param}, for a finite alphabet $\calX$, $\E[G_t(\utheta )] = \E[A_t(\utheta )]$.
\end{lemma}
\begin{proof} Consider any $1\leq q_i, q_j \leq Q$.  We show the $(q_i, q_j)$th elements 
are identical.  First consider $\E[G_t(\utheta' )_{(q_i,q_j)}]$.
\beqa
\E[G_t(\utheta ')_{(q_i,q_j)}] \!\!\!\!&=&\!\!\!\! \E\left[  \frac{\partial L_t(\utheta) }{\partial \theta_{q_1} } \frac{\partial L_t(\utheta) }{\partial \theta_{q_2} }  \bigg|_{\utheta=\utheta'} \right]  \nonumber \\
&=&\!\!\!\! \E\left[  \frac{\partial \log P_{\allX_t | \allX^{t-1}_{t-l};\utheta}(\allX_t | \allX^{t-1}_{t-l}) }{\partial \theta_{q_1} } \frac{\partial \log P_{\allX_t | \allX^{t-1}_{t-l};\utheta}(\allX_t | \allX^{t-1}_{t-l}) }{\partial \theta_{q_2} } \ \Bigg|_{\utheta=\utheta'} \right] \nonumber \\
&=&\!\!\!\! \E\left[  P_{\allX_t | \allX^{t-1}_{t-l};\utheta}(\allX_t | \allX^{t-1}_{t-l})^{-2}   \frac{\partial  P_{\allX_t | \allX^{t-1}_{t-l};\utheta}(\allX_t | \allX^{t-1}_{t-l}) }{\partial \theta_{q_1} } \frac{\partial  P_{\allX_t | \allX^{t-1}_{t-l};\utheta}(\allX_t | \allX^{t-1}_{t-l}) }{\partial \theta_{q_2}}  \Bigg|_{\utheta=\utheta'}  \right] \label{eq:prf:HisI:a3}
\eeqa where \eqref{eq:prf:HisI:a3} differentiates both logarithms and collects common terms. %
Now consider $\E[A_t(\utheta' )_{(q_i,q_j)}]$,
\beqa
\E[A_t(\utheta ')_{(q_i,q_j)}] \!\!\!\! &=&\!\!\!\! \E \left[  - \frac{\partial^2 L_t(\utheta) }{\partial \theta_{q_1} \partial \theta_{q_2} } \Bigg|_{\utheta=\utheta'} \right] \nonumber \\
&=&\!\!\!\! \E \left[  - (-1) P_{\allX_t | \allX^{t-1}_{t-l};\utheta}(\allX_t | \allX^{t-1}_{t-l})^{-2}   \frac{\partial P_{\allX_t | \allX^{t-1}_{t-l};\utheta}(\allX_t | \allX^{t-1}_{t-l})  }{\partial \theta_{q_1}  }    \frac{\partial P_{\allX_t | \allX^{t-1}_{t-l};\utheta}(\allX_t | \allX^{t-1}_{t-l}) }{ \partial \theta_{q_2} }    \right. \nonumber \\
&& \left. - P_{\allX_t | \allX^{t-1}_{t-l};\utheta}(\allX_t | \allX^{t-1}_{t-l})^{-1} \frac{\partial^2 P_{\allX_t | \allX^{t-1}_{t-l};\utheta}(\allX_t | \allX^{t-1}_{t-l}) }{ \partial \theta_{q_1}\partial \theta_{q_2} }   \Bigg|_{\utheta=\utheta'}  \right]  \label{eq:prfGisA:Aval2} \\
&=&\!\!\!\! \E[G_t(\utheta ')_{(q_i,q_j)}] - \E \left[P_{\allX_t | \allX^{t-1}_{t-l};\utheta}(\allX_t | \allX^{t-1}_{t-l})^{-1} \frac{\partial^2 P_{\allX_t | \allX^{t-1}_{t-l};\utheta}(\allX_t | \allX^{t-1}_{t-l}) }{ \partial \theta_{q_1}\partial \theta_{q_2} }   \Bigg|_{\utheta=\utheta'}  \right]  \label{eq:prf:HisI:5}\vspace{-1cm}
\eeqa
where \eqref{eq:prfGisA:Aval2} differentiates and \eqref{eq:prf:HisI:5} uses \eqref{eq:prf:HisI:a3}.  The second term in \eqref{eq:prf:HisI:5} evaluates as 

\beqa
&& \E \left[P_{\allX_t | \allX^{t-1}_{t-l};\utheta}(\allX_t | \allX^{t-1}_{t-l})^{-1} \frac{\partial^2 P_{\allX_t | \allX^{t-1}_{t-l};\utheta}(\allX_t | \allX^{t-1}_{t-l}) }{ \partial \theta_{q_1}\partial \theta_{q_2} }   \Bigg|_{\utheta=\utheta'}  \right] \nonumber \\
&& \hspace{0.5cm} =  \sum_{\allx^t }  \prod_{t'=1}^t P_{\allX_{t'} | \allX^{t'-1}_{t'-l};\utheta}(\allx_{t'} | \allx^{t'-1}_{t'-l})  
\left[ P_{\allX_t | \allX^{t-1}_{t-l};\utheta}(\allx_t | \allx^{t-1}_{t-l})^{-1} 
                  \frac{\partial^2 P_{\allX_t | \allX^{t-1}_{t-l};\utheta}(\allx_t | \allx^{t-1}_{t-l}) }{ \partial \theta_{q_1}\partial \theta_{q_2} }  
  \right] \Bigg|_{\utheta=\utheta'} 
\label{eq:prf:HisI:b1} \\
&& \hspace{0.5cm} =  \sum_{\allx^{t-1} } \prod_{t'=1}^{t-1} P_{\allX_{t'} | \allX^{t'-1}_{t'-l};\utheta}(\allx_{t'} | \allx^{t'-1}_{t'-l})   \left[  \sum_{\allx_t } %
\frac{\partial^2 P_{\allX_t | \allX^{t-1}_{t-l};\utheta}(\allx_t | \allx^{t-1}_{t-l}) }{ \partial \theta_{q_1}\partial \theta_{q_2} }   				 \right] \Bigg|_{\utheta=\utheta'}  \label{eq:prf:HisI:b3} \\
&& \hspace{0.5cm} =  \sum_{\allx^{t-1} } %
\prod_{t'=1}^{t-1} P_{\allX_{t'} | \allX^{t'-1}_{t'-l};\utheta}(\allx_{t'} | \allx^{t'-1}_{t'-l})   \left[      
				 \frac{\partial^2  }{ \partial \theta_{q_1}\partial \theta_{q_2} }  \sum_{\allx_t } 
				P_{\allX_t | \allX^{t-1}_{t-l};\utheta}(\allx_t | \allx^{t-1}_{t-l})
				 \right] \Bigg|_{\utheta=\utheta'} \label{eq:prf:HisI:b4} \\
&& \hspace{0.5cm} =  \sum_{\allx^{t-1}} %
\prod_{t'=1}^{t-1} P_{\allX_{t'} | \allX^{t'-1}_{t'-l};\utheta}(\allx_{t'} | \allx^{t'-1}_{t'-l})   \left[      
				 \frac{\partial^2  }{ \partial \theta_{q_1}\partial \theta_{q_2} }  1
				  \right] \Bigg|_{\utheta=\utheta'}  = 0 \label{eq:prf:HisI:b5} 
\eeqa
where \eqref{eq:prf:HisI:b1} uses the chain rule for $P_{\allX^t}$, \eqref{eq:prf:HisI:b3} cancels the $P_{\allX_t | \allX^{t-1}_{t-l};\utheta}(\allx_t | \allx^{t-1}_{t-l})$ terms and moves the summation inside, and \eqref{eq:prf:HisI:b4} uses linearity of differentiation.  Plugging \eqref{eq:prf:HisI:b5} back into \eqref{eq:prf:HisI:5} finishes the proof.
\end{proof}
\begin{remark} For non-finite alphabets, and an extra condition is necessary for Lemma~\ref{lem:est:par:HisI}, specifically for the analog of \eqref{eq:prf:HisI:b4}. Lemma~\ref{lem:est:par:HisI} is not necessary for Lemma~\ref{lem:asym_norm_theta}, but holds in this setting and simplifies the presentation.
\end{remark}

In the general case,
$\Sigma = \left[ \E[ A_t(\utheta^*)] \right]^{-1} \E\left[ G_t(\utheta^*) \right] \left[ \E[ A_t(\utheta^*)]\right]^{-1}. $
By Lemma~\ref{lem:est:par:HisI}, $\Sigma$ simplifies to \eqref{eq:def:param_cov_mat}.

In addition to Assumptions~\ref{def:spat_cond_ind},~\ref{assump:basic},~and~\ref{asmp:param}, the following are necessary to apply \cite{ling2010general} for asymptotic normality:  (i) $\E[ \sup_{\theta \in \Theta} [L_t(\utheta) ]] < \infty$ and   (ii) the vector $ [\frac{\partial L_t(\utheta) }{\partial \theta_{q} }\big|_{\utheta=\utheta^*}]_{1\leq q\leq Q}$ is a martingale difference in terms of $\allX^{t-1}$. 

\begin{lemma} For $P_{\allX}$ with finite alphabet $\calX$, under Assumptions~\ref{def:spat_cond_ind}~and~\ref{asmp:param}, (i) and (ii) hold.
\end{lemma}
\begin{proof} (i) Since $\calX$ is discrete, $P_{\allX_t | \allX^{t-1}_{t-l};\utheta}(\allx_t | \allx^{t-1}_{t-l})\leq1$ for all $\allx^t \in \calX^{mt}$ and $\theta\in\Theta$, so $L_t(\utheta) \leq 0$ and therefore $\E[ \sup_{\theta \in \Theta} [L_t(\utheta) ]] < \infty$.  

(ii) The vector $ [\frac{\partial L_t(\utheta) }{\partial \theta_{q} }\big|_{\utheta=\utheta^*}]_{1\leq q\leq Q}$ forms a martingale difference sequence if for all $1\leq q \leq Q$, (a) $ \E [\frac{\partial L_t(\utheta) }{\partial \theta_{q} } ]\big|_{\utheta=\utheta^*} < \infty $ and (b) $ \E [\frac{\partial L_t(\utheta) }{\partial \theta_{q} } | \allX_{t-1}]\big|_{\utheta=\utheta^*} = 0$, a.s. 

First consider (a),   \beqa
\E [\frac{\partial L_t(\utheta) }{\partial \theta_{q} } ]\big|_{\utheta=\utheta^*} &=& \E [\frac{\partial }{\partial \theta_{q} } \log P_{\allX_t | \allX^{t-1}_{t-l};\utheta}(\allX_t | \allX^{t-1}_{t-l}) ]\big|_{\utheta=\utheta^*} \\
&=& \E [ P_{\allX_t | \allX^{t-1}_{t-l};\utheta}(\allX_t | \allX^{t-1}_{t-l})^{-1} \frac{\partial }{\partial \theta_{q} } P_{\allX_t | \allX^{t-1}_{t-l};\utheta}(\allX_t | \allX^{t-1}_{t-l}) ]\big|_{\utheta=\utheta^*} \label{eq:prf:mrtgle:a2} \\
&=& 0 \label{eq:prf:mrtgle:a3}
\eeqa where \eqref{eq:prf:mrtgle:a2} applies the derivative and \eqref{eq:prf:mrtgle:a3} follows from \eqref{eq:prf:HisI:b1}--\eqref{eq:prf:HisI:b5} for the proof of Lemma~\ref{lem:est:par:HisI}.  For \eqref{eq:prf:HisI:b1}--\eqref{eq:prf:HisI:b5} the derivative is second order but the result holds for first order too.  Part (b) follows by the same proof.  
\end{proof}

\begin{remark}
In \cite{ling2010general}, the conditions for $ G_t$ and $A_t$ are:  (i) $\E[ G_t(\utheta^*)]$ is finite and positive definite, (ii) $\E[ A_t(\utheta^*)]$ is positive definite, and (iii) $\E [ \sup_{\utheta: \|\utheta - \utheta^* \|_2 < \eta}  \|A_t(\utheta)\|_2 ] < \infty$ for some $\eta>0$.  To simplify the presentation, Assumption~\ref{asmp:param} that $\E[ A_t(\utheta^*)]$ is also finite was included.  That ensured (iii) was satisfied and, by Lemma~\ref{lem:est:par:HisI}, (i).
\end{remark}

\section{Proof of Theorem~\ref{thm:par_est_smpl}} \label{app:prf:par_est_smpl}

\begin{proof}    We will upper and lower bound $\P(B_{\delta})$.  If the Jacobian matrix with $(r,q)$-th entry $\frac{\partial g_r}{\partial \theta_q}\big|_{\utheta = \utheta^* }$ is singular, take any maximal sized subset of $\{g_r(\utheta)\}_{r=1}^R$ that are linearly independent.  The convergence rate will still hold.  

Let $1_R$ denote a column vector of ones and $\I_R$ the $R$-dimensional identity matrix.  Let $Q$ and $\Lambda$ denote the orthonormal eigenvector and (diagonal) eigenvalue matrices of $\Sigma'$ respectively, so $\Sigma' Q = Q \Lambda$.  Let $A :=\Sigma'^{-1}$ denote the inverse so $A^{\frac{1}{2}} = Q^{-1} (\Lambda^{-1})^{\frac{1}{2}} Q$, since $A^{\frac{1}{2}} A^{\frac{1}{2}}  = A$.  Let $\vec{g}$ denote the column vector of directed information estimate errors $ \vec{g} := \left[(g_1(\wuthetan) - g_1(\utheta^*)), \dots,(g_{R}(\wuthetan) - g_{R}(\utheta^*))  \right]^\top.$

With this notation and letting ``$\leq$'' in the following to denote element-wise comparison for vectors, 
\beqa 
\P(B_{\delta})  &=& \P( - \delta 1_R \leq \vec{g} \leq \delta 1_R  ) \nonumber \\
&=& \P( - \delta \sqrt{n} A^{\frac{1}{2}} 1_R \leq  \sqrt{n} A^{\frac{1}{2}} \vec{g} \leq \delta \sqrt{n} A^{\frac{1}{2}} 1_R  ), \label{eq:grphcomp:1b} 
\eeqa where \eqref{eq:grphcomp:1b} multiplies through by $ \sqrt{n} A^{\frac{1}{2}}$.  By Theorem~5.4.2 of \cite{lehmann1999elements}, $\sqrt{n}A^{\frac{1}{2}} \vec{g} \sim \mathcal{N}(0,\I_R).$  This is analogous to normalizing one-dimensional Gaussian variables.  The interval limits in \eqref{eq:grphcomp:1b} imply the rate is $\delta = \mathcal{O}(n^{-1/2})$.

Next let $\lambda_\mathrm{min}$ denote the minimum eigenvalue of $\Sigma'$ (and thus in $\Lambda$).   Let $\Sigma_\mathrm{min} := \lambda_\mathrm{min} \I_R$ be a diagonal covariance matrix.  Similar to the above, let $A_\mathrm{min} := \Sigma_\mathrm{min}^{-1}$ so $A_\mathrm{min}^{\frac{1}{2}} = \frac{1}{\sqrt{\lambda_\mathrm{min}}} \I_R$.  Let $\vec{h}$ be an $R$-dimensional multivariate normal vector such that $ \sqrt{n} \vec{h} \sim \mathcal{N}(0,\Sigma_\mathrm{min})$.  Then 
\beqa
\P( - \delta 1_R \leq \vec{h} \leq \delta 1_R  ) 
&=& \P( - \delta \sqrt{n}A_\mathrm{min}^{\frac{1}{2}} 1_R \leq \sqrt{n} A_\mathrm{min}^{\frac{1}{2}} \vec{h} \leq \delta \sqrt{n} A_\mathrm{min}^{\frac{1}{2}} 1_R  ) \label{eq:grphcomp:3a} \\
&=& \P( - \frac{\sqrt{n}\delta 1_R}{\sqrt{ \lambda_\mathrm{min}}} \leq \frac{\sqrt{n} h_1 }{\sqrt{ \lambda_\mathrm{min}}}  \leq \frac{\sqrt{n}\delta 1_R }{\sqrt{ \lambda_\mathrm{min}}}  )^R \label{eq:grphcomp:3b} \\
&=& \left[ \mathrm{erf}\left( \delta\frac{\sqrt{n}}{\sqrt{2 \lambda_\mathrm{min}}} \right) \right]^{R}, \label{eq:grphcomp:3c} 
\eeqa where \eqref{eq:grphcomp:3a} multiplies by normalization factors, \eqref{eq:grphcomp:3b} uses that the normalized elements of $\vec{h}$ are i.i.d., and \eqref{eq:grphcomp:3c} uses the ``error'' function $\mathrm{erf}(x) = \frac{2}{\sqrt{\pi}} \int_0^x e^{-t^2} \, \mathrm{d}t$.

Now compare the volumes in \eqref{eq:grphcomp:1b} and \eqref{eq:grphcomp:3a}.  The respective probabilities integrate the likelihood of standard multivariate normal random vectors over the (rectangular) volume.  Note that $A^{\frac{1}{2}} 1_R$ specifies the corner in the positive orthant for the volume in \eqref{eq:grphcomp:1b}.  
\beqa 
A^{\frac{1}{2}} 1_R &=&  Q^{-1} (\Lambda^{-1})^{\frac{1}{2}} Q 1_R \label{eq:grphcomp:5a} \\
&\leq& \frac{1}{\sqrt{\lambda_\mathrm{min}}} 1_R \label{eq:grphcomp:5b} \\
&=& \frac{1}{\sqrt{\lambda_\mathrm{min}}} \I_R 1_R= A_\mathrm{min}^{\frac{1}{2}}  1_R \label{eq:grphcomp:5c}
\eeqa where \eqref{eq:grphcomp:5a} is by construction, \eqref{eq:grphcomp:5b} uses that $Q$ and $Q^{-1}$ are orthonormal and that the largest eigenvalue of $A^{\frac{1}{2}}$ is $1/\sqrt{\lambda_\mathrm{min}}$, and \eqref{eq:grphcomp:5c} is by construction.  Thus the volume in \eqref{eq:grphcomp:3a} contains that in \eqref{eq:grphcomp:1b}, so
\beqa
\P(B_{\delta})  &=& \P( - \delta \sqrt{n} A^{\frac{1}{2}} 1_R \leq  \sqrt{n} A^{\frac{1}{2}} \vec{g} \leq \delta \sqrt{n} A^{\frac{1}{2}} 1_R  ) \nonumber \\
&\leq& \P( - \delta \sqrt{n}A_\mathrm{min}^{\frac{1}{2}} 1_R \leq \sqrt{n} A_\mathrm{min}^{\frac{1}{2}} \vec{h} \leq \delta \sqrt{n} A_\mathrm{min}^{\frac{1}{2}} 1_R  ) \label{eq:grphcomp:4b}.
\eeqa

Combining \eqref{eq:grphcomp:3c} and \eqref{eq:grphcomp:4b}, and using the first two terms of the asymptotic expansion of $\mathrm{erf}(x)$, %
\beqa
 \P(B_{\delta})  \leq \left[ \mathrm{erf}\left( \delta\frac{\sqrt{n}}{\sqrt{2 \lambda_\mathrm{min}}} \right) \right]^{R} \!\!\!\!
&\approx& \!\!\!\! \left[1 - \frac{c_1}{\sqrt{n}} e^{-c_2 n}   \right]^{R} \nonumber \\
\!\!\!\! &\approx&\!\!\!\!  1 - \frac{m^2 c_1}{\sqrt{n}} e^{-c_2 n} \label{eq:grphcomp:2}\\
\!\!\!\! &=&\!\!\!\!  1 - c_1e^{ 2 \log(m) -c_2 n - \frac{1}{2} \log n},\label{eq:grphcomp:3}
\eeqa for appropriate constants $c_1$ and $c_2$.  Eq.~\eqref{eq:grphcomp:2} uses the first two terms in the binomial expansion and that $R=m(m-1)$, and \eqref{eq:grphcomp:3} moves the coefficients to the exponent.  

Repeat the above steps using the maximum eigenvalue of $\Sigma'$, $\lambda_\mathrm{max}$, with appropriate $A_\mathrm{max}$ to lower bound $\P(B_{\delta})$.  Eq.~\eqref{eq:grphcomp:3} will have the same form.  These bounds together  imply that for fixed $\delta$, $n$ must grow as $\log(m)$.
\end{proof}

\section{Proof of Theorem~\ref{thm:robustbnddegree}} \label{app:prf:robustbnddegree}

\begin{proof} \label{prf:robustbnddegree}
We first show that the parent sets of $\Phatrob$ can be identified independently.
\beqa
 \min_{\Phat_{\allX} \in \calPhat_K} \hspace{0.1cm} \max_{s \in \calS}  \hspace{0.1cm}  \left[ W(\Phat^*_{\allX}(s),s) - W(\Phat_{\allX},s) \right] 
&=& \min_{\Phat_{\allX} \in \calPhat_K} \hspace{0.1cm} \max_{s \in \calS} \left[ \sum_{i =1}^m \widehat{\I}_s(\allX_{\Ahat^*(i)} \to \X_{i}) - \widehat{\I}_s(\allX_{\Ahat(i)} \to \X_{i})  \right] \label{eq:prf:rob:1} \\
&\leq& \min_{\Phat_{\allX} \in \calPhat_K} \hspace{0.1cm}\sum_{i =1}^m  \max_{s \in \calS} \left[  \widehat{\I}_s(\allX_{\Ahat^*(i)} \to \X_{i}) -  \widehat{\I}_s(\allX_{\Ahat(i)} \to \X_{i})  \right] \label{eq:prf:rob:2} \\
&=& \sum_{i =1}^m  \min_{\Ahat(i) } \hspace{0.1cm} \max_{s \in \calS} \left[  \widehat{\I}_s(\allX_{\Ahat^*(i)} \to \X_{i}) -  \widehat{\I}_s(\allX_{\Ahat(i)} \to \X_{i})  \right] \label{eq:prf:rob:3a}
\eeqa
where $\{\Ahat^*(i)\}_{i=1}^m$ in \eqref{eq:prf:rob:1} are the parent sets in $\Phat^*_{\allX}(s)$ for the maximizing $s\in \calS$, \eqref{eq:prf:rob:2} brings the $\max$ inside, and \eqref{eq:prf:rob:3a} uses Theorem~\ref{thm:apx:gen_K_apx} that for any particular scenario $s \in \calS$, parent sets can be found independently.

If \eqref{eq:prf:rob:2} holds with equality, then the parent sets of $\Phatrob$ can be identified independently.  The first ${m-1 \choose K}$ coordinates of $\calS$ correspond to estimates $\{ \widehat{\I}_s(\allX_{\Ahat(1)} \to \X_{1}) \}$ for the ${m-1 \choose K}$ choices of $\Ahat(1)$.  The next ${m-1 \choose K}$ coordinates correspond to estimates $\{ \widehat{\I}_s(\allX_{\Ahat(2)} \to \X_{2}) \}$, etc.  Thus, for a given $i\in[m]$, the maximization in \eqref{eq:prf:rob:2} is only over the $i$th set of  ${m-1 \choose K}$ coordinates.  Since $\calS$ is rectangular, the values of the other coordinates are irrelevant.  Since each of the $m$ terms in the sum in \eqref{eq:prf:rob:2} are optimizing over disjoint sets of coordinates, \eqref{eq:prf:rob:2} holds with equality.

We next show that Algorithm~5 returns the individually most robust parent sets. Consider identifying robust parents for $\X_i$.  Using the notation $B_j$ in Algorithm~5, the worst case regret for parent set $B_j$ is
\beqa
\!\!\!\!\!\!\!\!\!\!\!R(B_j) \!\!\!\!&:=&\!\!\!\! \max_{s\in \calS}  \max_{j' \neq j} \left\{\!0, \widehat{\I}_s(\allX_{B_{j'}} \!\!\! \to\! \X_{i}) - \widehat{\I}_s(\allX_{B_j} \!\!\! \to\! \X_{i}) \! \right\}  \label{eq:prf:rob:5}\\
\!\!\!\!\!\!\!\!\!\!\!&\leq&\!\!\!\! \max \left\{\! 0, \left[ \max_{s\in \calS} \max_{j\neq j'} \widehat{\I}_s(\allX_{B_{j'}} \to \X_{i}) \right] -\left[ \min_{s\in \calS} \widehat{\I}_s(\allX_{B_j} \to \X_{i})\right] \!\right\} \label{eq:prf:rob:3} \\
\!\!\!\!\!\!\!\!\!\!\!&=&\!\!\!\!  \max\{0, \max_{j'\neq j} H(j') - L(B_j)\}\label{eq:prf:rob:4}.
\eeqa  The zero in \eqref{eq:prf:rob:5} is for the case that there is a $j$ such that $L(B_j) > H(j')$ for all $j' \neq j$.  Eq. \eqref{eq:prf:rob:3} applies the $\max$ to individual terms  and \eqref{eq:prf:rob:4} follows from lines~6~and~8 in Algorithm~5. Since $\calS$ is rectangular, for any set of values $\{\widehat{\I}_s(\allX_{B_j} \to \X_{i}) \in \calIhat(\allX_{B_j} \to \X_{i})\}_{j = 1}^{{m-1 \choose K}}$, there exist $s\in \calS$ with those values.  Thus, \eqref{eq:prf:rob:3} holds with equality.

Note that by lines~9~and~11 in Algorithm~5, 
\begin{equation}\label{eq:prf:rbst:a1}
R(B_j) = 
\begin{cases}
H(B_{j_1}) - L(B_j) & \text{if } j\neq j_1, \\
\max \{0, H(B_{j_3}) - L(B_{j_1})\} & \text{if } j = j_1.\\
\end{cases} 
\end{equation}
If $j_1 = j_2$, then for all $j \neq j_1$, $H(B_{j_1}) > H(B_j)$ and $L(B_{j_1}) \geq L(B_j)$, so by \eqref{eq:prf:rbst:a1},
\beqas
R(B_{j_1}) &=& H(B_{j_3}) - L(B_{j_1}) \\
&\leq& H(B_{j_1}) - L(B_j) 
=R(B_j).
\eeqas Thus, if $j_1=j_2$, then $B_{j_1}$ is the most robust parent set. Next consider the case that $j_1 \neq j_2$.  By \eqref{eq:prf:rbst:a1},
\beqas
\min_{j \neq j_1} R(B_j) &=& \min_{j \neq j_1} \left[ H(B_{j_1}) - L(B_j) \right] \\
&=&  H(B_{j_1}) - \max_{j \neq j_1}L(B_j) \\
&=&  H(B_{j_1}) - L(B_{j_2}) 
= R(B_{j_2}).
\eeqas  Thus, if $j_1 \neq j_2$, then either $B_{j_1}$ or $B_{j_2}$ would be most robust.
The parent set $B_{j_1}$ is selected if $j_1 = j_2$ or if
\beqa 
R(B_{j_1}) &\leq& R(B_{j_2})  \nonumber \\
 H(B_{j_3}) - L(B_{j_1}) &\leq& H(B_{j_1}) - L(B_{j_2}) \label{eq:prf:robind:8} \\
  H(B_{j_3}) + L(B_{j_2}) &\leq& H(B_{j_1}) + L(B_{j_1})  \label{eq:prf:robind:7}\\
    \frac{1}{2} (H(B_{j_3}) + L(B_{j_2})) &\leq& M(B_{j_1}). \label{eq:prf:robind:9}
\eeqa Eq.~\eqref{eq:prf:robind:8} uses \eqref{eq:prf:rbst:a1}, \eqref{eq:prf:robind:7} adds $L(B_{j_1})+ L(B_{j_2})$ to both sides, and \eqref{eq:prf:robind:9} uses  $M(B_j) = \frac{1}{2}(H(B_j) + L(B_j))$. This result shows that for each node $\X_i$, Algorithm~5 will return the individually most robust parents and thus $\Phatrob$.  
 \end{proof}

\section*{Acknowledgment}
The authors would like to thank a reviewer for bringing to their attention an early draft of \cite{eichler2012graphical}.

\bibliographystyle{IEEEtran}
\bibliography{DI_journalrefs}
\end{document}